\documentclass[12pt,authoryear,review]{elsarticle}
\makeatletter

  \def\ps@pprintTitle{%
 \let\@oddhead\@empty
 \let\@evenhead\@empty
 \def\@oddfoot{\centerline{\thepage}}%
 \let\@evenfoot\@oddfoot}
\makeatother
\usepackage{amsmath,amssymb,amsfonts}
\usepackage{graphicx}
\usepackage{subfigure}
\usepackage{dsfont}
\usepackage{color}
\usepackage[rightcaption]{sidecap}
\usepackage{epstopdf}
\usepackage{todonotes}

\newtheorem{theorem}{Theorem}
\newtheorem{lemma}[theorem]{Lemma}
\newtheorem{proposition}[theorem]{Proposition}

\newtheorem{assumption}[theorem]{Assumption}
\newenvironment{proof}[1][Proof]{\begin{trivlist}
\item[\hskip \labelsep {\bfseries #1}]}{\end{trivlist}}

\newcommand{\mcA}{{\mathcal{A}}}
\newcommand{\mcF}{{\mathcal{F}}}

\journal{TBA}

\begin{document}
\begin{frontmatter}

\title{Hedging Non-Tradable Risks with Transaction Costs and Price Impact
\tnoteref{t1}}
\tnotetext[t1]{SJ would like to acknowledge the support of the Natural Sciences and Engineering Research Council of Canada (NSERC), [funding reference numbers RGPIN-2018-05705 and RGPAS-2018-522715]. \\ The authors would like to thank participants at the Research in Options Conference, the SIAM Annual General Meeting, the Bachelier World Congress, the INFORMS Annual Meeting, the Western Conference on Mathematical Finance, and the SIAM Financial Mathematics and Engineering Conference for comments on this article.}
\tnotetext[t2]{The data that support the findings of this study are available from the corresponding author upon reasonable request.}

\author[author1]{\'Alvaro Cartea}\ead{alvaro.cartea@maths.ox.ac.uk}
\author[author2]{Ryan Donnelly}\ead{ryan.f.donnelly@kcl.ac.uk}
\author[author3]{Sebastian Jaimungal}\ead{sebastian.jaimungal@utoronto.ca}
\address[author1] {Mathematical Institute, University of Oxford, \\ Oxford-Man Institute of Quantitative Finance, Oxford,  United Kingdom }
\address[author2] {King's College London, United Kingdom}
\address[author3] {University of Toronto,   Toronto,  Canada }

\date{}

\begin{abstract}
	A risk-averse agent hedges her exposure to a non-tradable risk factor $U$ using a correlated traded asset $S$ and accounts for the impact of her trades on both factors.
	The effect of the agent's trades on $U$ is referred to as cross-impact. By solving the agent's stochastic control problem, we obtain a closed-form expression for the optimal strategy when the agent holds a linear position in $U$.  When the exposure to the non-tradable risk factor $\psi(U_T)$ is non-linear, we provide an approximation to the optimal strategy in closed-form, and prove that the value function is correctly approximated by this strategy when cross-impact and risk-aversion are small.
	We further prove that when $\psi(U_T)$   is non-linear, the approximate optimal strategy can be written in terms of the optimal strategy for a linear exposure with the size of the position changing dynamically according to the exposure's ``Delta'' under a particular probability measure.

\end{abstract}
\begin{keyword}
	non-tradable risk, hedging, algorithmic trading, price impact
\end{keyword}
\end{frontmatter}

\section{Introduction}

In this paper we show, for the first time, how a risk-averse agent manages her exposure to a non-tradable risk factor while taking into account trading price impact. The agent can trade in a correlated asset to hedge her exposure. Ideally, this position in the traded asset is achieved immediately, however price impact restricts the speed at which the agent can trade. On the other hand, trading too slowly exposes the agent to the risk associated with the non-tradable risk factor.

Price impact can generally be classified by the timescale of its persistence into two types: temporary or permanent. The first occurs when the volume of the trade exceeds the available liquidity at the best quote in the limit order book (LOB). The second occurs due to updates of limit orders to reflect the arrival of new information conveyed by the liquidity taking order. Some studies of price impact effects include \cite {potters2003more}, \cite{cont2014price}, \cite{donier2015fully}, and \cite{LehalleMML15}. Our problem is related to two strands of literature, one is the optimal execution of large positions, and the other is the hedging of non-tradable risks. The execution of large positions with price impact has been studied extensively in the literature, see the early work of \cite{almgren2001optimal}, and more recently \cite{GueantAMF}, \cite{cartea2015algorithmic}, \cite{bechler2015optimal}, and \cite{gueant2016financial}.

We provide three examples where investors are exposed to a non-tradable risk factor. The first two  are when the non-tradable factor is a financial instrument which the agent is restricted from trading for legal or regulatory reasons. The third is when the non-tradable factor is not a financial instrument.  (i) Employees who are given compensation in the form of options written on the stock of their firm may be bound to a covenant that precludes them from trading the options or the firm's stock for a period of time.  (ii) A regulatory body imposes a short selling ban on stocks as were the cases in 2008\footnote{https://www.ft.com/content/16102460-85a0-11dd-a1ac-0000779fd18c .} and 2011\footnote{https://www.ft.com/content/9a55839a-c42d-11e0-ad9a-00144feabdc0 .}. Investors holding derivatives written on assets with short-sell bans would have to seek out unrestricted and correlated assets to hedge their exposures. (iii) Weather derivatives may be hedged by taking positions in traded stocks of firms whose financial performance is  correlated to weather. Henceforth, we interpret the non-tradable factor as an asset which the agent is precluded from trading and the agent either holds shares of this asset or a European-style contingent claim on the non-tradable risk factor.

We solve in closed-form for the agent's value function (Proposition \ref{prop:value_function}) and optimal trading strategy (Theorem \ref{prop:optimal_control}) when she holds units of the non-tradable risk factor. When the exposure is in the form of a European-style contingent claim we approximate her value function, prove that the approximation is indeed valid (Theorem \ref{prop:asymptotic_approximation}), obtain an approximate trading strategy (Theorem \ref{prop:approx_nu}),  and show that risk-aversion and cross-price impact may influence the direction of her optimal trades in opposing directions. The effect of risk-aversion is to acquire a position that offsets the risk of the non-tradable asset, which is possible through correlation effects (a short/long position will be taken when correlation is positive/negative). Cross-price impact, on the other hand, gives incentive to the agent to take a long/short position in the traded asset when correlation is positive/negative, as this has a beneficial effect on the value of the non-tradable asset. We further demonstrate that, when the agent has a non-linear exposure to the non-tradable asset, the approximate strategy  can be written in terms of the optimal strategy for a linear exposure, where the number of units of the linear exposure is given by the ``Delta'' of the non-linear exposure (Proposition \ref{prop:closed_form}).

Our work is related to the literature on incomplete markets where an agent is exposed to sources of risk that cannot be fully diversified. A closely related stream of research is initiated by \cite{henderson2002valuation} who studies the valuation of claims on non-tradable assets, using utility indifference. Similarly, \cite{henderson2007valuing} and \cite{grasselli2011getting} study the valuation of irreversible investments on non-tradable assets, while \cite{CarteaJaiIrreversible} show how to account for model uncertainty. Along similar lines, \cite{leung2009accounting}, \cite{leung2009exponential}, and \cite{grasselli2009risk} study the valuation of employee stock options by trading partially correlated assets. In \cite{leung2016optimal} the authors consider the problem of statically hedging a contingent claim written on a correlated asset. However, none of these works account for the price impact of the agent's trades on the traded assets themselves nor on the non-tradable risk exposure.

The remainder of the paper is organized as follows. In Section \ref{sec:model} we introduce the dynamics of the traded asset and non-tradable risk factor, and present the dynamic optimization problem  faced by the agent. In Section \ref{sec:linear_exposure} we solve in closed-form for the agent's value function and optimal trading strategy when the exposure to the non-tradable risk factor affects her terminal wealth with linear dependence. In Section \ref{sec:non_linear_exposure} we consider the situation when the exposure to the non-tradable risk factor is non-linear. We derive an approximate optimal trading strategy, in closed-form, which we prove is valid when the cross-impact and risk-aversion are small. We also investigate the qualitative behaviour of the strategy through simulations of the underlying dynamics. Section \ref{sec:conclusion} concludes and longer proofs appear in the appendix.

\section{Model}\label{sec:model}

\subsection{Dynamics}

In this section we outline the dynamics of multiple assets that include the price impact effects of the agent's trading, as well as the dynamics of the agent's inventory and cash holdings. We denote by $S^\nu=(S_t^\nu)_{t\in[0,T]}$ the (controlled) midprice process of a traded asset, and by $U^\nu=(U_t^\nu)_{t\in[0,T]}$ the (controlled) value process of a non-tradable risk factor. We assume the agent is able to directly trade $S$ and she has additional exposure to $U$ such that her wealth increases by $\psi(U)$ at some future time $T$, where $\psi:\mathbb R\rightarrow \mathbb R$. Although she is unable to directly trade in $U$, trades that occur in $S$  have an effect on the value of the non-tradable risk factor. We also let $Q^\nu=(Q_t^\nu)_{t\in[0,T]}$ denote the (controlled) inventory process in the traded asset held by the agent, and let the control $\nu=(\nu_t)_{t\in[0,T]}$ denote the rate at which this asset is acquired (a positive/negative value indicates she is buying/selling the asset). The dynamics of the controlled inventory are
\begin{align}
	Q^\nu_t &= q+\int_0^t \nu_u\,du\,.
\end{align}
The traded asset price and non-tradable risk factor satisfy the SDEs
\begin{align}
	dS^\nu_t &= (\mu + b\,\nu_t) \,dt + \sigma\, dW_t\,,\\
	dU^\nu_t &= (\beta + c\,\nu_t)\,dt + \eta\, dZ_t\,.
	%\\
	%d[W,Z]_t &= \rho\, dt\,.
\end{align}
Here, $(W_t)_{t\in[0,T]}$ and $(Z_t)_{t\in[0,T]}$ are standard Brownian motions with correlation $\rho\in(-1,1)$. The term $b\,\nu_t$, with $b\ge0$ constant, represents a permanent price impact due to the agent's trading. We include the possibility that trading in asset $S$ has a cross-price impact on $U$. This is accounted for by the inclusion of the term $c\,\nu_t$, with $c$ constant, in the drift of $U$ and $\beta$ is a constant.

In addition to the permanent impact on midprices, we model a temporary price impact by introducing an execution price, which we denote by $\hat S^\nu=(\hat S^\nu_t)_{t\in[0,T]}$ and is given by
\begin{align}
	\hat{S}^\nu_t &= S^\nu_t + k\,\nu_t\,.
\end{align}
The execution price is the value  the agent pays to acquire shares of the traded asset. Trading at a faster rate induces an execution price which is farther away from the midprice, in addition to affecting the drift of the asset. The temporary price impact can be considered a result of limit order book microstructure.   The permanent price impact can be thought of, among other effects, as the result of information leakage which induces other market participants to modify existing orders. For further discussions on temporary and permanent impact arising from LOB dynamics see \cite{cartea2015algorithmic}.

As the agent executes trades in the asset $S$, she must withdraw or deposit appropriate funds from her cash holdings, which have value denoted by $X^\nu=(X_t^\nu)_{t\in[0,T]}$ and equals
\begin{align}
	X_t^\nu = x-\int_0^t \hat{S}^\nu_u \,\nu_u \,du\,.
\end{align}
Throughout, we work on the completed and filtered probability space $(\Omega, \mathbb{P}, \{\mcF_t \}_{t\in[0,T]})$ where $\mcF_t$ is the standard augmentation of the natural filtration generated by $(W_u,Z_u)_{u\in[0,t]}$.

\subsection{Performance Criterion}

The agent employs an exponential utility function with risk aversion parameter $\gamma$ and aims to maximize her expected utility of wealth at time $T$. At time $T$ the exposure to the non-tradable risk factor directly affects the agent's wealth. At this time her wealth consists of her cash holdings, the value included in her inventory holdings of the traded asset, and the exposure $\psi(U_T)$. If she acts according to a trading strategy $\nu\in\mcA$, where the set of admissible trading strategies $\mcA$ consists of $\mcF$-predictable processes such that $\mathbb{E}[\int_0^T\nu^2_udu] < \infty$, her performance criterion is given by
\begin{align}
	H^\nu(t,x,q,S,U) &= \mathbb{E}_{t,x,q,S,U}\biggl[-e^{-\gamma\,\big(X^\nu_T + Q^\nu_T\,\left(S^\nu_T - \alpha\, Q^\nu_T\right) + \psi(U_T^\nu)\big)}\biggr]\,,\label{eqn:performance_criteria}
\end{align}
where $\mathbb{E}_{t,x,q,S,U}[\,\cdot\,]$ represents expectation conditional on $X_t^\nu = x$, $Q_t^\nu = q$, $S_t^\nu = S$, and $U_t^\nu = U$. The term $\alpha\,(Q^\nu_T)^2$ represents a price penalty that the agent incurs  from having to liquidate her inventory at time $T$ and incentivizes her to hold a small inventory position near maturity.

{ In general the liquidation of terminal inventory $Q_T^\nu$ may have a cross-impact effect on the non-traded risk factor $U_T^\nu$, and so $\psi$ should depend on both $U_T$ and $Q_T$. This would complicate the analysis and in reality there are many situations in which this cross-impact would not be realized. For example, if the exposure $\psi(U_T)$ is cash settled and the liquidation of $Q_T^\nu$ shares is conducted immediately after the settlement, then any cross impact effect of this trade would be irrelevant because the agent does not physically hold exposure to $U_T$ anymore.}

Her value function is
\begin{align}
	H(t,x,q,S,U) &= \sup_{\nu\in\mcA}H^\nu(t,x,q,S,U)\,.\label{eqn:control_problem}
\end{align}

The control problem posed in \eqref{eqn:control_problem} has the associated Hamilton-Jacobi-Bellman (HJB) equation:
\begin{equation}
	%\left\{
	\begin{split}
		\partial_t H  + \sup_\nu\biggl\{
		\nu\,\partial_q H -(S+k\,\nu)\,\nu\,\partial_x H + (\mu + b\,\nu)\,\partial_S H
		\qquad\qquad\qquad\quad & \\
		 + \tfrac{1}{2}\,\sigma^2\,\partial_{SS}H + (\beta + c\,\nu)\,\partial_U H + \tfrac{1}{2}\,\eta^2\,\partial_{UU}H + \rho\,\sigma\,\eta\,\partial_{SU}H\biggr\} =&\, 0\,,
		\\
		H(T,x,q,S,U) =& -e^{-\gamma\,(x+q\,(S-\alpha\, q) + \psi(U))}\,.
	\end{split}
	%\right.
	\label{eqn:HJB}
\end{equation}
%subject to terminal condition $H(T,x,q,S,U) = -e^{-\gamma(x+q(S-\alpha q) + \psi(U))}$.
In the next section we  assume that the exposure $\psi(U)$ is linear in the value of the non-tradable risk factor. This allows us to solve for the value function and the optimal trading strategy in closed-form. In Section \ref{sec:non_linear_exposure} we relax this assumption and provide solutions which are correct up to corrections that vanish in the limit of small risk-aversion and cross-impact.

\section{Linear Exposure}\label{sec:linear_exposure}

We consider the special case where the exposure to the non-tradable risk factor is linear: $\psi(U) = \mathfrak{N}\,U$. A direct interpretation of this exposure is that the agent holds $\mathfrak{N}$ shares of the non-tradable risk factor and is restricted from trading it $\forall\,t\in[0,T)$, but at time $T$ this restriction is lifted and the shares are immediately liquidated.\footnote{In principle we could include a liquidation penalty as we do for the traded asset, however as the agent has no control over the number of shares of $U$ that she holds during $[0,T]$, this penalty would factor out of the performance criteria as a constant and would have no effect on the optimal trading strategy.}

We also assume $2\,\alpha - b > 0$ because this ensures the form of the value function given in the following proposition is well defined for all $t\in[0,T]$. If this inequality is not obeyed, then it is possible for the terms in \eqref{eqn:h_1} and \eqref{eqn:h_2}, which are shown below,  to explode. This inequality is typically satisfied in practical examples, because the reverse inequality induces the agent to buy (or sell) very large quantities and destabilize prices, then liquidate this large position with a smaller penalty than the gain incurred by the original price movements.

{The following proposition and theorem are particular cases of well-known results in linear-quadratic-exponential Gaussian control (for example, see \cite{jacobson1973optimal} and \cite{duncan2013linear}). We include them for completeness and because the expression in equation \eqref{eqn:optimal_control} in Theorem \ref{prop:optimal_control} plays an important role in one of our subsequent results.}

\begin{proposition}[Linear Exposure Value Function]\label{prop:value_function} With $\psi(U) = \mathfrak{N}\,U$ the solution to equation \eqref{eqn:HJB} together with its terminal condition is given by
	\begin{align}
		H(t,x,q,S,U) &= -\exp\left\{-\gamma\,(x + q\,S + \mathfrak{N}\,U + h(t,q;\mathfrak{N}))\right\}\,,\label{eqn:propH}
	\end{align}
	where
	\begin{align}
		h(t,q;\mathfrak{N}) = h_0(t;\mathfrak{N}) + h_1(t;\mathfrak{N})\,q + h_2(t)\,q^2\,.
	\end{align}
	The time-dependent functions $h_0,\,h_1,\, h_2$ are given by
	\begin{subequations}
		\addtolength{\jot}{3pt}
		\begin{align}
			h_0(t;\mathfrak{N}) =&\;\left(\beta\,\mathfrak{N} - \tfrac{1}{2}\,\gamma\,\eta^2\,\mathfrak{N}^2\right)\,(T-t)
			+ \tfrac{1}{4\,k}\,\int_t^T (h_1(s;\mathfrak{N}) + c\,\mathfrak{N})^2\,ds\,,
			\label{eqn:h_0}
			\\
			h_1(t;\mathfrak{N}) =&\; \frac{\zeta \,k}{\omega}
			\frac{\phi^-(1-e^{-\frac{\omega}{k}\,(T-t)}) - \phi^+\,(1-e^{\frac{\omega}{k}\,(T-t)})
			+
			\frac{2\,\omega^2\, c}{\zeta\,k}\,\mathfrak{N}}
			{\phi^- e^{-\frac{\omega}{k}\,(T-t)} + \phi^+\,e^{\frac{\omega}{k}\,(T-t)}}
			 - c\,\mathfrak{N}\,,
			 \label{eqn:h_1}
			\\
			h_2(t) =&\; \omega \,\frac{\phi^-\,e^{-\frac{\omega}{k}\,(T-t)} - \phi^+ \, e^{\frac{\omega}{k}\,(T-t)}}{\phi^- e^{-\frac{\omega}{k}\,(T-t)} + \phi^+\, e^{\frac{\omega}{k}\,(T-t)}} - \frac{b}{2}\,,\label{eqn:h_2}
		\end{align}
		\addtolength{\jot}{-3pt}
	\end{subequations}
	and the constants
	\begin{equation*}
		\zeta = \mu - \gamma\,\rho\,\sigma\,\eta\,\mathfrak{N}\,,\qquad \omega = \sqrt{\tfrac{k\,\gamma\,\sigma^2}{2}}\,,
		\qquad
		\text{and}
		\qquad
		\phi^\pm = \omega \,\pm \alpha \,\mp \tfrac{1}{2}b\,.
	\end{equation*}
\end{proposition}
\begin{proof}
	For a proof see the Appendix.
\end{proof}

\begin{theorem}[Optimal Trading Strategy: Linear Case]\label{prop:optimal_control} The optimal trading speed
	\begin{align}
		\nu^*_t &= \tfrac{1}{2\,k}\left(\,c\,\mathfrak{N} + h_1(t;\mathfrak{N}) + (2\,h_2(t) + b)\,Q_t^{\nu^*}\right)\,,\label{eqn:optimal_control}
	\end{align}
	is admissible, and the solution provided in \eqref{eqn:propH} is indeed the value function.
	Moreover, the optimal level of inventory is deterministic and is given by
	\begin{equation}
		Q^{\nu^*}_t =
		\biggl(\frac{\zeta\, k\,(\phi^- - \phi^+)}{4\,\omega^2} + \frac{c\, \mathfrak{N}}{2}\biggr)\frac{e^{\frac{\omega}{k} \,t} - e^{-\frac{\omega}{k} t}}{\ell(T)} - \frac{\zeta\, k}{2\,\omega^2}\biggl(\frac{\ell(T-t)}{\ell(T)} - 1\biggr) + Q_0\,\frac{\ell(T-t)}{\ell(T)}\label{eqn:optimal_inventory}
	\end{equation}
	with $\ell(t)=\phi^+\,e^{\frac{\omega}{k} \,t} + \phi^-\,e^{-\frac{\omega}{k}\, t}$, and $\zeta$, $\omega$, and $\phi^\pm$ as in Proposition \ref{prop:value_function}.
	%\seb{Shouldn't we state that this is indeed admissible, is the true optimal and the function in Prop 1 is the true value function?}
\end{theorem}
\begin{proof}
	For a proof see the Appendix.
\end{proof}

The optimal trading strategy in Theorem \ref{prop:optimal_control} shows how trading in asset $S$ is affected by the exposure to asset $U$. In the simplified case of no cross-impact ($c = 0$), the trading strategy is identical to the single asset case except with the drift modified to $\mu - \gamma\,\rho\,\sigma\,\eta\,\mathfrak{N}$. This modification represents the trade-off between a source of expected returns and a source of risk. Holding an inventory of $Q_t$ means that the agent's wealth is increasing at a rate of $\mu\, Q_t$, but at the same time there is a risk contribution of the form $\rho\,\sigma\,\eta\, \mathfrak{N}\, Q_t$ due to covariation between $S$ and $U$. This drift modification has an interesting consequence that is illustrated most clearly when $\mu = 0$ and $Q_0 = 0$. If the agent has no exposure to the non-tradable risk factor ($\mathfrak{N} = 0$), and if she does not speculate on the future value of the traded asset ($\mu = 0$), then she has no reason to acquire any shares and will optimally hold a zero position for the whole trading period. This becomes apparent in equation \eqref{eqn:optimal_inventory} when $\zeta = 0$. However, if she holds a linear position in $U$, then she takes a non-zero position in the traded asset due to her ability to partially hedge the risk in $U$. This qualitative difference in the trading strategy exemplifies the importance of considering the interaction between the traded and non-tradable risk factors.

Although it may not be apparent from the formulation of the problem or the explicit form of the equations which dictate the optimal strategy, there is a specific inventory level of the traded asset that the agent favors and attempts to hold if the trading period is long. Any deviation from this position is caused by the various forms of frictions and penalties that the agent has to pay. For example the temporary price impact incurs larger costs to the agent if she trades too quickly, and the terminal inventory penalty means the agent favors smaller inventory levels as the end of the trading period approaches.

To formulate our notion of the agent's desired long horizon position, we introduce a quantity we refer to as relative time. As $t\in[0,T]$, any instant in the trading period can be expressed in the form $t = \kappa \,T$ for $\kappa\in[0,1]$. We refer to $\kappa$ as the relative time.

\begin{proposition}[Long Horizon or Frictionless Position]\label{prop:long_horizon} Fix a relative time $\kappa\in(0,1)$ and let $t = \kappa \,T$. Then
	\begin{align}
		\lim_{T\rightarrow \infty} Q_{\kappa \,T}^{\nu^*} = \frac{\mu - \gamma\,\rho\,\sigma\,\eta\,\mathfrak{N}}{\gamma\,\sigma^2} = \lim_{k\rightarrow 0} Q_{\kappa \,T}^{\nu^*}\,.\label{eqn:long_horizon}
	\end{align}
\end{proposition}
\begin{proof}
For a proof see the  Appendix.
\end{proof}

We illustrate the optimal strategy numerically in Figure \ref{fig:inventory1}. As long as $T$ is sufficiently large, the first equality in \eqref{eqn:long_horizon} tells us that the agent desires to hold this inventory position for as long as possible. The exception is towards the beginning and end of the trading period. This behavior is reflected by the fact that we are required to exclude $\kappa = 0$ and $\kappa = 1$ in the proposition. The  agent favors this position because it maximizes the return versus risk over all possible inventory levels. For a fixed inventory level $q$ in the traded asset, the instantaneous expected return is $\mu \,q$. However, this position exposes the agent to an instantaneous level of risk of the form $\rho\,\sigma\,\eta\,\mathfrak{N}\,q + \frac{1}{2}\,\sigma^2 \,q^2$ (the agent is also exposed to an instantaneous risk of the form $\frac{1}{2}\,\eta^2\,\mathfrak{N}^2$ but the agent has no control over this quantity). Taking a difference of the return and risk scaled by $\gamma$ and then maximizing with respect to $q$ gives the same expression as in equation \eqref{eqn:long_horizon}. Thus, this is the optimal position in the traded asset which balances instantaneous risk and return. The second equality in \eqref{eqn:long_horizon} tells us that this is the optimal position that the agent would hold if frictionless trading were possible. The equivalence between these two limits is an indicator that the agent attempts to trade towards the frictionless optimal inventory level, but is only prevented from doing so due to the frictions involved with trading.

\begin{figure}
	\begin{center}
		{\includegraphics[trim=140 240 140 240, scale=0.48]{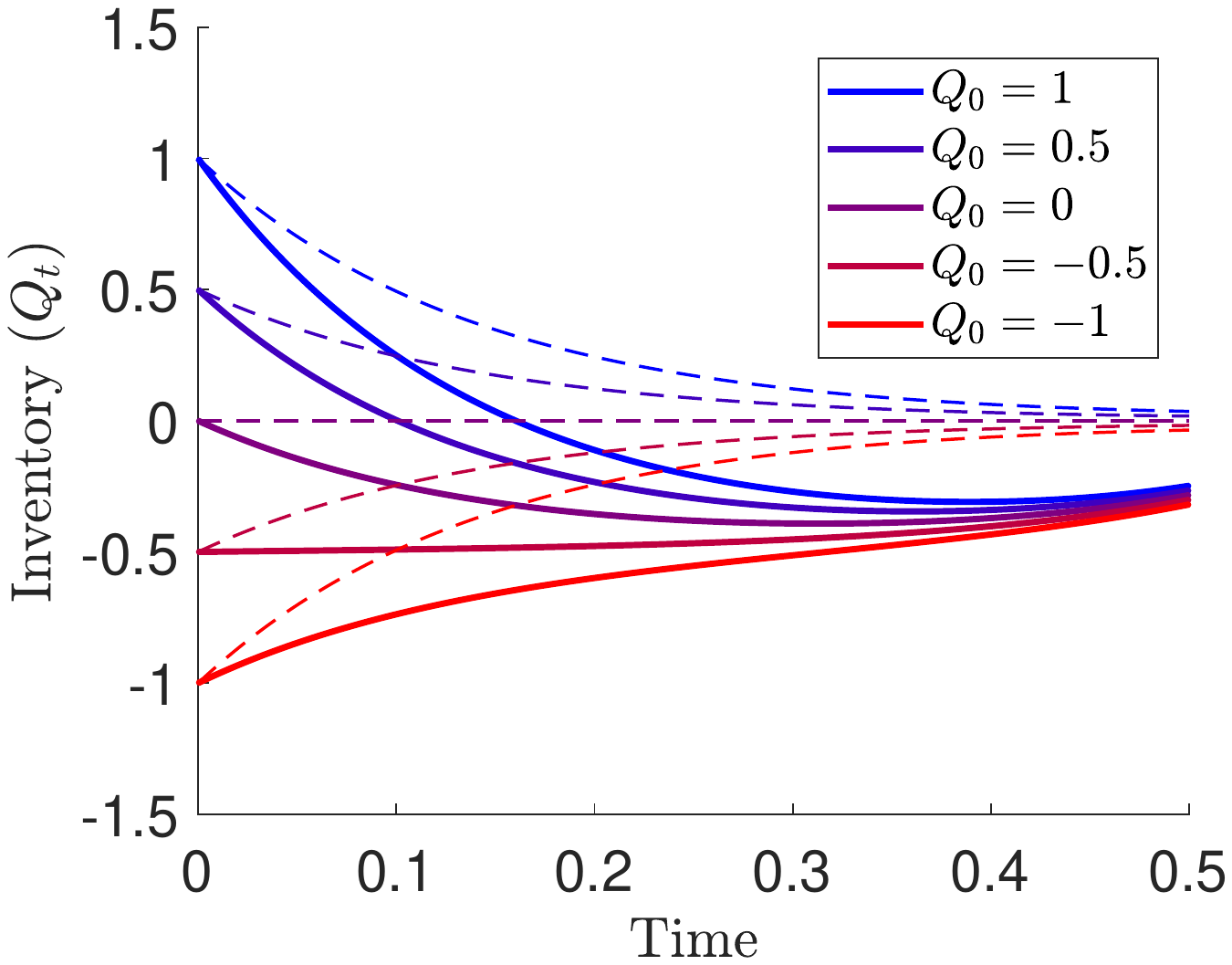}}\hspace{10mm}
		{\includegraphics[trim=140 240 140 240, scale=0.48]{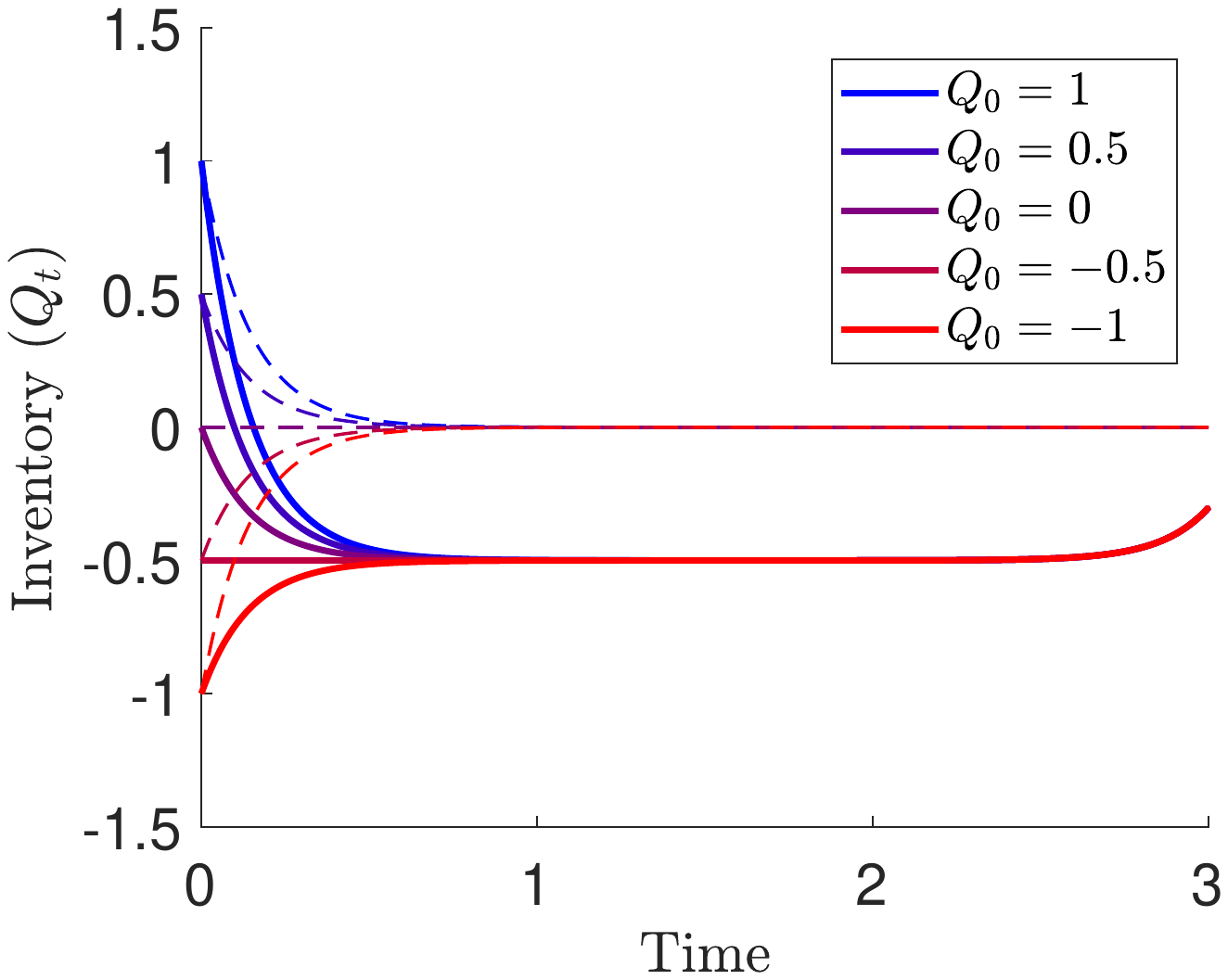}}
	\end{center}
	\vspace{-1em}
	\caption{Agent's optimal inventory position over time. In the left panel, the length of the trading period ends at $T = 0.5$ and on the right it ends at $T=3$. Other model parameters are $\mu = 0$, $\beta = 0$, $\sigma = 1$, $\eta = 1$, $\rho = 0.5$, $b = 10^{-2}$, $c = 10^{-3}$, $k=10^{-2}$, $\gamma = 1$, $\alpha = 0.05$. Solid curves are used when the agent is exposed to one share of the non-tradable risk factor ($\mathfrak{N} = 1$). Dashed curves represent the Almgren-Chriss strategy when there is no exposure to the non-tradable risk-factor ($\mathfrak{N} = 0$). The long horizon level of Proposition \ref{prop:long_horizon} which all solid curves approach in the right panel is given by $Q = -0.5$. \label{fig:inventory1}}
\end{figure}

\section{Non-Linear Exposure} \label{sec:non_linear_exposure}

In this section  the agent is exposed to the non-tradable risk factor in the form $\psi(U)$, which we may interpret  as holding a European-style contingent claim written on the non-tradable risk factor. The performance criterion, value function, and associated HJB equation are the same as \eqref{eqn:performance_criteria}, \eqref{eqn:control_problem}, and \eqref{eqn:HJB}, respectively. The non-linear payoff prevents us from disentangling the dependence between $U$ and the other variables, so we propose the ansatz
\begin{align}
	H_\psi(t,x,q,S,U;c,\gamma) &= -\exp\left\{-\gamma\,(x + q\,S + h_\psi(t,q,U;c,\gamma))\right\}\,.\label{eqn:ansatz}
\end{align}
We show explicit dependence of the functions $H_\psi$ and $h_\psi$ on $c$ and $\gamma$ because these two parameters are used in an expansion approximation, which we discuss below. We also make the dependence on $\psi$ explicit for clarity. Substituting this ansatz into \eqref{eqn:HJB} yields the following equation for $h_\psi$:

\begin{equation}
	%\left\{
	\begin{split}
		\partial_t h_\psi + \mu \,q
		- \tfrac{1}{2}\,\gamma\,\sigma^2\,q^2
		+ (\beta- \gamma\,\rho\,\sigma\,\eta \,q)\,\partial_Uh_\psi
		+ \tfrac{1}{2}\,\eta^2\,\partial_{UU}h_\psi
		\qquad &
		\\
		- \tfrac{1}{2}\,\gamma\,\eta^2\,(\partial_Uh_\psi)^2
		+ \sup_{\nu}\left\{\nu \partial_qh_\psi
		+ c\,\nu\,\partial_Uh_\psi + b\,q\,\nu - k\,\nu^2\right\} & = 0\,,
		\\
		h_\psi(T,q,U;c,\gamma) &= \psi(U) - \alpha \,q^2\,.
	\end{split}
	%\right.
	\label{eqn:h_psi}
\end{equation}

The supremum in the preceding equation, which provides us with the feedback form of the optimal strategy, is achieved at
\begin{align}
	\nu^*(t,q,U;c,\gamma) &= \tfrac{1}{2\,k}\left(\partial_qh_\psi + c\,\partial_Uh_\psi + b\,q\right)\,.\label{eqn:optimal_nu}
\end{align}
Substituting this value of $\nu$ into equation \eqref{eqn:h_psi} gives
\begin{align}
	\begin{split}
		\partial_t h_\psi + \mu\, q - \tfrac{1}{2}\,\gamma\,\sigma^2\,q^2 + (\beta- \gamma\,\rho\,\sigma\,\eta \,q)\,\partial_Uh_\psi + \tfrac{1}{2}\,\eta^2\,\partial_{UU}h_\psi
		&
		\\
		- \tfrac{1}{2}\,\gamma\,\eta^2(\partial_Uh_\psi)^2 + \tfrac{1}{4\,k}\left(\partial_qh_\psi + c\,\partial_Uh_\psi + b\,q\right)^2 & = 0\,.
	\end{split}
	\label{eqn:h_PDE}
\end{align}
It is easily checked that if $\psi(U) = \mathfrak{N}\,U$, then this equation along with its terminal condition are solved by $h_\psi(t,q,U;c,\gamma) = h(t,q) + \mathfrak{N}\,U$, which also gives $H_\psi = H$ (as in Proposition \ref{prop:value_function}) as expected. For general forms of the payoff $\psi$ we are not able to find closed-form expressions which solve equation \eqref{eqn:h_psi}, but if we consider small values of model parameters $c$ and $\gamma$ we can obtain solutions that are approximate in an asymptotic sense.

It is reasonable to suppose that the cross-price impact factor $c$ is smaller than both the temporary and permanent price impact factors. Indeed, the effect that trading in one stock has on the price of another stock should be significantly less than the effect that it has on its own price. For this reason, the parameter $c$ is one choice for which we may perform an asymptotic expansion. We also perform the expansion with respect to the risk-aversion parameter $\gamma$. To this end, we perform the expansion in each quantity simultaneously by introducing an expansion parameter $\theta$ and making the substitutions $c\mapsto\theta\,c$ and $\gamma\mapsto\theta\,\gamma$.

\begin{assumption}\label{ass:ass} We make the following technical assumptions to prove the validity of the expansion.
	\begin{enumerate}[i)]
		\item $\psi \in C^4(\mathbb{R})$ with all four derivatives bounded.
		\item $2\,\alpha - b > 0$.
		\item {
				Given initial states $x$, $q$, $S$, and $U$, there exist positive constants $\theta^*<1$, $\epsilon^*$, $C$, and $D$ that satisfy the following uniform boundedness condition: for every $\theta\in(0,\theta^*)$ and $\epsilon\in(0,\epsilon^*)$, if $\nu$ is an admissible control such that
				\begin{align}
					H^\nu(0,x,q,S,U;\theta\,c,\theta\,\gamma) + \epsilon \geq H_\psi(0,x,q,S,U;\theta\,c,\theta\,\gamma)\,,\nonumber
				\end{align}
				then
				\begin{align}
					\mathbb{E}\biggl[\int_0^Te^{D \,(|X^\nu_t| + |Q^\nu_tS^\nu_t| + |Q^\nu_t| + (Q^\nu_t)^2 + |U^\nu_t|)}\,dt\biggr] & \leq C\,. \label{eqn:uniform_bound_assumption}
				\end{align}
			}
	\end{enumerate}
\end{assumption}
Assumption \ref{ass:ass} i) eliminates the consideration of vanilla European option payoffs as they are not twice continuously differentiable (even in a weak sense the second derivative is not bounded). However this complication can be avoided by using a regularized version of the payoffs, e.g., by assuming that the option with maturity $T$ expires at  time $T+\delta t$, for $\delta t$ arbitrarily small. This condition ensures that many of the terms in the expansion below have bounded derivatives with respect to $U$ and allows us to make certain growth estimates more easily.

Assumption \ref{ass:ass} ii) is made for the same reason as the case of the linear payoff. It ensures that the terms in the expansion are well defined for all $t\in[0,T]$.

{
	Finally, Assumption \ref{ass:ass} iii) can be interpreted as a condition on boundedness/continuity with respect to the space of admissible controls. It states that a particular exponential moment is uniformly bounded over a set of controls sufficiently close to optimal. In proving the validity of our approximation, this inequality allows us to bound the magnitude of the error, and the key point is that one can choose the constant $C$ so that it does not depend on $\theta$ (though it may depend on $\theta^*$). Recall that if $\psi$ is linear then the optimal control is deterministic and we remark that such a bound can be found for all optimal controls locally uniformly with respect to $\theta$.
}

Before the theorem, we introduce a lemma which is useful in showing that many relevant quantities are differentiable and bounded. This lemma concerns the function
\begin{align}
	g(t,U) &= \mathbb{E}[\,\psi(\tilde{U}_T)\,|\,\tilde{U}_t = U\,]\,,\label{eqn:g1}
\end{align}
where the process $\tilde{U}=(\tilde{U}_t)_{t\in[0,T]}$ satisfies the SDE
\begin{align}
d\tilde{U}_t &= \beta\, dt + \eta \,dZ_t\,.
\end{align}
This function plays an important role in our approximation to the value function and in our candidate approximately optimal trading strategy. We remark that $\partial_Ug(t,U)$ measures the sensitivity of $g$ to changes in the underlier $U$ and therefore has an interpretation similar to that of the ``delta'' of an option.  It is helpful in the discussion below to directly interpret this derivative as an option's ``delta'' even though they are not strictly equal because the expectation in \eqref{eqn:g} is taken under the physical measure rather than an equivalent risk-neutral measure. In addition, the process $\tilde{U}$ above is a fictitious process that equals the path of $U$ when there is no cross impact from trading.

\begin{lemma}[Future Option Delta]\label{lem:future_delta}
	Suppose $\psi$ satisfies Assumption \ref{ass:ass} i) ($\psi \in C^4(\mathbb{R})$ with bounded derivatives up to fourth order). Then
	\begin{align}
		\mathbb{E}\left[\partial_U g(s,\tilde{U}_s)\,|\,\tilde{U}_t=U\right] &= \partial_U g(t,U)\,, \qquad\forall \; t \leq s \leq T\,.\label{eqn:lemma1}
	\end{align}
	In addition, if the function $f:\mathds R \mapsto \mathds R$ is integrable, then
	\begin{align}
		\mathbb{E}\biggl[\int_t^Tf(s)\,\partial_U g(s,\tilde{U}_s)\,ds\,\biggl|\,\tilde{U}_t=U\biggr] &= \partial_U g(t,U)\int_t^T f(s)\,ds\,.\label{eqn:lemma2}
	\end{align}
	Finally, the expressions in \eqref{eqn:lemma1} and \eqref{eqn:lemma2} have derivatives up to third order with respect to $U$ which are bounded and continuous.
\end{lemma}
\begin{proof}
	Write $g(t,U)$ in terms of the transition density of the process $\tilde{U}$. Let
	\begin{align*}
		p(z;t,T,U) &= \frac{1}{\sqrt{2\,\pi\,\eta^2\,(T-t)}}\,\exp\biggl\{-\frac{(z-U-\beta\,(T-t))^2}{2\,\eta^2\,(T-t)}\biggr\}\,,
	\end{align*}
	therefore
	\begin{align*}
		g(t,U) &= \mathbb{E}[\,\psi(\tilde{U}_T)\,|\,\tilde{U}_t = U\,]= \int_{-\infty}^\infty \psi(z)\, p(z;t,T,U)\,dz
		= \int_{-\infty}^\infty \psi(x+U)\, p(x;t,T,0)\,dx.
	\end{align*}
The Leibniz integration rule may be used to differentiate the expression above because the derivative of $\psi$ is bounded, and we write
	\begin{align*}
		\partial_U g(t,U)
&= \int_{-\infty}^\infty \frac{d\psi}{dU}(x+U)\, p(x;t,T,0)\,dx
= \int_{-\infty}^\infty \frac{d\psi}{dU}(z)\, p(z;t,T,U)\,dz
= \mathbb{E}\left[\,\frac{d\psi}{dU}(\tilde{U}_T)\,|\,\tilde{U}_t = U\,\right].
	\end{align*}
This final expression is a Doob martingale, which shows the first claim. The second claim follows from Fubini's Theorem. The third claim follows from applying the first and second claims to a modified payoff by replacing $\psi$ with $ d\psi /dU$, $ d^2\psi /dU^2$, or $ d^3\psi /dU^3$. \qed
\end{proof}

The first claim in this lemma shows that the process $\partial_U g(t,\tilde{U}_t)$ is a martingale, and therefore the expected value of an option's delta in the future is equal to its delta at the present. The second claim states that the expected average future value of the option's delta is equal to its present value when $f\equiv 1$. In addition to providing convenient bounds throughout much of the following, many of the appearances of $\partial_U g(t,\tilde{U}_t)$ within complicated expressions below easily simplify -- this motivates Proposition \ref{prop:closed_form}.

\begin{theorem}[Asymptotic Approximation of Value Function]\label{prop:asymptotic_approximation} The function $h_\psi$ in equation \eqref{eqn:ansatz} admits the following approximation:
	
	{i) \bf Expansion}:
		\begin{align}
		\begin{split}
			h_\psi(t,q,U;\theta \,c,\theta\,\gamma) &= \hat{h}(t,q,U;\theta\, c, \theta\, \gamma) + R(t,q,U;\theta)\,,\\
			\hat{h}(t,q,U;\theta \,c,\theta\,\gamma) &= h_0(t,q,U) + \theta \,\left(c\,h_1(t,q,U) + \gamma\, h_2(t,q,U)\right) \\
			& \hspace{10mm} + \theta^2\, \left(c^2 \, h_3(t,q,U) +  c\,\gamma \,h_4(t,q,U) + \gamma^2\,h_5(t,q,U)\right)\,,
		\end{split}\label{eqn:h_expansion}
		\end{align}
		such that
		\begin{equation}
			\lim_{\theta\downarrow0}\tfrac{1}{\theta^2}R(t,q,U;\theta)=0\,.\label{eqn:error_limit}
		\end{equation}
	{ii) \bf Zero and First Order Terms}:
		The functions $h_0$, $h_1$, and $h_2$ may be taken as
		\begin{subequations}
			\begin{align}
				h_0(t,q,U) &= f_0(t) + f_1(t)\,q + f_2(t)\,q^2 + g(t,U)\,,
				\label{eqn:value_order_0}				\\
				h_1(t,q,U) &= \lambda_0(t,U) + \lambda_1(t,U)\,q\,,								\label{eqn:h1}\\
				h_2(t,q,U) &= \Lambda_0(t,U) + \Lambda_1(t,U)\,q + \Lambda_2(t)\,q^2\,,			\label{eqn:h2}
			\end{align}
		\end{subequations}
		where by letting $m = 2\,\alpha-b$,
		\addtolength{\jot}{3pt}
		\begin{subequations}
			\begin{align}
				f_0(t) &= \tfrac{1}{4\,k}\int_t^Tf_1^2(s)\,ds\,,\\
				f_1(t) &= \frac{\mu\,(T-t)(4\,k+m\,(T-t))}{4\,k + 2\,m\,(T-t)}\,,\\
				f_2(t) &= \frac{-k\,m}{2\,k+m\,(T-t)} - \frac{b}{2}\,,
				\end{align}
			\label{eqn:f012}
		\end{subequations}
		\begin{align}
			g(t,U) &= \mathbb{E}[\,\psi(\tilde{U}_T)\,|\,\tilde{U}_t = U\,]\,,\label{eqn:g}
		\end{align}
		\begin{subequations}
			\begin{align}
				\lambda_0(t,U) &= \mathbb{E}\biggl[\int_t^T\frac{f_1(s)}{2\,k}\biggl(\lambda_1(s,\tilde{U}_s)+\partial_U g(s,\tilde{U}_s)\biggr)ds\,\biggl|\,\tilde{U}_t=U\biggr]\,,\label{eqn:lambda_0}\\
				%\lambda_0(t,U) &= \frac{\mu(T-t)^2}{2(2k + m(T-t))}\partial_Ug(t,U)\,,\label{eqn:lambda_0}\\
				\lambda_1(t,U) &= \frac{-m}{2\,k+m\,(T-t)}\,\mathbb{E}\biggl[\int_t^T\partial_U g(s,\tilde{U}_s)\,ds\,\biggl|\,\tilde{U}_t=U\biggr]\,,\label{eqn:lambda_1}
				%\lambda_1(t,U) &= \frac{-m(T-t)}{2k+m(T-t)}\partial_Ug(t,U)\,,\label{eqn:lambda_1}\\
			\end{align}
		\end{subequations}
		\begin{subequations}
			\begin{align}
				\Lambda_0(t,U) &=
				\tfrac{1}{2\,k}\,
				\mathbb{E}\biggl[\int_t^T \left(f_1(s)\Lambda_1(s,\tilde{U}_s)
				-k\,\eta^2\,(\partial_Ug(s,\tilde{U}_s))^2\right)\,\biggl|\,\tilde{U}_t=U\biggr]\,,\label{eqn:Lambda_0}\\
				\Lambda_1(t,U) &= \tfrac{1}{k}\,\mathbb{E}\biggl[\int_t^T\frac{2\,k+m\,(T-s)}{2\,k+m\,(T-t)}
				\left(f_1(s)\Lambda_2(s) - k\,\rho\,\sigma\,\eta\,\partial_Ug(s,\tilde{U}_s)\right)\,ds\,\biggl|\,\tilde{U}_t=U\biggr]\,,\label{eqn:Lambda_1}\\
				\Lambda_2(t) &= -\sigma^2\,(T-t)\,\frac{12\,k^2 + 6\,k\,m\,(T-t) + m^2\,(T-t)^2}{6\,(2\,k+m\,(T-t))^2}\,,\label{eqn:Lambda_2}
			\end{align}
		\end{subequations}
		\addtolength{\jot}{-3pt}
		where the process $\tilde{U}=(\tilde{U}_t)_{t\in[0,T]}$ satisfies the SDE
		\begin{align}
			d\tilde{U}_t &= \beta\, dt + \eta \,dZ_t\,,
		\end{align}
		{iii) \bf Second Order Terms}: The functions $h_3$, $h_4$, and $h_5$ may be taken as
		\begin{subequations}
			\begin{align}
			h_3(t,q,U) &= A_0(t,U) + A_1(t,U)\,q + A_2(t,U)\,q^2\,,							\label{eqn:h3}\\
			h_4(t,q,U) &= B_0(t,U) + B_1(t,U)\,q + B_2(t,U)\,q^2\,,							\label{eqn:h4}\\
			h_5(t,q,U) &= C_0(t,U) + C_1(t,U)\,q + C_2(t,U)\,q^2\,.							\label{eqn:h5}
			\end{align}
		\end{subequations}
		where each $A_{0,1,2}$, $B_{0,1,2}$, and $C_{0,1,2}$ is bounded and continuously differentiable with respect to $U$.
\end{theorem}

\begin{proof}
	See Appendix A.
\end{proof}
The decomposition of the value function warrants some discussion, but much of the intuition behind these expressions becomes clearer when we consider how they influence an approximately optimal trading speed. This is demonstrated in the next theorem. An immediate consequence of this theorem is that the inventory process becomes stochastic.

\begin{theorem}[Asymptotic Approximation of Optimal Trading Speed]\label{prop:approx_nu} Let $\hat{\nu}$ be a feedback control given by
	\begin{align}
		\hat{\nu}(t,q,U;\theta\, c,\theta\, \gamma) = \nu_0(t,q) + \theta \left(c\,\nu_1(t,U) + \gamma\, \nu_2(t,q,U)\right)\,,\label{eqn:approx_optimal_control}
	\end{align}
	with
	\begin{subequations}
		\begin{align}
			\nu_0(t,q) &= \tfrac{1}{2\,k}\left(f_1(t) + (2\,f_2(t) + b)\,q\right)\,,\label{eqn:nu_0}\\
			\nu_1(t,U) &= \tfrac{1}{2\,k}\left(\partial_U \,g(t,U) + \lambda_1(t,U)\right)\,, \label{eqn:nu_1}\\
			\nu_2(t,q,U) &= \tfrac{1}{2\,k}\left(\Lambda_1(t,U) + 2\,\Lambda_2(t)\,q \right)\,.\label{eqn:nu_2}
		\end{align}
		\label{eqn:nu012-approx}
	\end{subequations}
	Then $\hat{\nu}_t = \hat{\nu}(t,Q_t^{\hat{\nu}},U_t^{\hat{\nu}},\theta\, c,\theta\, \gamma)$ is an admissible control. Defining $h^{\hat{\nu}}$ by the relation
	\begin{align*}
		H^{\hat{\nu}}(t,x,q,S,U;\theta\, c, \theta\,\gamma) &= -e^{-\theta\,\gamma\,(x+q\,S + h^{\hat{\nu}}(t,q,U;\theta\, c, \theta\,\gamma))}\,,
	\end{align*}
	$\hat{\nu}$ is asymptotically optimal to second order:
	\begin{align*}
		h_\psi(t,q,U;\theta\, c, \theta\,\gamma) &= h^{\hat{\nu}}(t,q,U;\theta\, c, \theta\,\gamma) + o(\theta^2)\,.
	\end{align*}
\end{theorem}
\begin{proof}
	For the proof see the Appendix.
\end{proof}

For the purposes of discussing the interpretation of the quantities in \eqref{eqn:nu012-approx} we assume that $\partial_Ug$ is positive for all $t$ and $U$. Nearly all of the discussion below holds similarly if $\partial_U g$ is negative, except with the agent's actions also being appropriately changed (i.e., selling instead of buying).

The zero order term,  which we denote by $\nu_0$, has a clear interpretation. This term represents the optimal trading speed of a risk-neutral  agent when there is no cross-price impact between the traded and the non-tradable risk factors. The feedback form of this term is the same as the term that appears in an optimal execution program for a single asset with no risk-aversion. Observe that the zero order term $h_0$ of the value function in \eqref{eqn:value_order_0}  is the sum of the value of such an optimal trading program as well as the expected future payoff $\psi$ under Bachelier dynamics. This is again due to the lack of risk-aversion and, in this limit, the absence of any interaction between the $S$ and $U$.

The correction term $\nu_1$ in the optimal trading speed is due to cross-price impact and contains two components.  The term $\partial_U g(t,U)$ arises directly due to the impact that the agent's trades have on the current value of the option. As we assume $g$ is an increasing function with respect to $U$, this term has the effect of making the agent increase the speed of trading.  Buying more shares tends to increase the price process $U_t$,  which increases the value of the option.

With increasing $g$, the second component $\lambda_1(t,U)$ is negative as seen from \eqref{eqn:lambda_1}, which results in slowing down the rate of buying shares. This term arises from the agent's desire to finish with inventory close to zero to avoid the terminal liquidation penalty. As she knows that any shares she buys now she will partially liquidate in the future, she wants to avoid accumulating a large position in $S$ which results in costly round-trip trades. The value of $\lambda_1(t,U)$ is a measurement of the expected average future option delta weighted by how fast the agent expects to liquidate in the future. By lowering the trading speed by this amount, the agent is balancing the benefit of buying now and increasing the option value, while knowing she has to sell in the future and lowering the option value, both trades incur a cost due to temporary price impact.

Expression \eqref{eqn:nu_2} in the trading speed due to risk-aversion has two components. The term $2\,\Lambda_2(t)\,q$ acts to bring the agent's inventory closer to zero (note that $\Lambda_2$ is always negative). This term arises because the agent wants to avoid inventory risk, which exposes her to the risk in the traded asset price $S_t$.

The term $\Lambda_1(t,U)$ has indeterminate sign, so it could result in either more or less buying. It stems from two sources of risk as can be seen in the integrand of \eqref{eqn:Lambda_1}. The first is related to a tradeoff between inventory risk and passive gain when holding non-zero inventory. If $\mu \neq 0$ then the agent has incentive to hold non-zero inventory and benefit from the trending price of $S$, but this also exposes the agent to risk when holding inventory due to unexpected price changes. If $f_1(t)$ quantifies the desired speed of trading to benefit from price drift, then the first term in the integrand of \eqref{eqn:Lambda_1} quantifies the correction associated with not accumulating a risky position. The second source of risk in $\Lambda_1$ is that associated with holding the option and a non-zero position in the traded asset. As $S$ and $U$ are correlated, the agent can reduce her risk exposure by tending to favor a position in the traded asset which cancels out the random changes in the option value

\subsection{Closed-form Approximation to Optimal Trading Strategy}

While the approximation of the trading strategy given in Theorem \ref{prop:approx_nu} involves some complicated expressions, it makes it clear how each of the components of the dynamics affect the agent's trading speed. In this section, we approximate the optimal control process by another simpler control with a closed-form expression that is easier to evaluate.

Using Lemma \ref{lem:future_delta}, the approximation to the optimal control can be computed in closed-form, however, it involves the evaluation of several one-dimensional integrals of rational functions. Instead we employ the optimal strategy in the linear payoff case (which admits a closed-form expression, see Theorem \ref{prop:optimal_control}) to provide an approximation for the non-linear case. We let  $\mathfrak{v}^*$ be the feedback form of the optimal strategy when the agent has linear exposure of $X$ units of the non-tradable risk factor, which is given by
\begin{align}
	\mathfrak{v}^*(t,q,X;\theta\,c,\theta\,\gamma) = \tfrac{1}{2\,k}\left(\,\theta\,c\,X + h_1(t;X,\theta) + (2\,h_2(t,\theta) + b)\,q\right)\,.
\end{align}
Our closed-form approximation for the optimal trading strategy is summarized by the following two results.
\begin{proposition}[Closed-form Approximation of Optimal Trading Speed]\label{prop:closed_form}
	The following approximation holds locally uniformly in $(t,q,U)$:
	\begin{align}
		\mathfrak{v}^*(t,q,\partial_Ug(t,U);\theta\,c,\theta\,\gamma) = \hat{\nu}(t,q,U;\theta\,c,\theta\,\gamma) + o(\theta)\,.\label{eqn:local_uniform}
	\end{align}
	Let $\nu'$ be a control given by
	\begin{align}
		\nu'_t = \mathfrak{v}^*(t,Q_t^{\nu'},\partial_U g(t,U_t^{\nu'});\theta\, c,\theta\, \gamma)\,.\label{eqn:admissible_control}
	\end{align}
	Then $\nu'$ is admissible. Define $h^{\nu'}$ by the relation
	\begin{align*}
		H^{\nu'}(t,x,q,S,U;\theta\, c, \theta\,\gamma) &= -e^{-\theta\,\gamma(x+q\,S + h^{\nu'}(t,q,U;\theta\, c, \theta\,\gamma))}\,,
	\end{align*}
	so that $\nu'$ is asymptotically approximately optimal to second order:
	\begin{align}
		h_\psi(t,q,U;\theta\, c, \theta\,\gamma) &= h^{\nu'}(t,q,U;\theta\, c, \theta\,\gamma) + o(\theta^2)\,.\label{eqn:nu_bar_approx}
	\end{align}
\end{proposition}
\begin{proof}
	For a proof see the  Appendix.
\end{proof}

This proposition shows that the agent can approximate the  optimal trading speed by  trading at time $t$ as if she were holding $\partial_Ug(t,U_t)$ units of the non-tradable risk factor, and, as before, the value of these units will not be paid until $T$. This approximation is sensible because an option's delta represents locally the equivalent number of shares of the underlier that the agent holds in terms of risk and reward exposure.

This closed-form approximation works for the trading speed, but no such approximation holds by making a similar substitution of $\partial_Ug(t,U)$ for $\mathfrak{N}$ in the closed form expressions of inventory and value function in the linear payoff case. The inventory position at time $t$  depends on the entire path of $U$ up to time $t$, which is given by
\begin{align*}
	Q_t &= Q_0 + \int_0^t\nu_s \,ds\,.
\end{align*}
Thus, even if $\nu_t$ depends on the process $U$ only through its value at time $t$, the inventory does not have this property.

\subsection{Simulation of Agent's Inventory Position}

In this section we consider a specific form of the exposure $\psi$ and investigate the agent's optimal trading strategy. The exposure is in the form of $N$ European call options written on $U$ with strike $K = U_0$. The maturity of the option is  $T+\delta t$ for a small value of $\delta t$. This ensures that the payoff function $\psi$ is twice continuously differentiable around $T$. Our approximation to the value function and optimal trading speed require us to compute the value of the option and its delta under Bachelier dynamics. Elementary computations show that $g$ in equation \eqref{eqn:g} and its derivative are given by
\begin{subequations}
	\begin{align}
		g(t,U) &= N\,\eta\,\sqrt{T+\delta t-t}\,\left(z\,\Phi(z) + \phi(z)\right)\,,\\
		\partial_Ug(t,U) &= N\,\Phi(z)\,,\qquad \text{and}\\
		z &= \tfrac{U - K}{\eta}\,(T+\delta t-t)^{-\frac{1}{2}} + \tfrac{\beta}{\eta}\,(T+\delta t-t)^{\frac{1}{2}}\,,
	\end{align}
\end{subequations}
where $\Phi$ and $\phi$ are the standard normal cumulative distribution and density functions, respectively.

\subsubsection{Effect of Cross Price Impact}

We begin with the case $\gamma = 0$ to observe the impact of the parameter $c$ on the agent's trading speed. When $\gamma=0$, we do not apply Proposition \ref{prop:closed_form} to approximate the trading speed because many quantities in \eqref{eqn:optimal_control} are undefined at $\gamma=0$. It is possible, however, to compute them in the limiting sense $\gamma\rightarrow 0$. Instead, an application of Theorem \ref{prop:approx_nu}, along with Lemma \ref{lem:future_delta} for computing $\lambda_1$, shows that for small $c$ the optimal trading speed may be approximated by
\begin{align}
	\hat{\nu}_t &= \tfrac{1}{2\,k}\left(f_1(t) + (2\,f_2(t) + b)\,Q_t\right) + c\,(2\,k + m\,(T-t))^{-1}\, \partial_Ug(t,U_t)\,.\label{eqn:speed_c}
\end{align}
Interestingly, if we force the agent to finish with zero inventory, by taking the limit $\alpha\rightarrow\infty$, then the effect of cross-price impact disappears (recall $m = 2\,\alpha -b$) and the agent behaves according to an optimal trading program with one asset. This is because the net effect of the agent's trading on the process $U$ only depends on the net change in inventory, which is always equal to $-c\,Q_0$ if the agent must have $Q_T = 0$. If the total effect on $U$ is the same regardless of the trading strategy then there can be no additional benefit of basing the trades off of $U$.

We simulate several paths of the price process $U_t$ taking into account the cross impact of the agent's own trades and plot the resulting inventory paths. These are shown in Figure \ref{fig:inventory_c_1}. We see distinct behaviour depending on whether the option ends in-the-money or not. As the option maturity approaches, if the agent can be relatively certain that it will expire out-of-the-money then she begins to adopt a strategy which essentially mimics a risk-neutral optimal liquidation program as in \cite{almgren2001optimal}.

On the other hand, if the agent believes the option will end up in-the-money, then she chooses a target inventory level which is not zero. If $U_t$ is sufficiently larger than $K$ and $t$ sufficiently close to $T$, then $\partial_Ug(t,U_t)$ is equal to $N$, the number of options held, until maturity. This is seen by expanding and rearranging \eqref{eqn:speed_c} to give
\begin{equation}\label{eqn:v with m}
	\hat{\nu}_t = \frac{f_1(t)}{2\,k}-\frac{m\,(Q_t - \frac{c}{m}\,\partial_Ug(t,U_t))}{2\,k + m\,(T-t)}\,.
\end{equation}
The second term on the right-hand side in \eqref{eqn:v with m} has the effect of making the inventory $Q$ tend to $c\,\partial_Ug(t,U_t)/m$ -- recall that $dQ_t = \nu_t\,dt$. The magnitude of this effect is intensified as  the strategy gets closer to $T$. For the choice of parameters in Figure \ref{fig:inventory_c_1}, the value of $c/m$ is approximately $0.01$ and we see that the in-the-money inventory paths approach this value multiplied by $N$ at time $T$.

At the beginning of the trading period the agent begins to purchase shares, which exerts a small pressure to increase the value of the option. Once the path of the non-tradable risk factor begins to develop, she updates the probability which she assigns to the option expiring in or out-of-the-money.

\begin{figure}
	\begin{center}
		{\includegraphics[trim=140 240 140 240, scale=0.42]{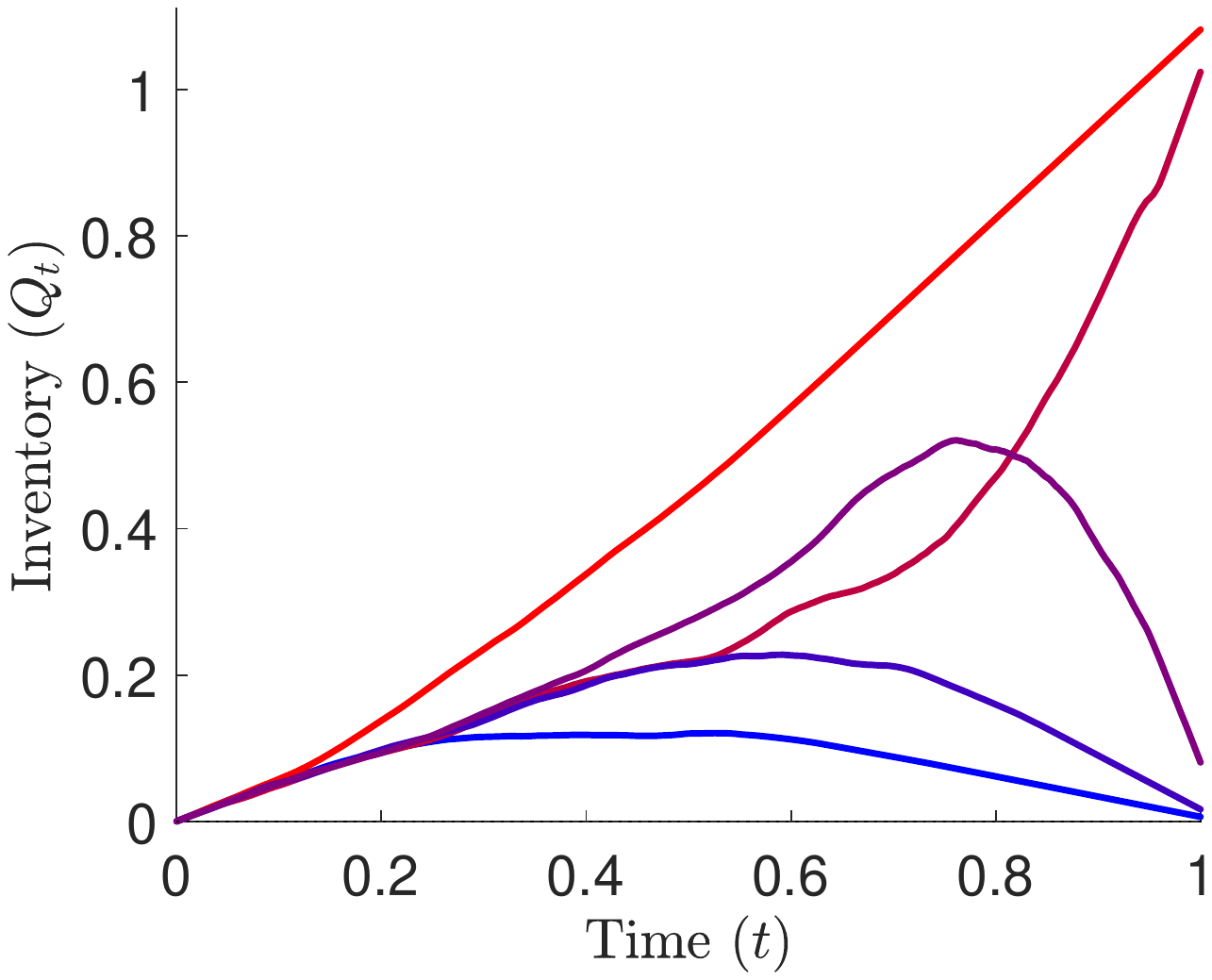}}\hspace{8mm}
		{\includegraphics[trim=140 240 140 240, scale=0.42]{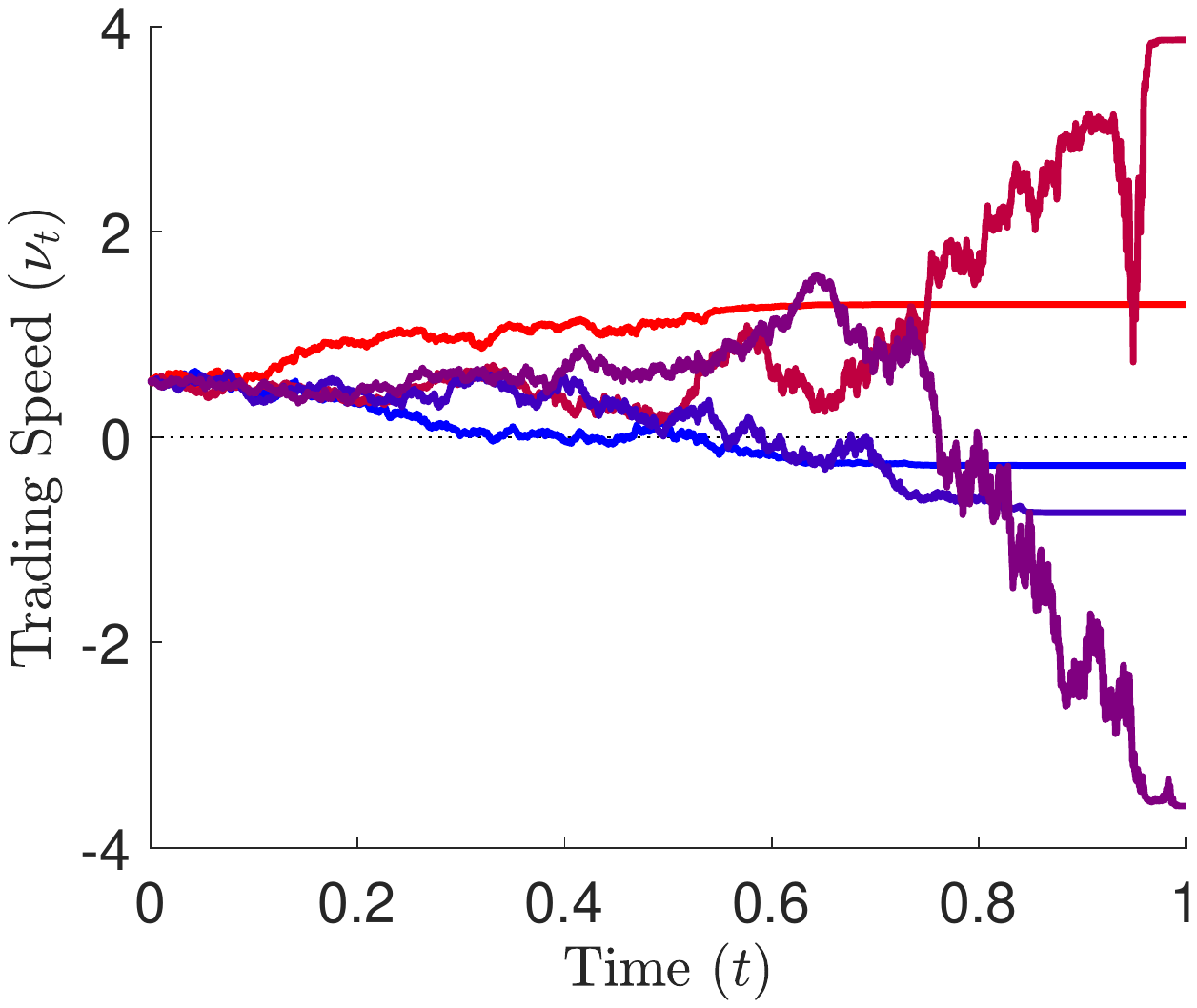}}\hspace{8mm}
		{\includegraphics[trim=140 240 140 240, scale=0.42]{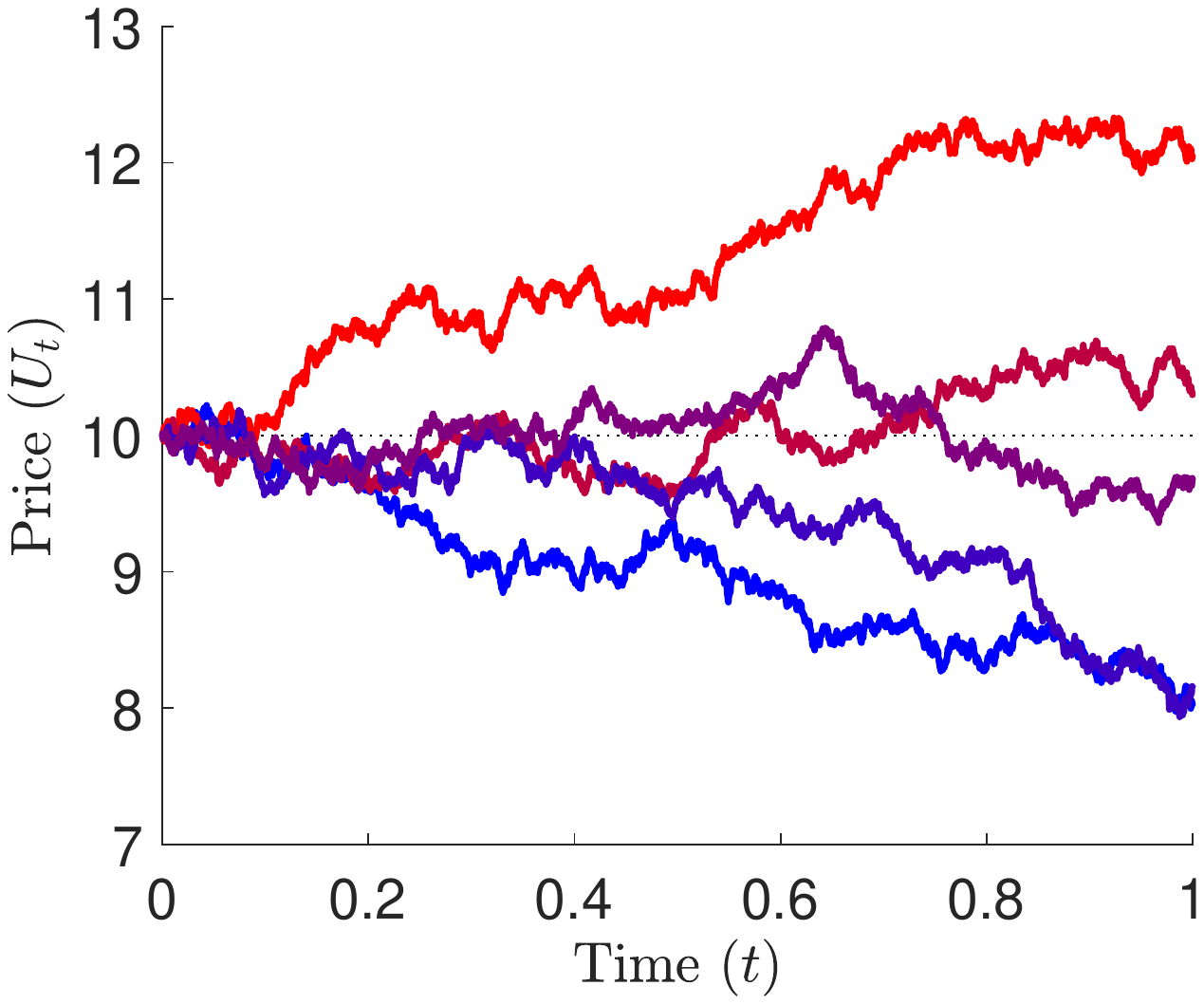}}
	\end{center}
	\vspace{-1em}
	\caption{Agent's optimal inventory position over time for $5$ simulated paths of $U$. In the left panel the agent's inventory position is displayed. The middle panel shows the agent's trading speed. The right panel shows the price of the non-tradable risk factor. Colors are chosen based on the final value of $U_T$ (larger values are red, smaller values are blue). Other model parameters are $\mu = 0$, $\beta = 0$, $\sigma = 1$, $\eta = 1$, $\rho = 0.5$, $b = 10^{-2}$, $c = 10^{-3}$, $k=10^{-3}$, $\gamma = 0$, $\alpha = 0.05$, $N = 100$, $\delta t = 10^{-5}$.\label{fig:inventory_c_1}}
\end{figure}

Figure \ref{fig:inventory_c_1} shows only a small number of paths, but the general distribution of some values is also of interest, in particular the distribution of terminal inventory. Figure \ref{fig:dist_c_1} shows the distribution of total inventory along with a scatter plot of the terminal inventory versus the terminal value of the non-tradable risk factor.

\begin{figure}
	\begin{center}
		{\includegraphics[trim=140 240 140 240, scale=0.48]{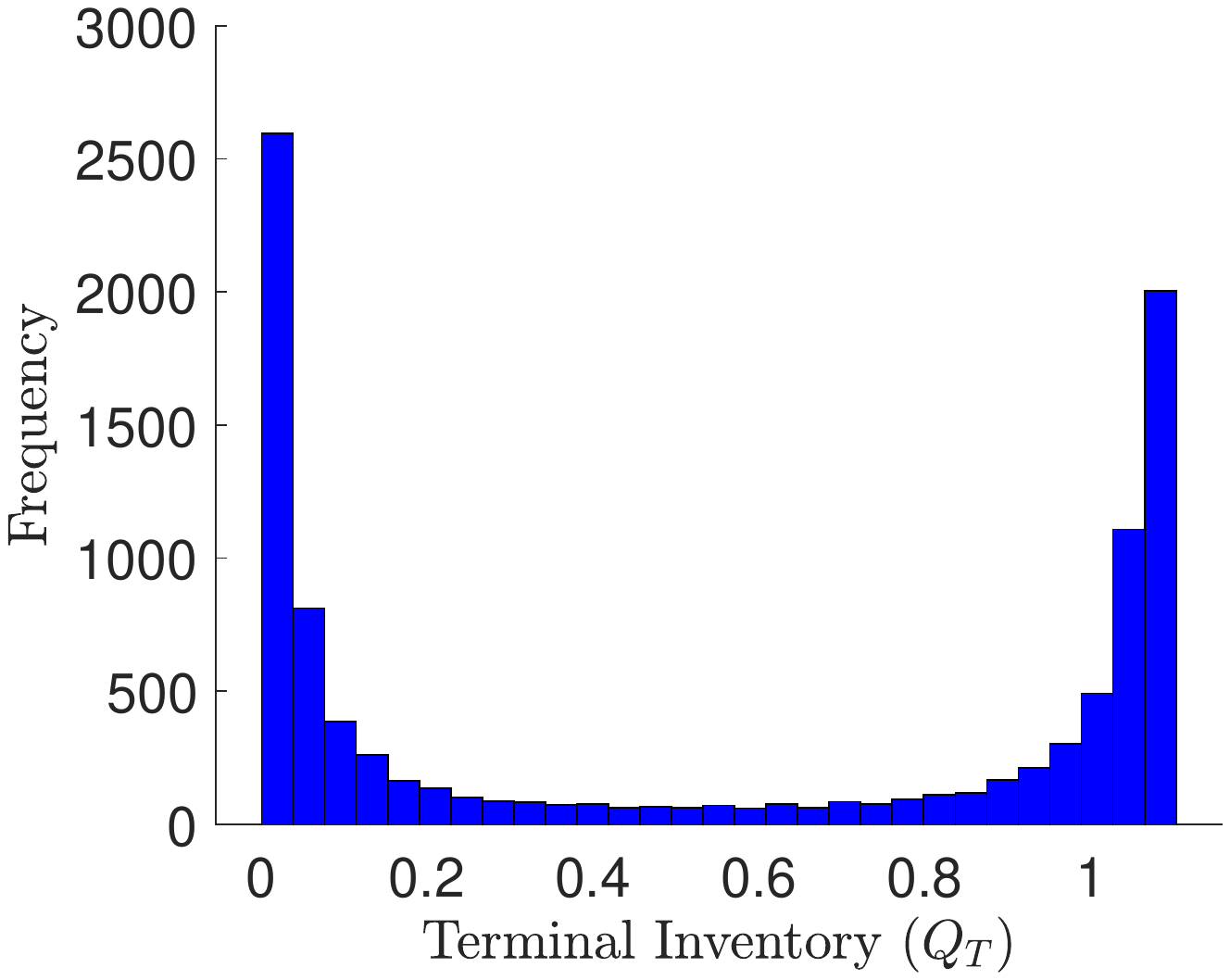}}\hspace{15mm}
		{\includegraphics[trim=140 240 140 240, scale=0.48]{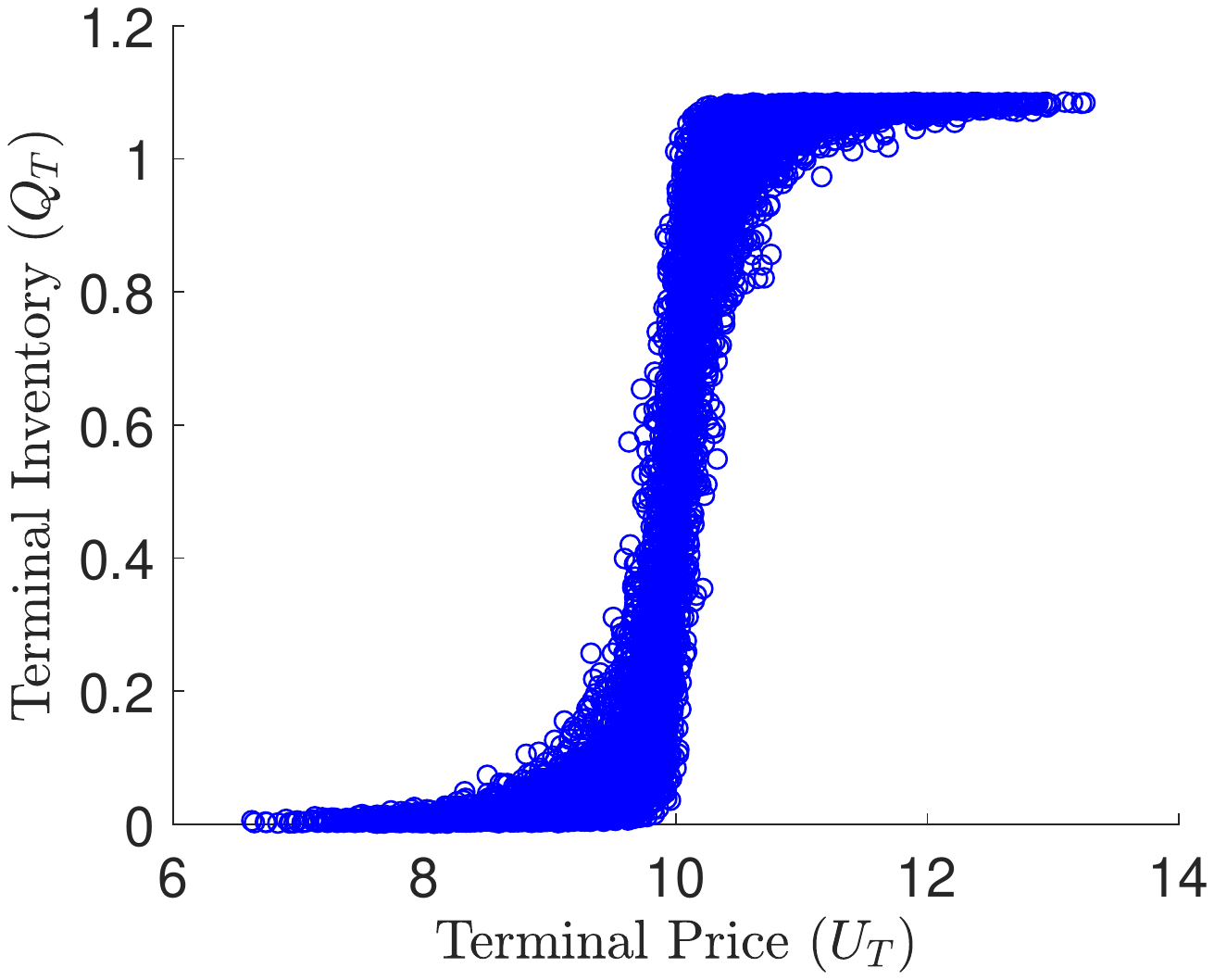}}
	\end{center}
	\vspace{-1em}
	\caption{Distribution of agent's terminal inventory and dependence of terminal inventory on terminal value of the non-tradable risk factor. Parameters are identical to those in Figure \ref{fig:inventory_c_1}. Number of simulations is $M = 10,000$. \label{fig:dist_c_1}}
\end{figure}

\subsubsection{Effect of Risk-Aversion}

Here, we set $c=0$ to consider only the effect of risk-aversion. We apply Proposition \ref{prop:closed_form} directly to compute the approximate optimal trading speed in closed-form. The paths of $U$ shown in the right panel of Figure \ref{fig:inventory_gamma_1} are the same as the unaffected paths from the previous example. That is, the realizations of the two Brownian motions are the same, but due to cross-price impact the actual paths are different. The magnitude of the difference is imperceptible in this example.

The general effect of risk-aversion in this example is to take a short position in the traded asset in a gradual manner, and then part way through the trading period to buy back that position and end with inventory close to zero. This is expected from a risk-averse agent when the payoff is a call option and the two assets have positive instantaneous correlation. The short position tends to decrease the variability in the overall holdings, which  consists of the traded asset and the option.

If we compare the results in Figure \ref{fig:inventory_gamma_1} with those in Figure \ref{fig:inventory_c_1}, we see that the effect of risk-aversion is opposite to the effect of cross-price impact. A positive cross-price impact parameter incentivizes the agent to acquire a long position in the traded asset whereas risk-aversion will always give incentive to short. In addition, the amount the agent desires to short depends on her estimate of the probability that the option will end up in-the-money or not. If it is very likely that the option ends in-the-money, then she will acquire a larger short position. If the $U$ moves in such a way that the agent expects with great confidence that the option will expire out-of-the-money, then she ceases the acquisition of the short position early and trades to target zero inventory at the end of the trading period. These two extreme opposite outcomes are seen by comparing the two price paths in Figure \ref{fig:inventory_gamma_1} which end at the highest and lowest points. The remaining paths have an intermediate behavior.

Also of notable interest is that the variance of the inventory position is greatest at the half way point of the trading period.
\begin{figure}
	\begin{center}
		{\includegraphics[trim=140 240 140 240, scale=0.42]{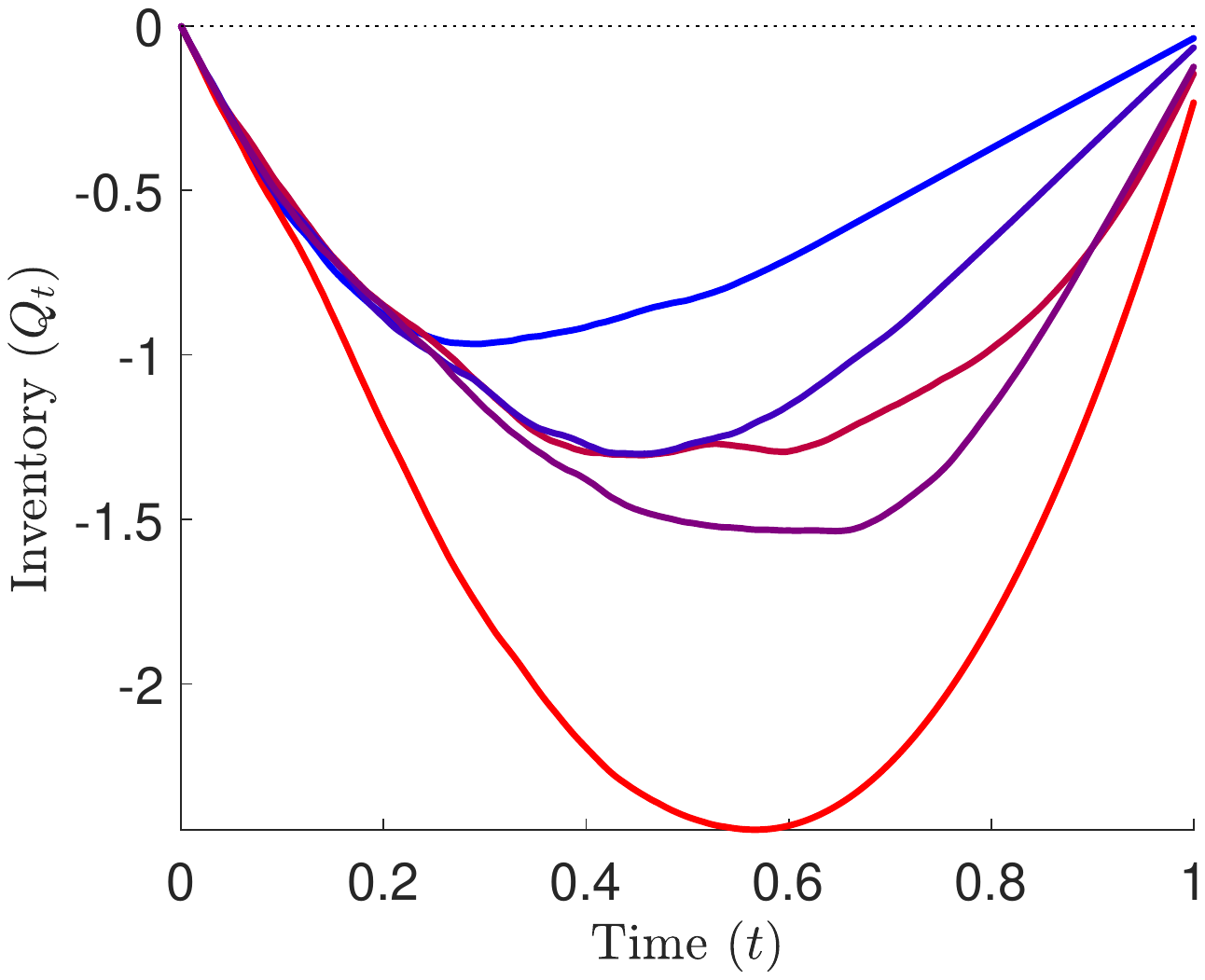}}\hspace{8mm}
		{\includegraphics[trim=140 240 140 240, scale=0.42]{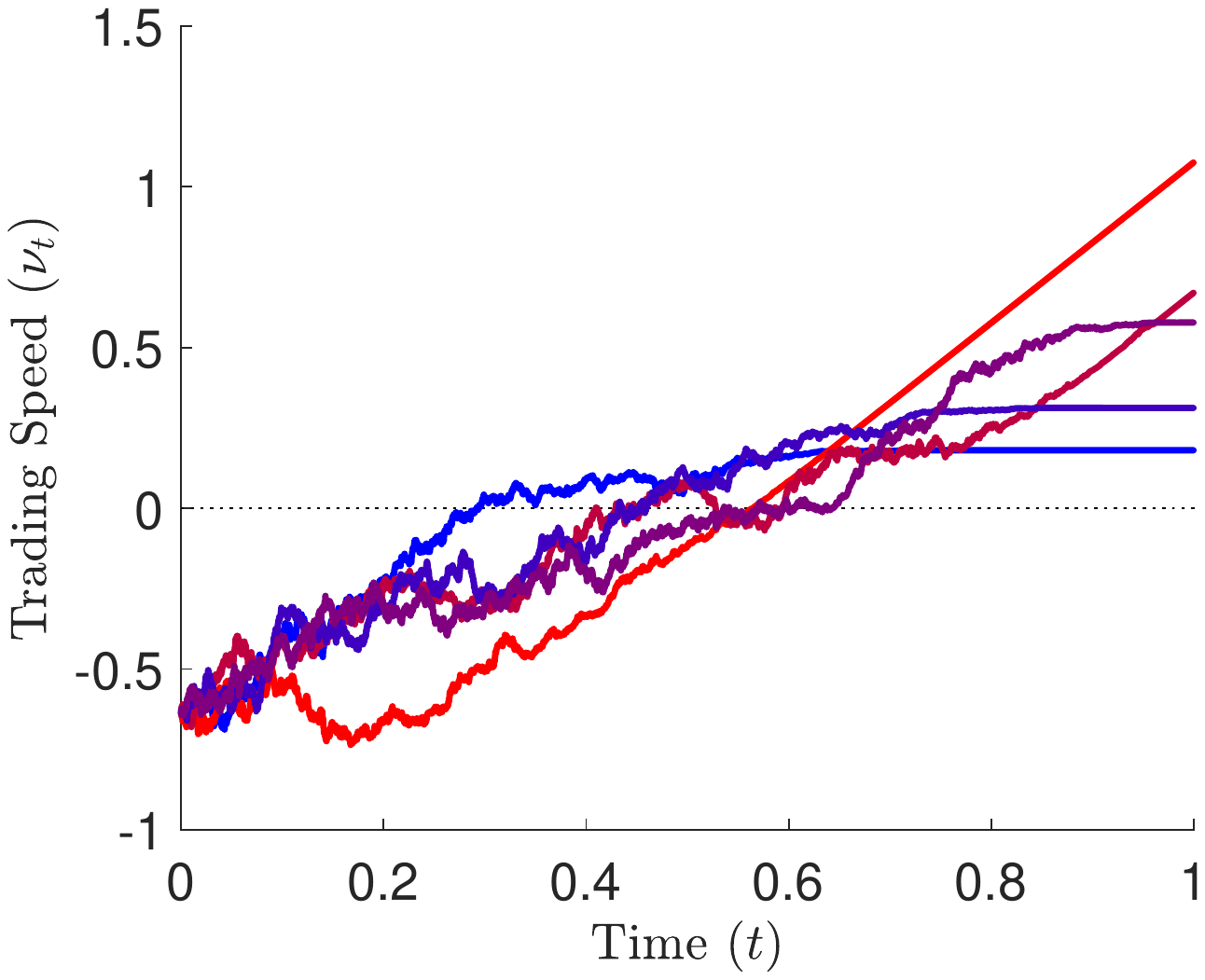}}\hspace{8mm}
		{\includegraphics[trim=140 240 140 240, scale=0.42]{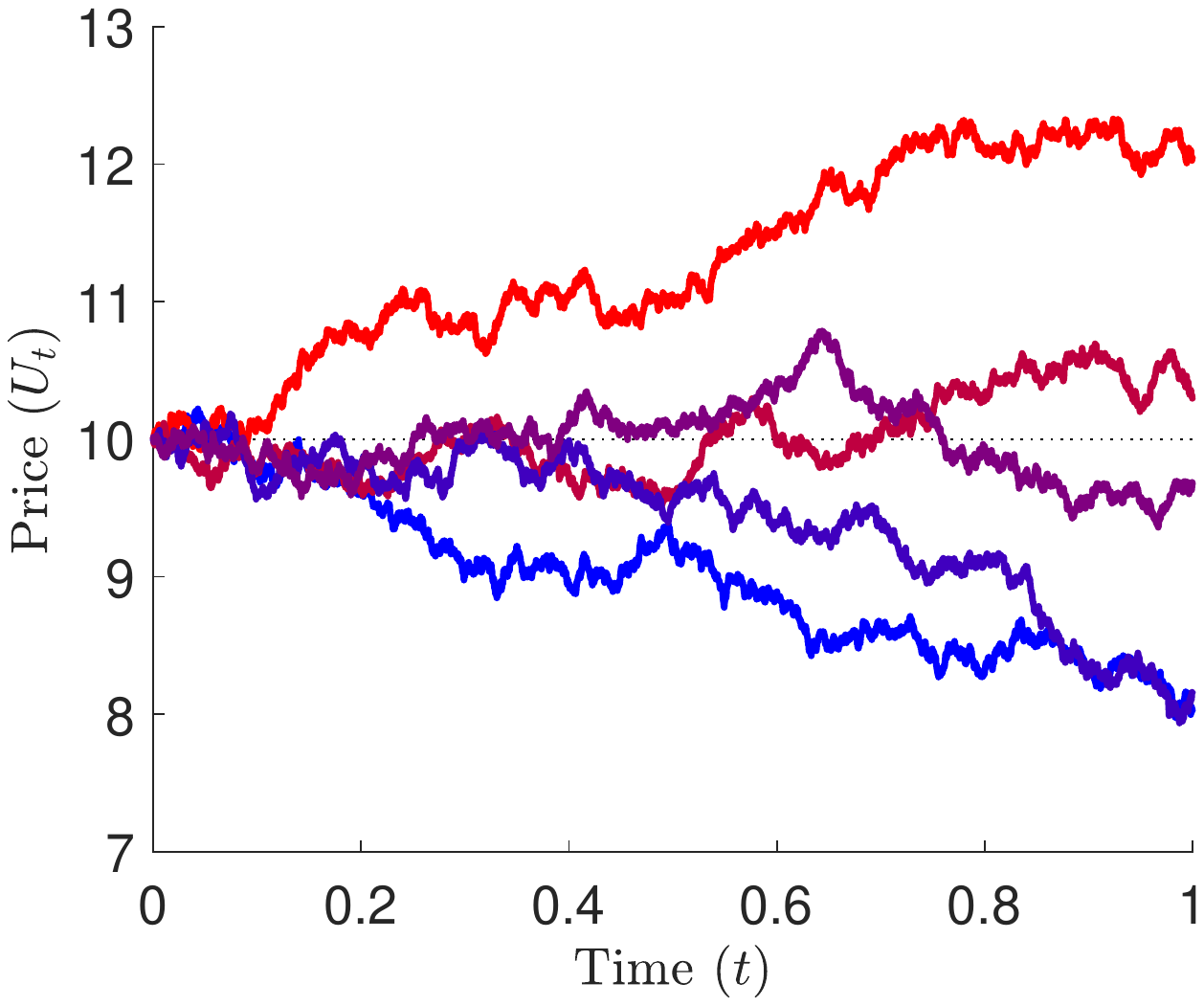}}
	\end{center}
	\vspace{-1em}
	\caption{Agent's optimal inventory position over time for $5$ simulated paths of $U$. In the left panel the agent's inventory position is displayed. The middle panel shows the agent's trading speed. The right panel shows the value of the non-tradable risk factor. Colors are chosen based on the final value of $U_T$ (larger values are red, smaller values are blue). Other model parameters are identical to those in Figure \ref{fig:inventory_c_1} except $c = 0$ and $\gamma = 10^{-3}$. \label{fig:inventory_gamma_1}}
\end{figure}

\begin{figure}
	\begin{center}
		{\includegraphics[trim=140 240 140 240, scale=0.48]{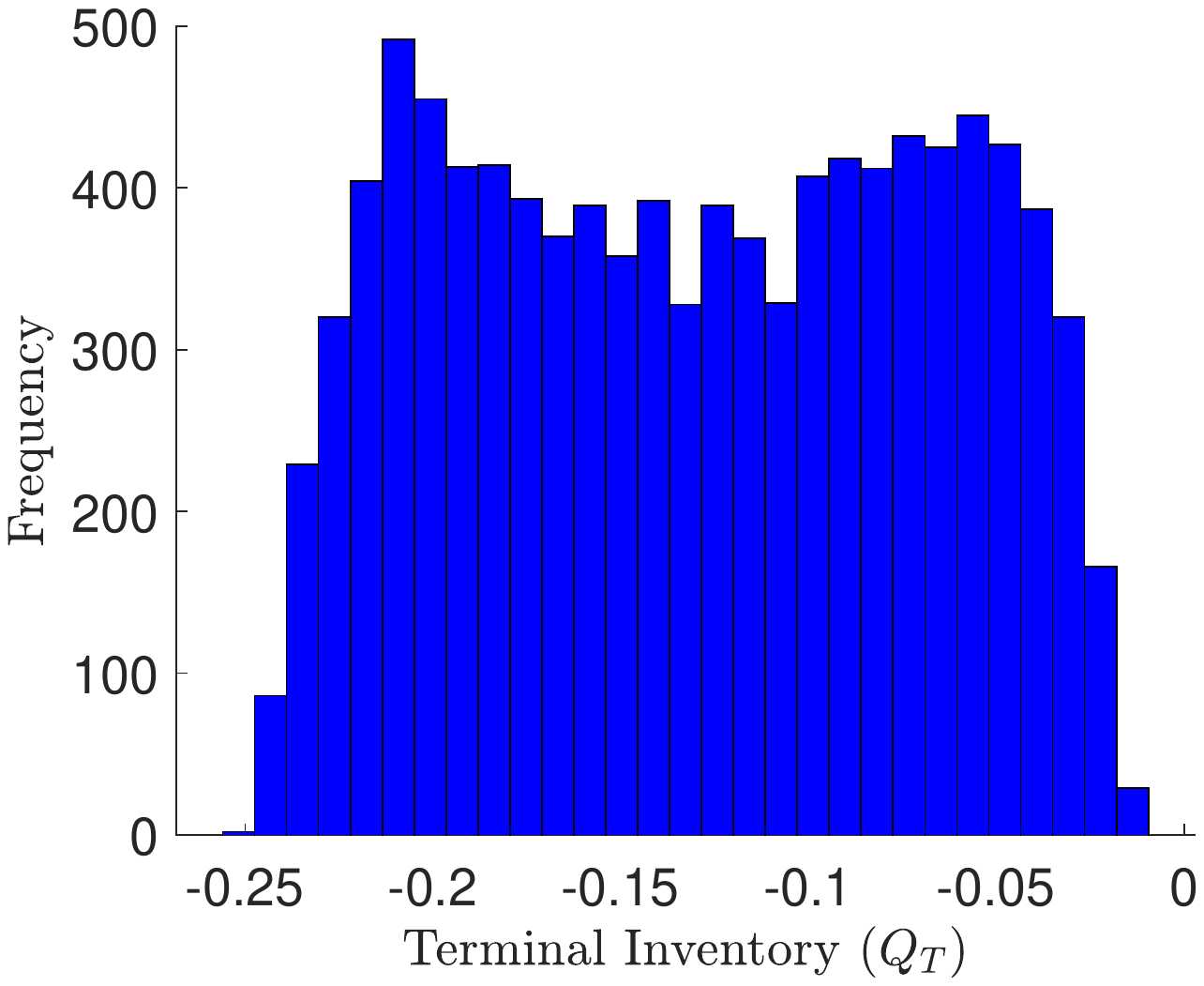}}\hspace{15mm}
		{\includegraphics[trim=140 240 140 240, scale=0.48]{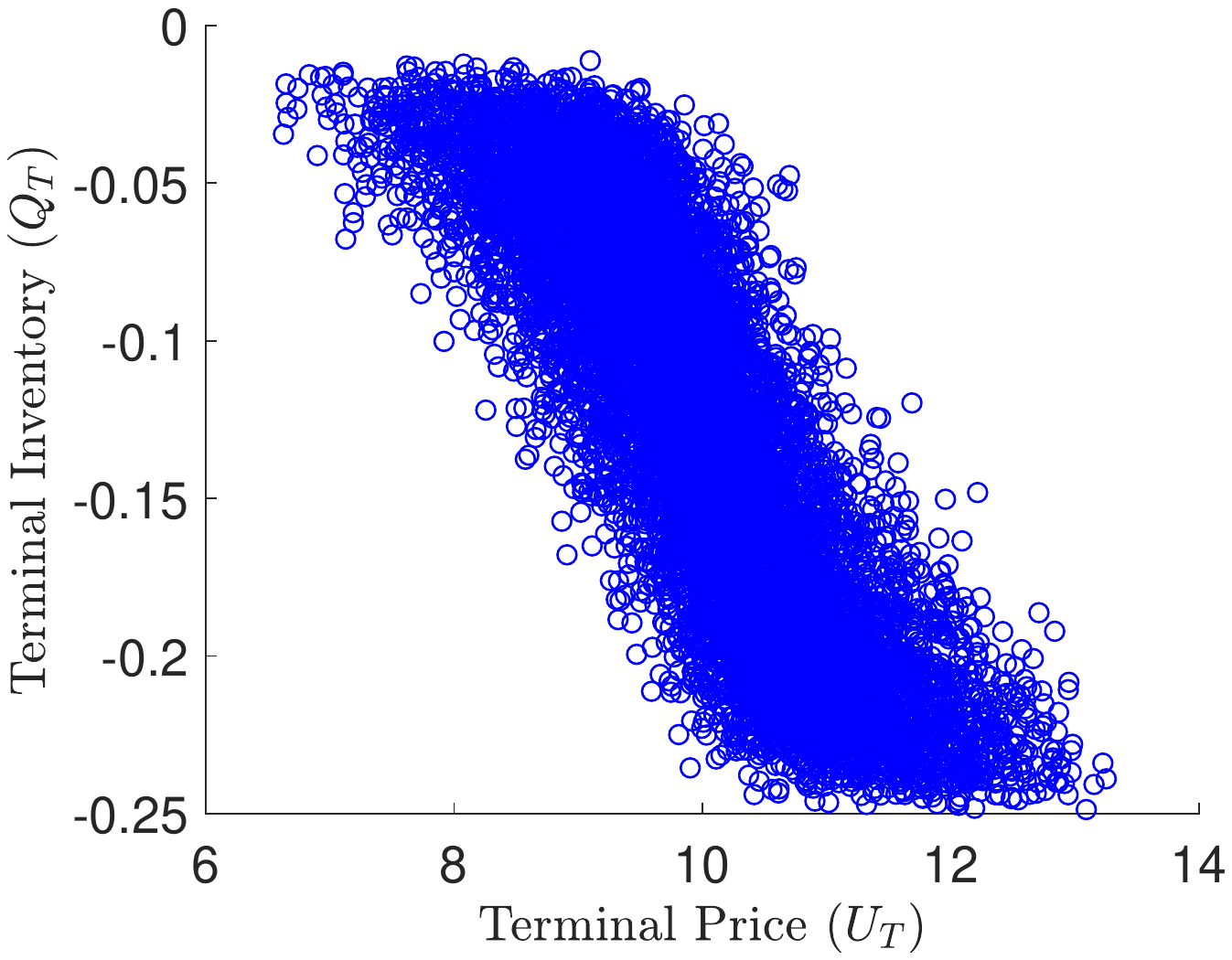}}
	\end{center}
	\vspace{-1em}
	\caption{Distribution of agent's terminal inventory and dependence of terminal inventory on the terminal value of the non-tradable risk factor. Parameters are identical to those in Figure \ref{fig:inventory_gamma_1}. Number of simulations is $M = 10000$. \label{fig:dist_gamma_1}}
\end{figure}

\subsubsection{Simultaneous Effect of Risk-Aversion and Cross Price Impact}

It is of interest to consider the behavior of the strategy when the effects of cross-price impact and risk-aversion are present because  these effects  tend to oppose each other. Figure \ref{fig:inventory_mix_1} shows the trading strategy and associated inventory path when both cross-price impact and risk-aversion are present. We see a combination of the counteracting effects that take place, namely the agent acquires a short position over most of the trading period to mitigate risk, but rather than liquidating this position she has incentive to acquire a long position before maturity if she is confident the option will expire in-the-money.
\begin{figure}
	\begin{center}
		{\includegraphics[trim=140 240 140 240, scale=0.42]{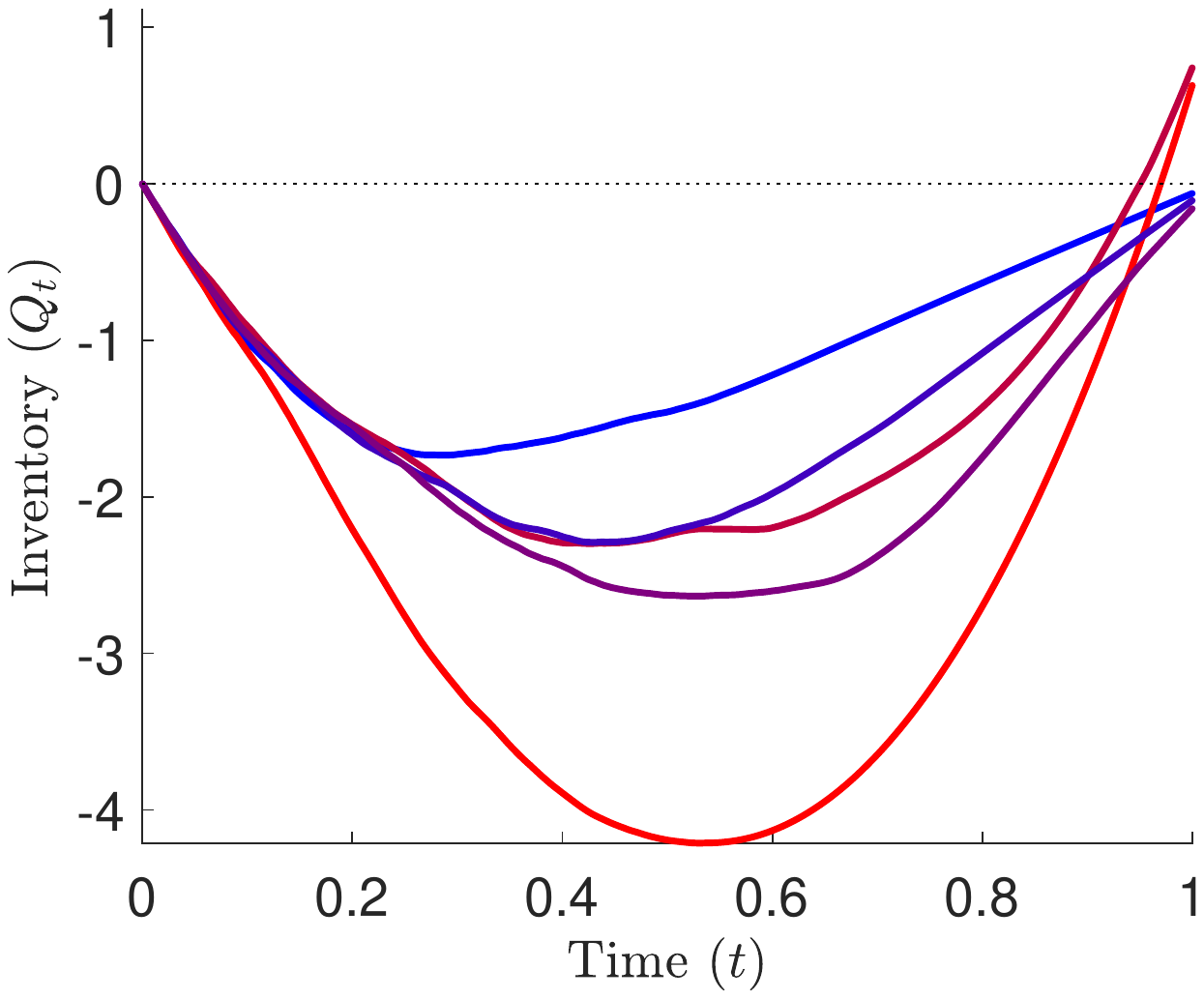}}\hspace{8mm}
		{\includegraphics[trim=140 240 140 240, scale=0.42]{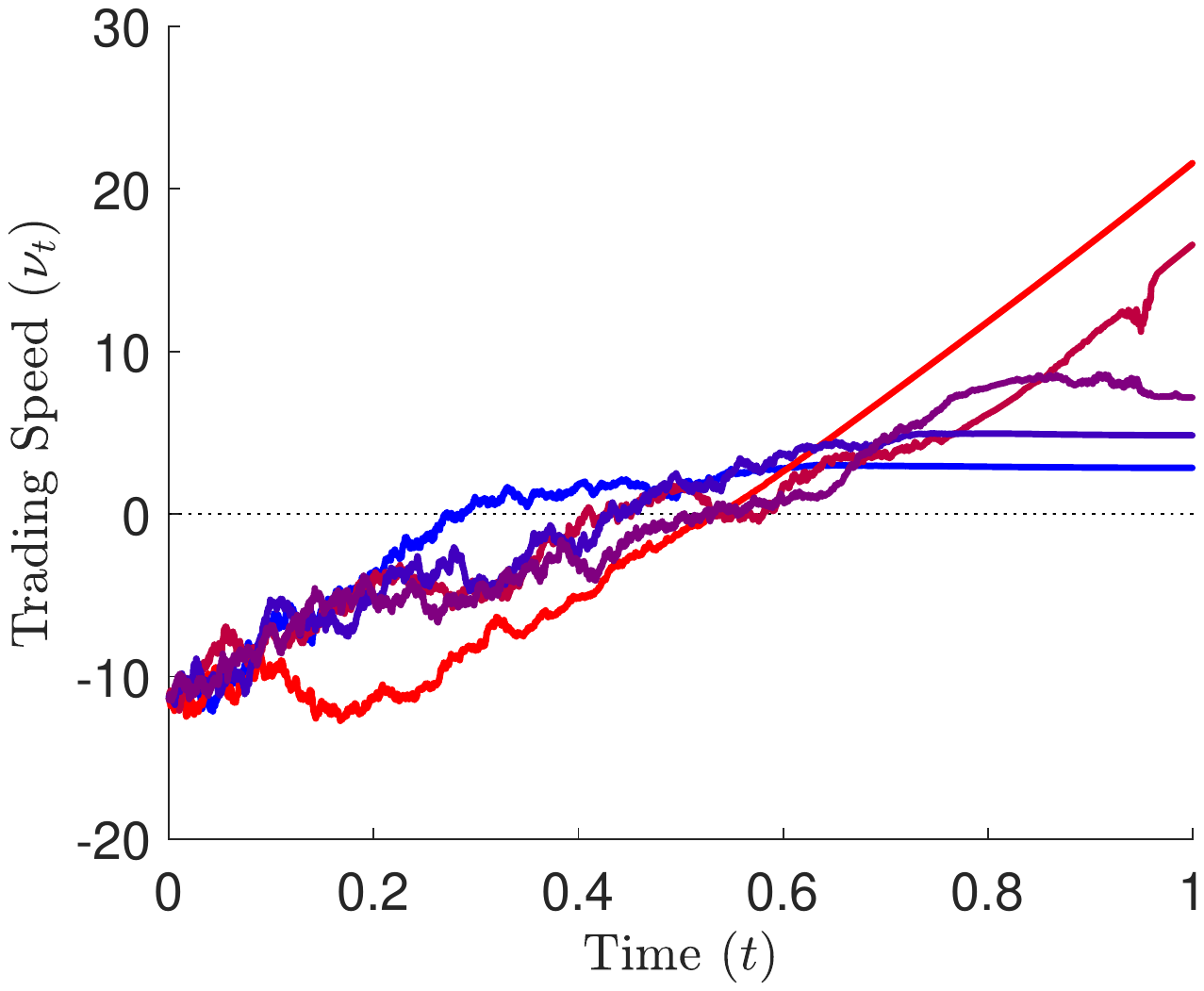}}\hspace{8mm}
		{\includegraphics[trim=140 240 140 240, scale=0.42]{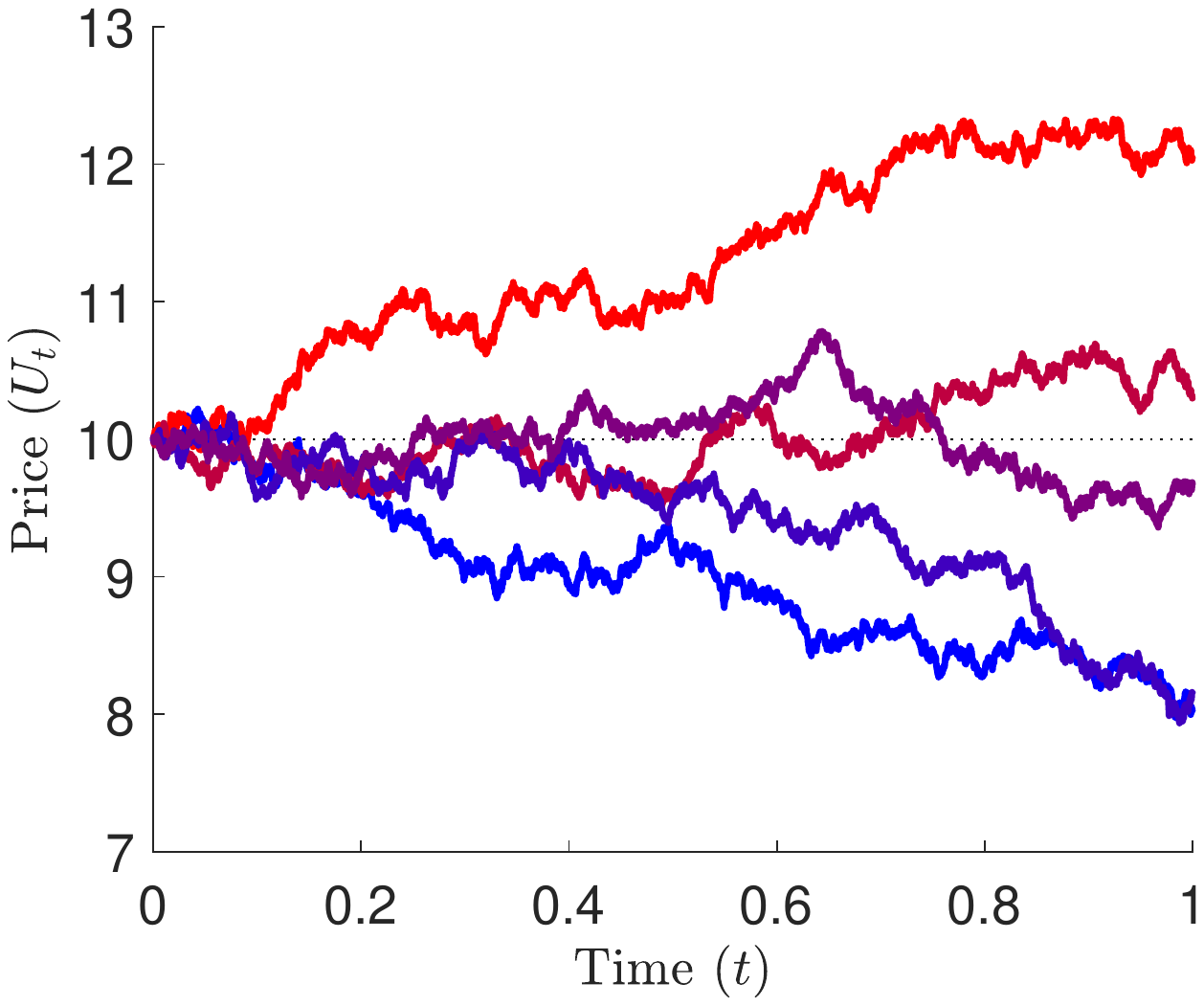}}
	\end{center}
	\vspace{-1em}
	\caption{Agent's optimal inventory position over time for $5$ simulated paths of $U$. In the left panel the agent's inventory position is displayed. The middle panel shows the agent's trading speed. The right panel shows the value of $U_T$. Colors are chosen based on the final value of $U_T$ (larger values are red, smaller values are blue). Other model parameters are identical to those in Figure \ref{fig:inventory_c_1} except $c = 10^{-3}$ and $\gamma = 2\cdot 10^{-3}$. \label{fig:inventory_mix_1}}
\end{figure}

\begin{figure}
	\begin{center}
		{\includegraphics[trim=140 240 140 240, scale=0.48]{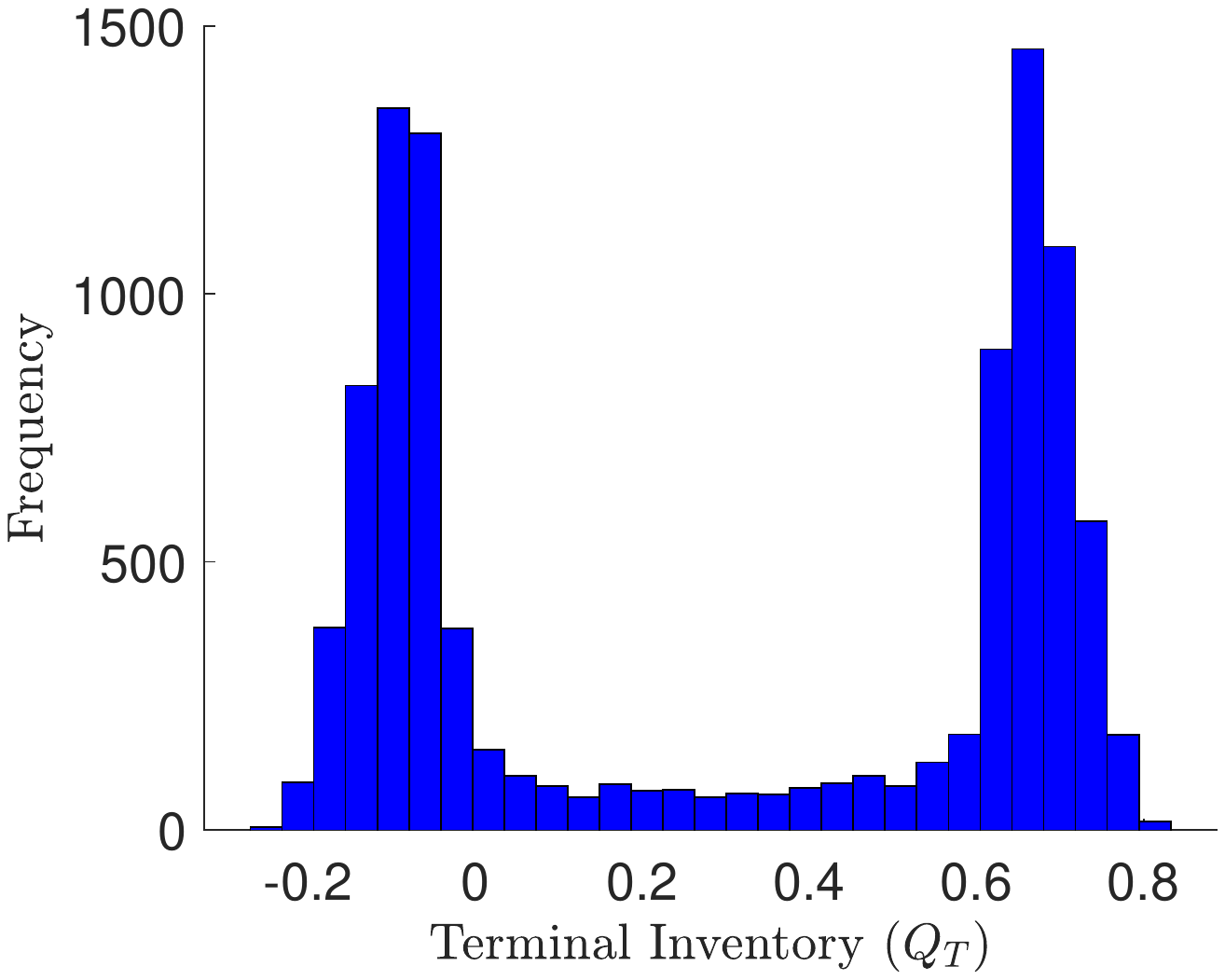}}\hspace{15mm}
		{\includegraphics[trim=140 240 140 240, scale=0.48]{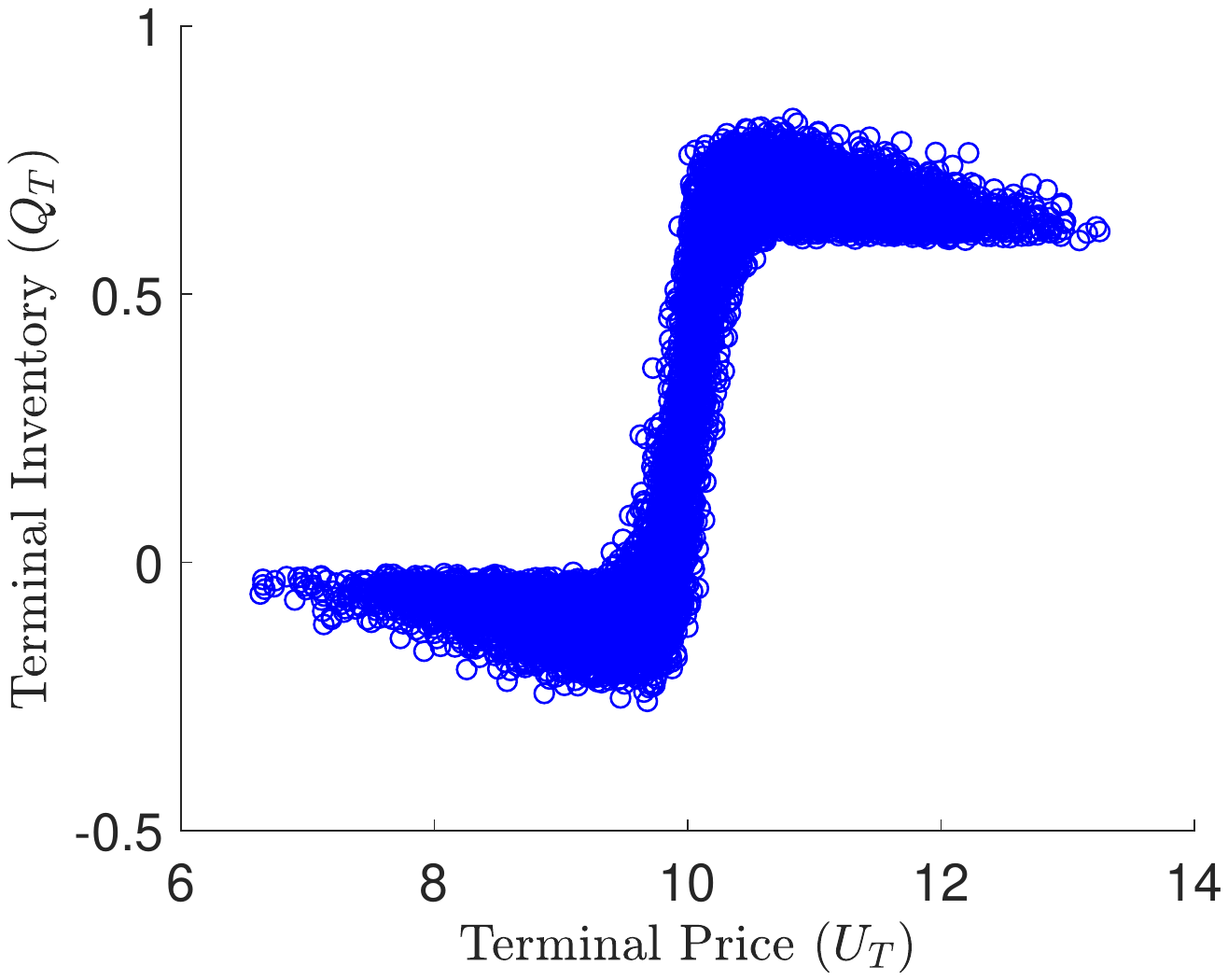}}
	\end{center}
	\vspace{-1em}
	\caption{Distribution of agent's terminal inventory and dependence of terminal inventory on $U_T$. Parameters are identical to those in Figure \ref{fig:inventory_mix_1}. Number of simulations is $M = 10000$. \label{fig:dist_mix_1}}
\end{figure}

The counteracting effects of the two expansion parameters also leads to interesting behavior regarding the distribution of the agent's inventory through time. Many  algorithms that trade off    expected returns and risks or trading penalties have their lowest variance at the endpoints of the trading period (the variance will be zero at time $0$ because the agent knows what their inventory holding is). Low variance at the end of the trading period is generally expected for various reasons, such as the fact that a trading target is acquired or nearly acquired, or because non-zero inventory positions are undesirable over night. Figure \ref{fig:sample_stats} displays the sample mean and standard deviation of the agent's inventory as a function of time.

\begin{figure}
	\begin{center}
		{\includegraphics[trim=140 240 140 240, scale=0.48]{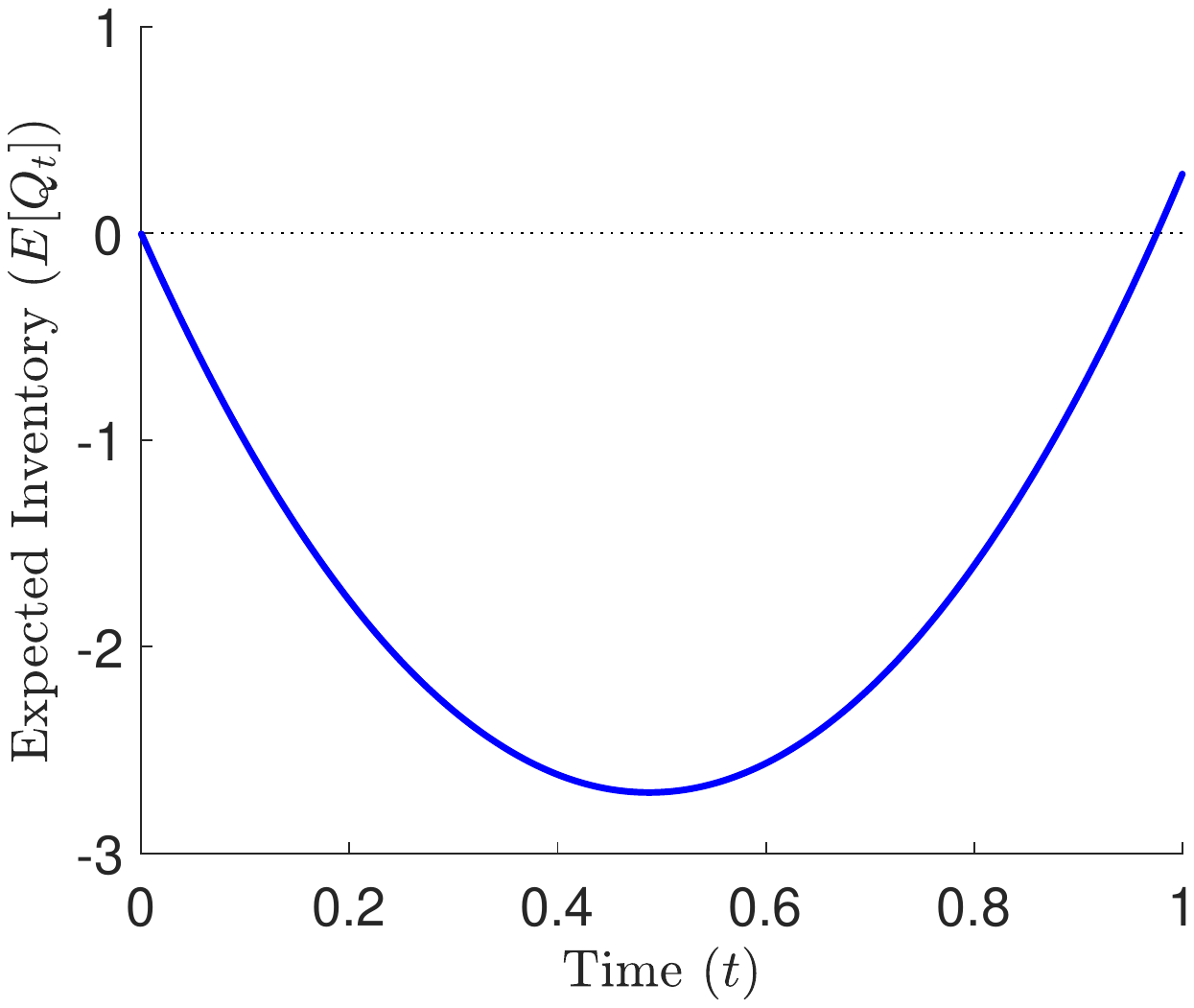}}\hspace{15mm}
		{\includegraphics[trim=140 240 140 240, scale=0.48]{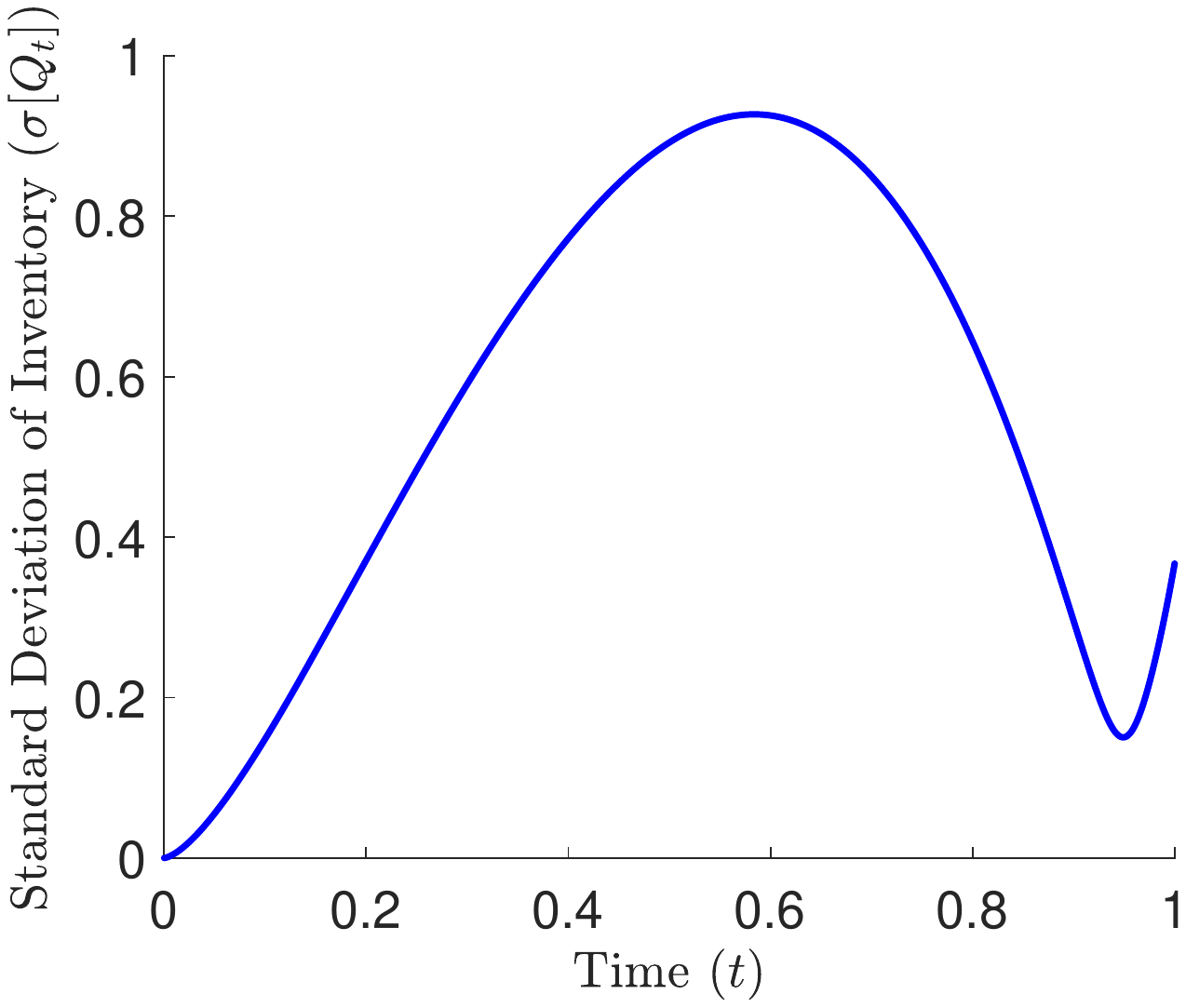}}
	\end{center}
	\vspace{-1em}
	\caption{Sample mean and standard deviation of agent's inventory over the course of the trading period. Parameters are identical to those in Figure \ref{fig:inventory_mix_1}. Number of simulations is $M = 10000$. \label{fig:sample_stats}}
\end{figure}

\section{Conclusions}\label{sec:conclusion}

We solved a problem of an agent who  has exposure to a risk-factor that cannot be directly traded. The agent can trade in an asset which is correlated to the risk-factor to reduce risk exposure. In addition, the agent's trades have an effect on the immediate and future price of the traded asset as well as the future value of the non-tradable risk factor. When the exposure to the factor is linear we solve for the agent's value function and optimal trading strategy in closed-form. This closed-form consists of several terms that illustrate how the agent trades off the risks and rewards of the combination of the positions in the two assets.

When the exposure to the non-tradable risk factor has non-linear dependence we derive an approximation to the agent's value function which holds when the cross-price impact and risk-aversion parameters are small. In addition, an observation about this expansion approximation allows us to assert that the agent has a simple trading strategy (in closed-form) which is also an approximation to the optimal strategy. Given the trading strategy which is optimal when the factor exposure is linear and interpreting the non-linear exposure as a European option written on the non-tradable risk factor, the agent should trade at time $t$ as if she were holding a number of units of the non-tradable risk factor that is equal to the option's delta at time $t$. The parameters of the expansion, cross-price impact and risk-aversion, affect the optimal trading strategy in qualitatively different ways, inducing either long or short positions depending on which effect is stronger.

%\appendix
\section{Proofs}

\section*{Appendix A: Proofs for Section \ref{sec:linear_exposure} (Linear Exposure)}
\label{sec:proofs_linear_exposure}

%\begin{proof}[Proof of Proposition \ref{prop:value_function}]
\subsection{Proof of Proposition \ref{prop:value_function}}
	The form of the terminal conditions and the coefficients of the HJB equation suggest that we make the ansatz $H(t,x,q,S,U) = -e^{-\gamma\,(x + q\,S + \mathfrak{N}\,U + h(t,q))}$. Substitute the expression into the HJB equation to obtain an equation satisfied by $h(t,q)$:
	\begin{align}
		\partial_th + \sup_\nu\biggl\{
		\nu\,\partial_q h - k\,\nu^2 + \mu\, q + b\,q\,\nu - \tfrac{1}{2}\,\sigma^2\,\gamma q^2 + (\beta + c\,\nu)\,\mathfrak{N} - \tfrac{1}{2}\,\eta^2\,\gamma\,\mathfrak{N}^2 - \rho\,\sigma\,\eta\,\gamma\, \mathfrak{N}\,q\biggr\} = 0\,,\label{eqn:ODE_h}
	\end{align}
	subject to terminal condition $h(T,q) = -\alpha\, q^2$. The   supremum is  obtained at
	\begin{align}
		\nu^* &= \tfrac{1}{2\,k}\,\left(\partial_q h + b\,q + c\,\mathfrak{N}\right)\,.
		\label{eqn:nu-feedback}
	\end{align}
	Substitute the optimal control into equation \eqref{eqn:ODE_h} to write the following non-linear PDE:
	\begin{align}
		\partial_th + \mu \,q - \frac{1}{2}\,\sigma^2\,\gamma \,q^2 + \beta\,\mathfrak{N} - \frac{1}{2}\,\eta^2\,\gamma\,\mathfrak{N}^2 - \rho\,\sigma\,\eta\,\gamma \,\mathfrak{N}\,q + \frac{(\partial_qh + b\,q + c\,\mathfrak{N})^2}{4\,k} = 0\,.\label{eqn:ODE_h2}
	\end{align}
	Once again, based on the form of the coefficients and the terminal conditions for $h$, we suggest the following form:
	\begin{align}
		h(t,q) = h_0(t) + h_1(t)\,q + h_2(t)\,q^2\,.\label{eqn:h_form}
	\end{align}
	Substitute this form into equation \eqref{eqn:ODE_h2} and group by like powers of $q$ gives the following system of equations:
	\begin{subequations}
		\addtolength{\jot}{3pt}\label{eqn:ODE_h_system}
		\begin{align}
			h_0'(t) + \beta\,\mathfrak{N}
			- \tfrac{1}{2}\,\eta^2\,\gamma\,\mathfrak{N}^2
			+ \tfrac{1}{4\,k}\,(h_1(t) + c\,\mathfrak{N})^2 &= 0\,,\label{eqn:ODE_h_0}\\
			h_1'(t) + \mu - \rho\,\sigma\,\eta\,\gamma\,\mathfrak{N}
			+ \tfrac{1}{2\,k}\,(h_1(t) + c\,\mathfrak{N})(2\,h_2(t) + b) &= 0\,,\label{eqn:ODE_h_1}\\
			h_2'(t) - \tfrac{1}{2}\,\sigma^2\,\gamma + \tfrac{1}{4\,k}\,(2\,h_2(t) + b)^2 &= 0\,,\label{eqn:ODE_h_2}
		\end{align}
		%\addtolength{\jot}{-3pt}
	\end{subequations}
	subject to the terminal conditions $h_0(T) = 0$, $h_1(T) = 0$, and $h_2(T) = -\alpha$.
	Equation \eqref{eqn:ODE_h_2} is uncoupled and of Riccati type, and may be solved explicitly. One may check that the solution is given by \eqref{eqn:h_2}, it can be substituted into equation \eqref{eqn:ODE_h_1}, and the solution of this equation can be checked to be given by \eqref{eqn:h_1}. \qed
%\end{proof}

%\begin{proof}[Proof of Theorem \ref{prop:optimal_control}]
\subsection{Proof of Theorem \ref{prop:optimal_control}}
	Given the explicit form of the candidate solution in Proposition \ref{prop:value_function}, insert $h$ in equation \eqref{eqn:h_form} into \eqref{eqn:nu-feedback}, so that
	\[
		\nu^*(t,q)= \tfrac{1}{2\,k}\left(\,c\,\mathfrak{N} + h_1(t) + (2\,h_2(t) + b)\,q\right).
	\]
	Assumption \ref{ass:ass} ii) (recall: $2\,\alpha - b > 0$) implies that $h_1$ and $h_2$ are bounded, thus the ODE $dQ_t^{\nu^*}=\nu^*_t\,dt$ has a solution for all $t\in[0,T]$. It is straightforward but tedious to show that the solution is given by \eqref{eqn:optimal_inventory}. The solution $Q_t^{\nu^*}$ is deterministic, therefore it is bounded, and  so is $\nu_t^* = \nu^*(t,Q_t^{\nu^*})$, thus $\int_0^T (\nu_u)^2\,du<+\infty$. Hence, as the solution to the associated HJB equation is classical and the feedback-form strategy is admissible, the strategy is indeed the one we seek, and the solution in Proposition \ref{prop:value_function} is indeed the value function.
	\qed
%\end{proof}

%\begin{proof}[Proof of Proposition \ref{prop:long_horizon}]
\subsection{Proof of Proposition \ref{prop:long_horizon}}
	Substitute $t=\kappa\, T$ into equation \eqref{eqn:optimal_inventory} and perform some elementary algebra to obtain
	\begin{align*}
		Q_{\kappa T}^{\nu^*} &=  \biggl(\frac{\zeta \,k(\phi^- - \phi^+)}{4\,\omega^2} + \frac{c \,\mathfrak{N}}{2}\biggr)\biggl(\frac{e^{\frac{\omega}{k} (\kappa-1)\, T} - e^{-\frac{\omega}{k} (\kappa+1)\, T}}{\phi^+ + \phi^-e^{-\frac{2\,\omega}{k} \,T}}\biggr)\\
		& \hspace{15mm} - \frac{\zeta k}{2\,\omega^2}\biggl(\frac{\phi^+e^{-\frac{\omega}{k} \,\kappa \,T} + \phi^-e^{-\frac{\omega}{k} (2-\kappa)\,T}}{\phi^+ + \phi^-e^{-\frac{2\,\omega}{k}\, T}} - 1\biggr)
+ Q_0\,\biggl(\frac{\phi^+e^{-\frac{\omega}{k} \,\kappa\, T} + \phi^-e^{-\frac{\omega}{k} (2-\kappa)\,T}}{\phi^+ + \phi^-e^{-\frac{2\,\omega}{k} \,T}}\biggr).
	\end{align*}
	Recall that $\omega = \sqrt{\frac{k\,\gamma\,\sigma^2}{2}}$ and $\phi^\pm = \omega \pm \alpha \mp \frac{b}{2}$, and that we assume $2\,\alpha - b > 0$. As we  restrict $\kappa\in(0,1)$, as $T\rightarrow \infty$ the numerator of each fraction above with exponential terms go to zero, and the denominators go to $\omega + \alpha - \frac{b}{2} > 0$. The only remaining term gives
	\begin{align*}
		\lim_{T\rightarrow\infty} Q_{\kappa T}^{\nu^*} &= \frac{\zeta\, k}{2\,\omega^2} = \frac{\mu - \gamma\,\rho\,\sigma\,\eta\,\mathfrak{N}}{\gamma\,\sigma^2}\,,
	\end{align*}
	as desired. Similarly, as $k\downarrow 0$ the numerators approach zero and the denominators approach $\alpha - b/2 > 0$. There is again a single remaining term giving
	\begin{align*}
		\lim_{k\rightarrow 0} Q_{\kappa T}^{\nu^*} &= \lim_{k\rightarrow 0} \frac{\zeta\, k}{2\,\omega^2} = \frac{\mu - \gamma\,\rho\,\sigma\,\eta\,\mathfrak{N}}{\gamma\,\sigma^2}\,.
	\end{align*}
	\qed
%\end{proof}

\section*{Appendix B: Proofs for Section \ref{sec:non_linear_exposure} (Non-Linear Exposure)}
\label{sec:proofs_non_linear_exposure}

{
	Each of the three main proofs in this section (for Theorems \ref{prop:asymptotic_approximation} and \ref{prop:approx_nu}, and Proposition \ref{prop:closed_form}) is broken into multiple parts. The main component of each proof is to perform an approximate verification argument. These proceed by applying Ito's Lemma to a candidate approximation of the value function where the underlying processes are controlled by a candidate approximation of the optimal control. The desired approximation results then amount to bounding the magnitude of the error with respect to optimality and showing that this error tends to zero at the appropriate rate.
	
	The verification in Theorem \ref{prop:asymptotic_approximation} shows that our candidate approximation of the value function is accurate up to second order. The verifications in Theorem \ref{prop:approx_nu} and Proposition \ref{prop:closed_form} show that the candidate approximation is accurate up to second order with respect to the performance criteria of both of our candidate controls. Combining these results means that these performance criteria are also accurate up to second order to the value function.
}
%\begin{proof}[Proof of Proposition \ref{prop:asymptotic_approximation}]
\subsection{Proof of Theorem \ref{prop:asymptotic_approximation}}
	The proof proceeds in two parts. First we substitute the formal expansion of \eqref{eqn:h_expansion} into equation \eqref{eqn:h_PDE} (with $c$ and $\gamma$ replaced by $\theta \,c$ and $\theta\, \gamma$) and group terms according to the zero, first, and second order in $\theta$. Second, we show that this formal second order expansion is valid in the sense that the limit in \eqref{eqn:error_limit} holds by performing a verification argument.

	\underline{Part I} (formal solution): Substituting \eqref{eqn:h_expansion} into \eqref{eqn:h_PDE} and setting terms proportional to $\theta^0$ to vanish gives
	\begin{equation}
		\left\{
		\begin{split}
			\partial_t h_0 + \mu\, q + \beta\, \partial_U h_0
			+ \tfrac{1}{2}\,\eta^2\,\partial_{UU}h_0
			+ \tfrac{1}{4\,k}\,(\partial _q h_0 + b\,q)^2 &= 0\,,\\%[0.5em]
			h_0(T,q,U) &= -\alpha\, q^2 + \psi(U)\,.
		\end{split}
		\right.	\label{eqn:pf-h_0}
	\end{equation}
	It is easily verified that equation \eqref{eqn:pf-h_0} has solution given by
\begin{subequations}
	\begin{align}
		h_0(t,q,U) &= f_0(t) + f_1(t)\,q + f_2(t)\,q^2 + g(t,U)\,,\\
		f_0(t) &= \tfrac{1}{4k}\,\int_t^T\,(f_1(s))^2\,ds\,,\label{eqn:f_0}\\
		f_1(t) &= \frac{\mu\,(T-t)(4\,k+m\,(T-t))}{4\,k + 2\,m\,(T-t)}\,,\\
		f_2(t) &= \frac{-k\,m}{2\,k+m\,(T-t)} - \frac{b}{2}\,,\\
		g(t,U) &= \mathbb{E}[\,\psi(\tilde{U}_T)\,|\,\tilde{U}_t = U\,]\,,\label{eqn:g_pf}\\
		d\tilde{U}_t &= \beta\, dt + \eta \,dZ_t\,.
	\end{align}%
\end{subequations}%
	Similarly, grouping terms proportional to $\theta^1$ gives
	\begin{equation}
		\left\{
		\begin{split}
			c\,\left[\partial_t h_1 + \beta\, \partial_Uh_1 + \tfrac{1}{2}\,\eta^2\,\partial_{UU}h_1 + \tfrac{1}{2\,k}\,(\partial_qh_0 + b\,q)\,(\partial_qh_1 + \partial_U h_0)\right]\hspace{10mm}\\
			+ \gamma\,\left[\partial_t h_2 + \beta \,\partial_Uh_2 + \tfrac{1}{2}\,\eta^2\,\partial_{UU}h_2 + \tfrac{1}{2\,k}\,(\partial_q h_0 + bq)\,\partial_qh_2\right.
			\\
			\left.- \tfrac{1}{2}\,\sigma^2\,q^2 - \rho\,\sigma\,\eta\, q\,\partial_Uh_0 - \tfrac{1}{2}\,\eta^2\,(\partial_U h_0)^2\right]&= 0\,,\\%0.5em
			c\,h_1(T,q,U) + \gamma\,h_2(T,q,U) &= 0\,.
		\end{split}
		\right.
		\label{eqn:pf-h-1}
	\end{equation}
	
	We seek solutions to equation \eqref{eqn:pf-h-1} that do not depend on $c$ or $\gamma$, hence, we set each term in square brackets in equation \eqref{eqn:pf-h-1} to zero independently.
	
	Thus, set the first square bracket in \eqref{eqn:pf-h-1} to zero, write $h_1(t,q,U)$ in the form $h_1(t,q,U) = \lambda_0(t,U) + \lambda_1(t,U)q$, and set $q^0$ and $q^1$ terms to vanish independently, and obtain
	\begin{equation}
		\left\{
		\begin{split}
			\partial_t\lambda_0 + \beta \partial_U\lambda_0 + \tfrac{1}{2}\,\eta^2\,\partial_{UU}\lambda_0 + \tfrac{1}{2\,k}\,f_1(\lambda_1 + \partial_U g) &= 0\,,\\%[0.5em]
			\lambda_0(T,U) &= 0\,,
		\end{split}
		\right.
		\label{eqn:pf_lambda_0}
	\end{equation}
	and
	\begin{equation}
		\left\{
		\begin{split}
			\partial_t\lambda_1 + \beta \partial_U\lambda_1 + \tfrac{1}{2}\,\eta^2\,\partial_{UU}\lambda_1 + \tfrac{1}{2\,k}\,(2\,f_2 + b)\,\lambda_1 + \tfrac{1}{2\,k}\,(2\,f_2 + b)\,\partial_U g &= 0\,,\\%0.5em
			\lambda_1(T,U) &= 0\,,
		\end{split}
		\right.
		\label{eqn:pf_lambda_1T}
	\end{equation}
	where $f_{0,1,2}(t)$ and $g(t,U)$ are given in equations \eqref{eqn:f_0} to \eqref{eqn:g_pf}. By the Feynman-Kac formula, equations \eqref{eqn:pf_lambda_0} and \eqref{eqn:pf_lambda_1T} have solutions given by
	\begin{align}
		\lambda_0(t,U) &= \mathbb{E}\biggl[\int_t^T\frac{f_1(s)}{2\,k}\biggl(\lambda_1(s,\tilde{U}_s)+\partial_U g(s,\tilde{U}_s)\biggr)ds\,\biggl|\,\tilde{U}_t=U\biggr]\,,\label{eqn:lambda_0_pf}\\
		%\lambda_0(t,U) &= \frac{\mu(T-t)^2}{2(2k + m(T-t))}\partial_Ug(t,U)\,,\label{eqn:lambda_0}\\
		\lambda_1(t,U) &= \frac{-m}{2\,k+m\,(T-t)}\,\mathbb{E}\biggl[\int_t^T\partial_U g(s,\tilde{U}_s)\,ds\,\biggl|\,\tilde{U}_t=U\biggr]\,.\label{eqn:lambda_1_pf}
		%\lambda_1(t,U) &= \frac{-m(T-t)}{2k+m(T-t)}\partial_Ug(t,U)\,,\label{eqn:lambda_1}\\
	\end{align}
	
	Next, set the second square bracket of \eqref{eqn:pf-h-1} to zero, write $h_2(t,q,U)$ in form $h_2(t,q,U) = \Lambda_0(t,U) + \Lambda_1(t,U)\,q + \Lambda_2(t)\,q^2$, and set $q^0$, $q^1$, and $q^2$ terms to zero independently, and write
	\begin{equation}
		\left\{
		\begin{split}
			\partial_t\Lambda_0 + \beta\,\partial_U\Lambda_0 + \tfrac{1}{2}\,\eta^2\,\partial_{UU}\Lambda_0 + \tfrac{1}{2\,k}\,f_1\,\Lambda_1 - \tfrac{1}{2}\,\eta^2\,(\partial_U g)^2 &= 0\,,
			\\
			\Lambda_0(T,U) &= 0\,,
		\end{split}
		\right.
		\label{eqn:pf_Lambda_0}
	\end{equation}
	\begin{equation}
		\left\{
		\begin{split}
			\partial_t\Lambda_1 + \beta\,\partial_U\Lambda_1 + \tfrac{1}{2}\,\eta^2\,\partial_{UU}\Lambda_1 + \tfrac{1}{2\,k}\,(2\,f_2 + b)\,\Lambda_1 + \tfrac{1}{k}\,\Lambda_2\,f_1 - \rho\,\sigma\,\eta\,\partial_U g &= 0\,,\\
			\Lambda_1(T,U) &= 0\,,
		\end{split}
		\right.
		\label{eqn:pf_Lambda_1}
	\end{equation}
	\begin{equation}
		\left\{
		\begin{split}
			\partial_t\Lambda_2 + \tfrac{1}{k}\,(2\,f_2 + b)\,\Lambda_2 - \tfrac{1}{2}\,\sigma^2 &= 0\,,\\
			\Lambda_2(T) &= 0\,.
		\end{split}
		\right.
		\label{eqn:pf_Lambda_2}
	\end{equation}
	The solution to ODE \eqref{eqn:pf_Lambda_2} is
	\begin{align}
		\Lambda_2(t) &= -\sigma^2\,(T-t)\,\frac{12\,k^2 + 6\,k\,m\,(T-t) + m^2\,(T-t)^2}{6\,(2\,k+m\,(T-t))^2}\,.\label{eqn:Lambda_2_pf}
	\end{align}
	By the Feynman-Kac formula, equations \eqref{eqn:pf_Lambda_0} and \eqref{eqn:pf_Lambda_1} have solutions
	\begin{align}
		\Lambda_0(t,U) &=
		\tfrac{1}{2\,k}\,
		\mathbb{E}\biggl[\int_t^T \left(f_1(s)\Lambda_1(s,\tilde{U}_s) -k\,\eta^2\,(\partial_Ug(s,\tilde{U}_s))^2\right)\,\biggl|\,\tilde{U}_t=U\biggr]\,,\label{eqn:Lambda_0_pf}\\
		\Lambda_1(t,U) &= \tfrac{1}{k}\,\mathbb{E}\biggl[\int_t^T\frac{2\,k+m\,(T-s)}{2\,k+m\,(T-t)}
		\left(f_1(s)\Lambda_2(s) - k\,\rho\,\sigma\,\eta\,\partial_Ug(s,\tilde{U}_s)\right)\,ds\,\biggl|\,\tilde{U}_t=U\biggr]\,.\label{eqn:Lambda_1_pf}
	\end{align}
	
	Finally, group the terms proportional to $\theta^2$ and obtain
	\begin{equation}
		\left\{
		\begin{split}
			c^2\,\left[\partial_t h_3 + \beta\, \partial_Uh_3 + \tfrac{1}{2}\,\eta^2\,\partial_{UU}h_3 \right.\hspace{85mm}\\
			 \left.+ \tfrac{1}{4\,k}\,(\partial_U h_0 + \partial_q h_1)^2 + \tfrac{1}{2\,k}\,(\partial_qh_0 + b\,q)\,(\partial_uh_1 + \partial_qh_3)\right]\hspace{40mm}\\
			+ c\,\gamma \, \left[\partial_t h_4 + \beta\, \partial_Uh_4 + \tfrac{1}{2}\,\eta^2\,\partial_{UU}h_4 + \tfrac{1}{2\,k}\,(\partial_q h_1 + \partial_U h_0)\,\partial_q h_2 \right.\hspace{20mm}\\
			 \left. + \tfrac{1}{2\,k}\,(\partial_qh_0 + b\,q)\,(\partial_Uh_2 + \partial_q h_4) - \eta^2\partial_Uh_0\,\partial_Uh_1 -\rho\,\sigma\,\eta\,q\,\partial_U h_1    \right]\hspace{5mm}\\
			+ \gamma^2\,\left[\partial_t h_5 + \beta\, \partial_Uh_5 + \tfrac{1}{2}\,\eta^2\,\partial_{UU}h_5 + \tfrac{1}{4\,k}\,(\partial_qh_2)^2 \right.\hspace{22mm}\\
			 \left. + \tfrac{1}{2\,k}\,(\partial_qh_0 + b\,q)\,\partial_qh_5 - \eta^2\partial_Uh_0\,\partial_Uh_2 -\rho\,\sigma\,\eta\,q\,\partial_U h_2   \right] &= 0\,,\\%0.5em
			c^2\,h_3(T,q,U) + c\,\gamma\,h_4(T,q,U) + \gamma^2\,h_5(T,q,U) &= 0\,.
		\end{split}
		\right.
	\label{eqn:pf-h-2}
		\end{equation}
	We seek solutions to \eqref{eqn:pf-h-2} that do not depend on $c$ and $\gamma$, so we set each of the three terms in square brackets equal to zero independently. Make the substitutions
	\begin{subequations}
		\label{eqn:h345}
		\begin{align}
			h_3(t,q,U) &= A_0(t,U) + A_1(t,U)\,q + A_2(t,U)\,q^2\,,\label{eqn:h3_pf}\\
			h_4(t,q,U) &= B_0(t,U) + B_1(t,U)\,q + B_2(t,U)\,q^2\,,\label{eqn:h4_pf}\\
			h_5(t,q,U) &= C_0(t,U) + C_1(t,U)\,q + C_2(t,U)\,q^2\,,\label{eqn:h5_pf}
		\end{align}
	\end{subequations}
	to arrive at a system of PDE's for $A_{0,1,2}$, $B_{0,1,2}$, and $C_{0,1,2}$. In Lemma \ref{lem:ABC} (which appears at the end of this proof) we show that these functions are bounded and continuously differentiable with respect to $U$ with bounded derivatives.
	
	\underline{Part II:} (accuracy of approximation).
	{
		 With $\hat{h}$ given by \eqref{eqn:h_expansion}, define
		\begin{align}
			\hat{H}(t,x,q,S,U;\theta \,c ,\theta\,\gamma) &= -e^{-\theta\gamma (x + qS + \hat{h}(t,q,U;\theta\, c, \theta \,\gamma))}\,.\label{eqn:H_hat}
		\end{align}
		Then the desired limit in \eqref{eqn:error_limit} is equivalent to
		\begin{align}
			H_\psi(t,x,q,S,U;\theta \,c, \theta\, \gamma) &= \hat{H}(t,x,q,S,U;\theta\, c, \theta\,\gamma) + o(\theta^3)\,,\label{eqn:H_error}
		\end{align}
		where the additional power of $\theta$ follows from a Taylor expansion of the exponential function and noting the additional factor of $\theta$ that appears in the exponential of \eqref{eqn:H_hat}. For simplicity, we prove the approximation in \eqref{eqn:H_error} holds for $t=0$ with initial states given by $x$, $q$, $S$, and $U$. The case of $t\neq0$ follows similarly.
		
		Henceforth, consider the initial states $x$, $q$, $S$, and $U$ to be fixed, and take $\theta\in(0,\theta^*)$, $\epsilon\in(0,\epsilon^*)$  where $\theta^*,\epsilon^*$ are as in Assumption \ref{ass:ass} iii). Further, let $\nu^{\theta,\epsilon}$ be an admissible control which is $\epsilon\,\theta^3$-optimal, specifically such that
		\begin{align}
			H^{\nu^{\theta,\epsilon}}(0,x,q,S,U;\theta\, c, \theta\,\gamma)  + \epsilon\,\theta^3 \geq H_\psi(0,x,q,S,U;\theta\, c, \theta\, \gamma)\,.\label{eqn:e_optimal}
		\end{align}
		Applying Ito's Lemma to the process $G_t = \hat{H}(t,X^{\nu^{\theta,\epsilon}}_t,Q^{\nu^{\theta,\epsilon}}_t,S^{\nu^{\theta,\epsilon}}_t,U^{\nu^{\theta,\epsilon}}_t;\theta\,c,\theta\,\gamma)$ yields
		\begin{equation}
			\begin{split}
			G_T - G_0 =& \int_0^T (\partial_t + \mathcal{L}^{\nu^{\theta,\epsilon}})
			\hat{H}(t,X^{\nu^{\theta,\epsilon}}_t,Q^{\nu^{\theta,\epsilon}}_t,
			S^{\nu^{\theta,\epsilon}}_t,
			U^{\nu^{\theta,\epsilon}}_t;\theta\,c,\theta\,\gamma)\,dt
			\\
			&-\theta\,\gamma \, \sigma\int_0^T \hat{H} (t, X_t^{\nu^{\theta,\epsilon}}, Q_t^{\nu^{\theta,\epsilon}}, S_t^{\nu^{\theta,\epsilon}}, U_t^{\nu^{\theta,\epsilon}}; \theta\,c,\theta\,\gamma)\,Q_t^{\nu^{\theta,\epsilon}}\, dW_t
			\\
			%&\hspace{30mm}
			&- \theta\,\gamma \, \eta  \int_0^T \hat{H} (t, X_t^{\nu^{\theta,\epsilon}}, Q_t^{\nu^{\theta,\epsilon}}, S_t^{\nu^{\theta,\epsilon}}, U_t^{\nu^{\theta,\epsilon}}; \theta\,c,\theta\,\gamma)\,\partial_U\hat{h}(t,Q_t^{\nu^{\theta,\epsilon}},U_t^{\nu^{\theta,\epsilon}};\theta\, c, \theta\,\gamma)\,dZ_t\,,
			  \end{split}
			  \label{eqn:ito}
		\end{equation}
	{	where the differential operator $\mathcal{L}^\nu$ is given by
		\begin{align*}
			\mathcal{L}^\nu &= \nu \, \partial_q - (S+k\,\nu)\,\nu\,\partial_x + (\mu + b\,\nu)\,\partial_S + \frac{1}{2}\,\sigma^2\,\partial_{SS} + (\beta+\theta\,c\,\nu)\,\partial_U + \frac{1}{2}\,\eta^2\,\partial_{UU} + \rho\,\sigma\,\eta\,\partial_{SU}\,.
		\end{align*}
	}
		
		Inspection of $\partial_U\hat{h}(t,q,U;\theta\,c,\theta\,\gamma)$ shows that it is a polynomial with respect to $q$ of degree 2 with coefficients that are bounded with respect to $(t,U)$ due to Lemma \ref{lem:future_delta}. 
		
		Next, we apply the uniform bound in \eqref{eqn:uniform_bound_assumption} from Assumption \ref{ass:ass} iii) to show that both stochastic integrals have expectation zero for sufficiently small $\theta$. There is a sufficiently large $N_1$ independent of $\theta\in(0,\theta^*)$ such that
		\begin{align*}
		\left|
		\hat{H}^2(t,x,q,S,U;\theta\,c,\theta\,\gamma)\,q^2
		\right|
		& \leq N_1 \,e^{\theta \, \gamma\, N_1\, (|x| + |q\,S| + |q| + q^2 + |U|)}\,,\\
		\left|
		\hat{H}^2(t,x,q,S,U;\theta\,c,\theta\,\gamma)(\partial_U\hat{h}(t,q,U;\theta\, c, \theta\,\gamma))^2
		\right|
		& \leq N_1 \,e^{\theta \, \gamma\, N_1\, (|x| + |q\,S| + |q| + q^2 + |U|)}\,.
		\end{align*}
		Therefore, by Assumption \ref{ass:ass} iii), if $\theta < \frac{D}{\gamma \,N_1}$ then the integrands in both stochastic integrals in \eqref{eqn:ito} are square-integrable over $[0,T]\times\Omega$ and therefore have zero expectation. If $\theta^* > \frac{D}{\gamma \,N_1}$, then henceforth we further restrict $\theta\in(0,\frac{D}{\gamma\,N_1})$.
		
Given the explicit form of $\hat{H}$, we obtain the bound
		\begin{align}
			(\partial_t + \mathcal{L}^{\nu^{\theta,\epsilon}})\hat{H}(t,x,q,S,U;\theta\,c,\theta\,\gamma) & \leq \sup_\nu(\partial_t + \mathcal{L}^\nu)\hat{H}(t,x,q,S,U;\theta\,c,\theta\,\gamma)\label{eqn:LHhatSup}\\
			&= -\theta\,\gamma\,\hat{H}(t,x,q,S,U;\theta\,c,\theta\,\gamma) \, \sum_{i=3}^6 \theta^i\,P_i(t,q,U)\,.\label{eqn:LHhatInequality}
		\end{align}
		{The supremum in \eqref{eqn:LHhatSup} is attained at
		\begin{align*}
			\nu^\dagger &= \frac{\partial_q\hat{h} + \theta\, c\, \partial_U\hat{h} + b\,q}{2\,k}\,,
		\end{align*}
		which after direct substitution and some tedious but straightforward computations results in \eqref{eqn:LHhatInequality},} where, by Lemma \ref{lem:future_delta}, each $P_{i}(t,q,U)$, $i\in\{3,4,5,6\}$, is a polynomial with respect to $q$ of degree at most four with coefficients that are bounded with respect to $t$ and $U$ (full expressions appear in \eqref{eqn:P3456} in Appendix C). 
Taking expectations in \eqref{eqn:ito}, substituting the definition of $G_t$, and using the inequality \eqref{eqn:LHhatInequality}, results in the inequalities

\begin{align*}
& \mathbb{E}\biggl[\int_0^T -\theta\,\gamma\,\hat{H}(t, X_t^{\nu^{\theta,\epsilon}}, Q_t^{\nu^{\theta,\epsilon}}, S_t^{\nu^{\theta,\epsilon}}, U_t^{\nu^{\theta,\epsilon}}; \theta\,c,\theta\,\gamma) \,\sum_{i=3}^6 \theta^i\,P_i(t,Q_t^{\nu^{\theta,\epsilon}},U_t^{\nu^{\theta,\epsilon}})\,dt\biggr]
\\
& \qquad \ge \mathbb{E}[\hat{H}(T,X^{\nu^{\theta,\epsilon}}_T,
Q^{\nu^{\theta,\epsilon}}_T,S^{\nu^{\theta,\epsilon}}_T,
U^{\nu^{\theta,\epsilon}}_T;\theta\,c,\theta\,\gamma)] - \hat{H}(0,x,q,S,U;\theta\,c,\theta\,\gamma)
\\
& \qquad =
H^{\nu^{\theta,\epsilon}}(0,x,q,S,U;\theta\,c,\theta\,\gamma) - \hat{H}(0,x,q,S,U;\theta\,c,\theta\,\gamma) \,.
\end{align*}

Rearrange and recall that $\nu^{\theta,\epsilon}$ is $\epsilon\,\theta^3$-optimal so that we have
\begin{multline}
\frac{1}{\theta^3}
\left(H_\psi(0,x,q,S,U;\theta\,c,\theta\,\gamma) - \hat{H}(0,x,q,S,U;\theta\,c,\theta\,\gamma)\right)
\\
\leq
\epsilon + \mathbb{E}\biggl[\int_0^T -\theta\,\gamma\,\hat{H}(t, X_t^{\nu^{\theta,\epsilon}}, Q_t^{\nu^{\theta,\epsilon}}, S_t^{\nu^{\theta,\epsilon}}, U_t^{\nu^{\theta,\epsilon}}; \theta\,c,\theta\,\gamma) \,\sum_{i=3}^6 \theta^{i-3}\,P_i(t,Q_t^{\nu^{\theta,\epsilon}},
U_t^{\nu^{\theta,\epsilon}})\,dt\biggr]\,.
\end{multline}
We again apply the uniform bound in \eqref{eqn:uniform_bound_assumption} from Assumption \ref{ass:ass} iii). By construction, $\hat{h}$ has at most linear growth in $U$. Moreover, the zeroth order, linear, and quadratic dependence on $q$ appear with bounded coefficients. Next, as each $P_i$ is at most degree four in $q$, with bounded coefficients, there is a sufficiently large $N_2$, independent of $\theta\in(0,\theta^*)$, such that
		\begin{align*}
\left|
\hat{H}(t,x,q,S,U;\theta\,c,\theta\,\gamma)\sum_{i=3}^6 \theta^{i-3}\,P_i(t,Q_t^{\nu^{\theta,\epsilon}},U_t^{\nu^{\theta,\epsilon}})
\right|
& \leq N_2 \,e^{\theta \, \gamma\, N_2\, (|x| + |q\,S| + |q| + q^2 + |U|)}\,.
		\end{align*}
If $\theta^* > \frac{D}{\gamma\, N_2}$ and $N_2>N_1$, then further restrict $\theta\in(0,\frac{D}{\gamma\,N_2})$, and as $\epsilon\in(0,\epsilon^*)$ and $\theta<\theta^*$, the uniform bound in \eqref{eqn:uniform_bound_assumption} applies and hence
		\begin{align}
			\tfrac{1}{\theta^3}\left|H_\psi(0,x,q,S,U;\theta\,c,\theta\,\gamma) - \hat{H}(0,x,q,S,U;\theta\,c,\theta\,\gamma)\right|
\leq \epsilon + \theta\,\gamma\,N_2\,C\,.
		\end{align}
Finally, as $\epsilon\in(0,\epsilon^*)$ is arbitrary, we have
\begin{align}
    \lim_{\theta\downarrow 0}
    \tfrac{1}{\theta^3}\left|H_\psi(0,x,q,S,U;\theta\,c,\theta\,\gamma) - \hat{H}(0,x,q,S,U;\theta\,c,\theta\,\gamma)\right| = 0\,,
\end{align}
which is the desired limit.
	}
		\qed
%\end{proof}

\begin{lemma}\label{lem:ABC}
	The functions $A_{0,1,2}$, $B_{0,1,2}$, and $C_{0,1,2}$ are bounded and continuously differentiable with respect to $U$ with bounded derivatives.
\end{lemma}
\begin{proof}
	Let $\mathcal{L} = \beta\,\partial_U + \tfrac{1}{2}\,\eta^2\,\partial_{UU}$. Upon substituting \eqref{eqn:h345} into \eqref{eqn:pf-h-2}, the functions $A_{0,1,2}$, $B_{0,1,2}$, and $C_{0,1,2}$ satisfy the following systems of PDE's:
	\begin{equation}
	\left\{
	\begin{split}
	\partial_t A_0 + \mathcal{L}A_0 + \tfrac{1}{4\,k}\,(\lambda_1 + \partial_Ug)^2 + \tfrac{1}{2\,k}\,f_1\,(\partial_U\lambda_0 + A_1) &= 0\,,\\
	\partial_t A_1 + \mathcal{L}A_1 + \tfrac{1}{2\,k}\,(2f_2 + b)\,(\partial_U\lambda_0 + A_1) + \tfrac{1}{2\,k}\,f_1\,(\partial_U\lambda_1 + 2A_2) &= 0\,,\\
	\partial_t A_2 + \mathcal{L}A_2 + \tfrac{1}{2\,k}\,(2f_2 + b)\,(\partial_U\lambda_1 + 2A_2) &= 0\,,\\%0.5em
	A_{0,1,2}(T,U) &= 0\,,
	\end{split}
	\right.
	\end{equation}
	\begin{equation}
	\left\{
	\begin{split}
	\partial_t B_0 + \mathcal{L}B_0 + \tfrac{1}{2\,k}\,f_1\,(\partial_U\lambda_0 + B_1) + \tfrac{1}{2\,k}\,\Lambda_1\,(\lambda_1 + \partial_U g) - \eta^2\partial_U\lambda_0\partial_U g  &= 0\,,\\
	\partial_t B_1 + \mathcal{L}B_1 + \tfrac{1}{2\,k}\,f_1\,(\partial_U\Lambda_1 + 2B_2) + \tfrac{1}{k}\,\Lambda_1\,(\lambda_1+\partial_U g) \hspace{20mm}\\
	+ \tfrac{1}{2\,k}\,(2f_2 + b)\,(\partial_U\Lambda_0 + B_1) - \eta^2\partial_U\lambda_1\partial_Ug -\rho\sigma\eta\partial_U\lambda_0 &= 0\,,\\
	\partial_t B_2 + \mathcal{L}B_2 + \tfrac{1}{2\,k}\,(2f_2 + b)\,(\partial_U\Lambda_1 + 2B_2) + \tfrac{1}{2\,k}\,f_1\,\partial_U\Lambda_1 - \rho\sigma\eta\partial_U\lambda_1 &= 0\,,\\%0.5em
	B_{0,1,2}(T,U) &= 0\,,
	\end{split}
	\right.
	\end{equation}
	\begin{equation}
	\left\{
	\begin{split}
	\partial_t C_0 + \mathcal{L}C_0 + \tfrac{1}{4\,k}\,(\Lambda_1)^2 + \tfrac{1}{2\,k}\,f_1C_1 - \eta^2\partial_U\Lambda_0\partial_U g &= 0\,,\\
	\partial_t C_1 + \mathcal{L}C_1 + \tfrac{1}{k}\,\Lambda_1\,\Lambda_2 + \tfrac{1}{k}\,f_1\,C_2 + \tfrac{1}{2\,k}\,(2f_2+b)\,C_1 - \eta^2\partial_U\Lambda_1\partial_U g - \rho\sigma\eta \partial_U\Lambda_0 &= 0\,,\\
	\partial_t C_2 + \mathcal{L}C_2 + \tfrac{1}{k}\,(\Lambda_2)^2 + \tfrac{1}{k}\,(2f_2+b)\,C_2 - \rho\sigma\eta\partial_U\Lambda_1 &= 0\,,\\%0.5em
	C_{0,1,2}(T,U) &= 0\,.
	\end{split}
	\right.
	\end{equation}
	Inspection shows that within each of the three systems, the coupling is only in one direction so the equations may be solved one by one. We also see that each individual equation takes the form
	\begin{equation}\label{eqn:pf-w}
	\partial_t w + \mathcal{L} w + F + Gw = 0\qquad\text{and}\qquad
	w(T,U) = 0\,.
	\end{equation}
	By the Feynman-Kac formula the solution for $w$ is
	\begin{align*}
	w(t,U) = \mathbb{E}\biggl[\int_t^T e^{\int_t^s G(r,\tilde{U}_r)dr}F(s,\tilde{U}_s) \biggl|\tilde{U}_t = U\biggr]\,,\qquad\text{where}\qquad
	d\tilde{U}_t = \beta\, dt + \eta \, dZ_t\,.
	\end{align*}
	The forcing term $F$ in each equation is  bounded and continuously differentiable with respect to $U$ 	because the functions $f_{0,1,2}$, $\lambda_{0,1}$, $\Lambda_{0,1,2}$, and $\partial_U g$ are bounded and continuously differentiable with respect to $U$. In addition, inspection shows that each discount term $G$ is bounded and a function only of $t$. Therefore each $A_{0,1,2}$, $B_{0,1,2}$, and $C_{0,1,2}$ is bounded, and continuously differentiable with respect to $U$ by Lemma \ref{lem:future_delta}. \qed
\end{proof}

%\begin{proof}[Proof of Proposition \ref{prop:approx_nu}]
\subsection{Proof of Theorem \ref{prop:approx_nu}}
	{
		
Fix $\theta_0 > 0$ and take $\theta\in(0,\theta_0)$. Next, consider the inventory and non-tradable risk factor path when the agent follows the conjectured approximate strategy, specifically such that
		\begin{subequations}
			\label{eqn:approx-path-U}
			\begin{align}
				dQ_t^{\hat{\nu}} &= \hat{\nu}\left(t,Q_t^{\hat{\nu}},U_t^{\hat{\nu}}\right)\,dt\,,\\
				dU_t^{\hat{\nu}} &= \left(\beta + c\,\hat{\nu}\left(t,Q_t^{\hat{\nu}},U_t^{\hat{\nu}}\right)\right)\,dt + \eta\, dZ_t\,.
			\end{align}
		\end{subequations}
By Lemma \ref{lem:future_delta}, the function $\hat{\nu}$ may be written as
		\begin{align}
			\hat{\nu}(t,q,U) &= F_1(t) + F_2(t;\theta)\,q + F_3(t;\theta)\,\partial_Ug(t,U)\,,\label{eqn:pf_nu_hat}
		\end{align}
with $\partial_Ug(t,U)$ and $\partial_{UU}g(t,U)$ bounded, therefore $\hat{\nu}(t,q,U)$ is Lipschitz with linear growth in the variables $q$ and $U$. Thus, the SDEs \eqref{eqn:approx-path-U} have a unique strong solution (see \cite{karatzas2012brownian} Theorem 5.2.9). Moreover, choose the linear growth coefficient uniformly with respect to $\theta\in(0,\theta_0)$, so that
		\begin{align*}
			\mathbb{E}\left[\left(Q_t^{\hat{\nu}}\right)^2 + \left(U_t^{\hat{\nu}}\right)^2\right] \leq C\,e^{Ct}\,, \quad \forall\,t \in[0,T]\,,
		\end{align*}
		for some constant $C$. Therefore, by Fubini's Theorem, we have $\mathbb{E}\left[\int_0^T\hat{\nu}^2_u\,du\right] < \infty$.
		
To show that $\hat{\nu}$ is asymptotically approximately optimal, we proceed with a verification argument while keeping track of the magnitude of the error with respect to optimization, analogous to the proof of Theorem \ref{prop:asymptotic_approximation}. We also remark that as
		\begin{align*}
			H_\psi(t,x,q,S,U;\theta\, c, \theta\, \gamma) &= -e^{-\theta\, \gamma(x + q\,S + h_\psi(t,q,U;\theta\, c, \theta\,\gamma))}\,,\\
			H^{\hat{\nu}}(t,x,q,S,U;\theta\, c, \theta\, \gamma) &= -e^{-\theta\, \gamma\,(x + q\,S + h^{\hat{\nu}}(t,q,U;\theta\, c, \theta\,\gamma))}\,,
		\end{align*}
our desired approximation result is equivalent to
\begin{align}
			H_\psi(t,x,q,S,U;\theta\, c, \theta\, \gamma) &= H^{\hat{\nu}}(t,x,q,S,U;\theta\, c, \theta\, \gamma) + o(\theta^3)\,,\label{eqn:pf_nu_hat_H_approx}
		\end{align}
which follows from a Taylor expansion of the exponential function.
		
We prove the accuracy result at $t = 0$ with given initial states $x$, $q$, $S$, and $U$, which we henceforth consider to be fixed. The general result for $t\neq 0$ follows similarly.
		
Given the control $\hat{\nu}$, and the resulting state processes $X_t^{\hat{\nu}}$, $Q_t^{\hat{\nu}}$, $S_t^{\hat{\nu}}$, and $U_t^{\hat{\nu}}$, define the process $(G_t)_{t\in[0,T]}$ where
		\begin{equation*}
			G_t = \hat{H}(t, X_t^{\hat{\nu}}, Q_t^{\hat{\nu}}, S_t^{\hat{\nu}}, U_t^{\hat{\nu}}; \theta\, c, \theta\, \gamma)\,,
\quad \text{and} \quad
			\hat{H}(t,x,q,S,U;\theta\, c, \theta\, \gamma) = -e^{-\theta \,\gamma \left(x + q\,S + \hat{h}(t,q,U;\theta\, c, \theta\,\gamma\right)}.
		\end{equation*}		
Here, $\hat{h}$ is the approximation of $h_\psi$ given in Theorem \ref{prop:asymptotic_approximation} Equation \eqref{eqn:h_expansion}. Applying Ito's Lemma to $G$ gives
		\begin{align}
G_T - G_0 =& \int_0^T (\partial_t + \mathcal{L}^{\hat{\nu}}) \hat{H} (t, X_t^{\hat{\nu}}, Q_t^{\hat{\nu}}, S_t^{\hat{\nu}}, U_t^{\hat{\nu}}; \theta\, c, \theta\, \gamma)\, dt
\nonumber
\\
&\quad -\theta\,\gamma \, \sigma \int_0^T \hat{H} (t, X_t^{\hat{\nu}}, Q_t^{\hat{\nu}}, S_t^{\hat{\nu}}, U_t^{\hat{\nu}}; \theta\, c, \theta\, \gamma)\,Q_t^{\hat{\nu}}\, dW_t
\nonumber
\\
&\quad\quad- \theta\,\gamma \, \eta \int_0^T \hat{H} (t, X_t^{\hat{\nu}}, Q_t^{\hat{\nu}}, S_t^{\hat{\nu}}, U_t^{\hat{\nu}}; \theta \,c, \theta\, \gamma)\,\partial_U\hat{h}(t,Q_t^{\hat{\nu}},U_t^{\hat{\nu}};\theta\,c,\theta\,\gamma)\,dZ_t\nonumber\\
\begin{split}
 =&
-\theta\,\gamma\int_0^T \hat{H} (t, X_t^{\hat{\nu}}, Q_t^{\hat{\nu}}, S_t^{\hat{\nu}}, U_t^{\hat{\nu}}; \theta\, c, \theta\, \gamma) \biggl(\sum_{i=3}^5 \theta^i \, M_i(t,Q_t^{\hat{\nu}},U_t^{\hat{\nu}})\biggr)\,dt
\\
&\quad  -\theta\,\gamma \,\sigma\int_0^T \hat{H} (t, X_t^{\hat{\nu}}, Q_t^{\hat{\nu}}, S_t^{\hat{\nu}}, U_t^{\hat{\nu}}; \theta\, c, \theta \,\gamma)\,Q_t^{\hat{\nu}}\, dW_t\\
&\quad\quad- \theta\,\gamma \,\eta\int_0^T \hat{H} (t, X_t^{\hat{\nu}}, Q_t^{\hat{\nu}}, S_t^{\hat{\nu}}, U_t^{\hat{\nu}}; \theta \,c, \theta \,\gamma)\,\partial_U\hat{h}(t,Q_t^{\hat{\nu}},U_t^{\hat{\nu}};\theta\,c,\theta\,\gamma)\,dZ_t\,,
			\end{split}\label{eqn:pf_ito}
		\end{align}
		where each $M_{i}(t,q,U)$, $i\in\{3,4,5\}$, is a polynomial in $q$ of degree at most four with coefficients that are uniformly bounded functions of $t$ and $U$ (see \eqref{eqn:M345} in Appendix C for the explicit expressions).
		
		We proceed to show that for $\theta\in(0,\theta_0)$ both stochastic integrals have zero expectation and that Fubini's Theorem may be applied to the expectation of the Riemann integral. First, we construct appropriate bounds on the underlying processes.

The linear growth conditions of $\hat{\nu}$ and boundedness of $\partial_Ug$ implies
		\begin{align*}
			\underline{\nu}(t,q) \leq \hat{\nu}(t,q,U) \leq \overline{\nu}(t,q)\,,\quad\text{where}\quad \overline{\nu}(t,q) = C_1\,(1+|q|) \quad\text{and}\quad \underline{\nu}(t,q) = -\overline{\nu}(t,q)
		\end{align*}
		for some constant $C_1>0$. In addition, the processes $(Q_t^{\overline{\nu}})_{t\in[0,T]}$ and $(Q_t^{\underline{\nu}})_{t\in[0,T]}$ are deterministic and satisfy
		\begin{align*}
			Q_t^{\underline{\nu}} \leq Q_t^{\hat{\nu}} \leq Q_t^{\overline{\nu}}\,.
		\end{align*}
Similarly, there exists processes $(S_t^{\overline{\nu}})_{t\in[0,T]}$, $(S_t^{\underline{\nu}})_{t\in[0,T]}$, $(U_t^{\overline{\nu}})_{t\in[0,T]}$, and $(U_t^{\underline{\nu}})_{t\in[0,T]}$ such that
		\begin{equation*}
			S_t^{\underline{\nu}} \leq S_t^{\hat{\nu}} \leq S_t^{\overline{\nu}}\qquad\text{and}\qquad
			U_t^{\underline{\nu}} \leq U_t^{\hat{\nu}} \leq U_t^{\overline{\nu}}\,,
		\end{equation*}
almost surely (see \cite{karatzas2012brownian} Proposition 5.2.18). Therefore, there exists $C_2>0$ and $C_3>0$ such that
\begin{equation*}
			|S_t^{\hat{\nu}}| \leq C_2\,\biggl(1+\max_{0\leq t\leq T}\{|W_t|\}\biggr)\qquad\text{and}\qquad
			|U_t^{\hat{\nu}}| \leq C_3\,\biggl(1+\max_{0\leq t\leq T}\{|Z_t|\}\biggr)\,.\\
		\end{equation*}

Next, define $M_W = \max_{0\leq t\leq T}\{|W_t|\}$ and $M_Z = \max_{0\leq t\leq T}\{|Z_t|\}$. These bounds provides the following bounds for $X_t^{\hat{\nu}}$
		 \begin{align*}
		 	|X_t^{\hat{\nu}}| & \leq |x| + \int_0^t |S_t^{\hat{\nu}} + k\,\hat{\nu}(s,Q_s^{\hat{\nu}},U_s^{\hat{\nu}})|\,|\hat{\nu}(s,Q_s^{\hat{\nu}},U_s^{\hat{\nu}})| \,ds\\
		 	& \leq |x| + \int_0^T |S_t^{\hat{\nu}}|\,|\hat{\nu}(s,Q_s^{\hat{\nu}},U_s^{\hat{\nu}})|\,ds + k\,\int_0^T |\hat{\nu}(s,Q_s^{\hat{\nu}},U_s^{\hat{\nu}})|^2\, ds\\
		 	& \leq |x| + T\,C_2\,C_4\,\biggl(1+M_W\biggr) + k\,T\,C_4^2\,,
		 \end{align*}
		 where $C_4$ is a constant.

The uniform bounds on $\partial_Ug(t,U)$ and $\partial_{UU}g(t,U)$ implies $\hat{h}$ has at most linear growth in $U$ and hence
		 \begin{align*}
		 	|\hat{h}(t,Q_t^{\hat{\nu}},U_t^{\hat{\nu}})| &\leq C_5 \,(1 + M_Z) \qquad\text{and}\qquad
		 	|\partial_U\hat{h}(t,Q_t^{\hat{\nu}},U_t^{\hat{\nu}})| \leq C_6\,,
		 \end{align*}
		 where $C_5$ and $C_6$ are constants.

Applying the above bounds together provides
		 \begin{align}
			e^{-\theta_0\, \gamma \,C_7\,(1 + M_W + M_Z)} \leq |\hat{H}(t, X_t^{\hat{\nu}}, Q_t^{\hat{\nu}}, S_t^{\hat{\nu}}, U_t^{\hat{\nu}}; \theta \,c, \theta \,\gamma)| \leq e^{\theta_0 \, \gamma\, C_7\,(1 + M_W + M_Z)}\,.\label{eqn:pf_H_hat_bound}
		 \end{align}
We may choose the constants $C_i$ independent of $\theta\in(0,\theta_0)$ and therefore
\begin{align*}
		 	\hat{H}^2(t, X_t^{\hat{\nu}}, Q_t^{\hat{\nu}}, S_t^{\hat{\nu}}, U_t^{\hat{\nu}}; \theta \,c, \theta\, \gamma)(Q_t^{\hat{\nu}})^2 &\leq C_8\,e^{2\,\theta_0 \,\gamma \,C_7\,(1 + M_W + M_Z)}\,,\\
		 	\hat{H}^2(t, X_t^{\hat{\nu}}, Q_t^{\hat{\nu}}, S_t^{\hat{\nu}}, U_t^{\hat{\nu}}; \theta \,c, \theta\, \gamma)\,(\partial_U\hat{h}(t,Q_t^{\hat{\nu}},U_t^{\hat{\nu}};\theta\,c,\theta\,\gamma))^2 &\leq C^2_6\,e^{2\,\theta_0 \,\gamma\, C_7\,(1 + M_W + M_Z)}\,,
		 \end{align*}
		 where $C_8 = \max\{(Q^{\overline{\nu}}_T)^2\,,(Q^{\underline{\nu}}_T)^2\}$. As right-hand sides of both inequalities are integrable over $[0,T]\times\Omega$, the stochastic integrals in \eqref{eqn:pf_ito} have zero expectation.

Next, as noted above, $M_i$, $i\in\{3,4,5\}$, is polynomial in $q$ of degree at most four with coefficients that are uniformly bounded functions of $t$ and $U$. Hence,
		 \begin{align}
		 	\left|\sum_{i=3}^5 \theta^i \, M_i(t,Q_t^{\hat{\nu}},U_t^{\hat{\nu}})\right| &\leq \theta^3 C_9\,,
\label{eqn:C9}
		 \end{align}
where $C_9$ is a constant which does not depend on $\theta\in(0,\theta_0)$. This bound, along with \eqref{eqn:pf_H_hat_bound}, allows us to apply Fubini's Theorem to the Riemann integral in \eqref{eqn:pf_ito}. Putting this together with the result that stochastic integrals on the rhs of \eqref{eqn:pf_ito} have zero expectation, we have
		 \begin{multline}
				 \mathbb{E}[H^{\hat{\nu}}(T, X_T^{\hat{\nu}}, Q_T^{\hat{\nu}}, S_T^{\hat{\nu}}, U_T^{\hat{\nu}};\theta\,c,\theta\,\gamma)] - \hat{H}(0,x,q,S,U;\theta\,c,\theta\,\gamma)
 \\
				  = - \theta\,\mathbb{E}\biggl[\gamma\,\int_0^T \hat{H} (t, X_t^{\hat{\nu}}, Q_t^{\hat{\nu}}, S_t^{\hat{\nu}}, U_t^{\hat{\nu}}; \theta \,c, \theta\, \gamma)\, \biggl(\sum_{i=3}^5 \theta^i \, 	M_i(t,Q_t^{\hat{\nu}},U_t^{\hat{\nu}})\biggr)\,dt\biggr]
\end{multline}
Using the bound \eqref{eqn:pf_H_hat_bound}, we further have
\begin{equation}
    \tfrac{1}{\theta^3}
    \left|
    H^{\hat{\nu}}(0,x,q,S,U;\theta\,c,\theta\,\gamma)
    - \hat{H}(0,x,q,S,U;\theta\,c,\theta\,\gamma)
    \right|
    \leq
    \theta\,\gamma\, C_9\, T\, \mathbb{E}[e^{\theta_0\, \gamma\, C_7\,(1 + M_W + M_Z)}]\,.
    \label{eqn:pf_H_hat_order}
\end{equation}
From Theorem \ref{prop:asymptotic_approximation}, we have
\begin{align*}
    \lim_{\theta\downarrow 0} \tfrac{1}{\theta^3}
    \left|
    H_\psi(0,x,q,S,U;\theta\,c,\theta\,\gamma) - \hat{H}(0,x,q,S,U;\theta\,c,\theta\,\gamma)
    \right|
    &= 0\,.
\end{align*}
Combining the above with \eqref{eqn:pf_H_hat_order} implies
\begin{align}
    \lim_{\theta\downarrow 0} \tfrac{1}{\theta^3}
    \left|
    H_\psi(0,x,q,S,U;\theta\,c,\theta\,\gamma) - H^{\hat{\nu}}(0,x,q,S,U;\theta\,c,\theta\,\gamma)
    \right|
    &= 0\,,
\end{align}
as desired.
}
	 \qed
%\end{proof}

%\begin{proof}[Proof of Proposition \ref{prop:closed_form}]
\subsection{Proof of Proposition \ref{prop:closed_form}}
	{
		The proof proceeds in three parts. (i) We prove the local uniform approximation given by \eqref{eqn:local_uniform}; (ii) we prove the control in \eqref{eqn:admissible_control} is admissible; and (iii) finally we prove the control \eqref{eqn:admissible_control} is approximately optimal to second order in the sense of \eqref{eqn:nu_bar_approx}.
		
\underline{Part (i):} (local uniform approximation): The feedback form of the optimal control when the agent holds $\mathfrak{N}$ units of the non-tradable risk factor is given in closed-form by equation \eqref{eqn:optimal_control}. Denote this function by $\mathfrak{v}^*(t,q,\mathfrak{N};c,\gamma)$. The feedback form of the approximate optimal control when the agent has exposure of the form $\psi(U)$ is in equation \eqref{eqn:approx_optimal_control}. Due to Lemma \ref{lem:future_delta}, the dependence of $\nu_1$ and $\nu_2$ on $U$ in equation \eqref{eqn:approx_optimal_control} appears only through $\partial_U g(t,U)$. Denote the first three terms on the right-hand side of \eqref{eqn:approx_optimal_control}, with $\partial_U g(t,U)$ replaced by $\Delta$, by $\hat{\mathfrak{v}}(t,q,\Delta;c,\gamma)$. Write $\mathfrak{v}^*$ and $\hat{\mathfrak{v}}$ as
\begin{align}
    \mathfrak{v}^*(t,q,\Delta;\theta\,c,\theta\,\gamma)
    &=  \tfrac{1}{2\,k}\,
    \left(\mathfrak{v}_0(t;\theta) + \mathfrak{v}_1(t;\theta)\, q + \mathfrak{v}_2(t;\theta)\, \Delta
    \right)\,,
    \label{eqn:pf_omega}
\\
    \hat{\mathfrak{v}}(t,q,\Delta;\theta\,c,\theta\,\gamma)
    &= \tfrac{1}{2\,k}\,\left(\hat{\mathfrak{v}}_0(t;\theta) + \hat{\mathfrak{v}}_1(t;\theta)\, q + \hat{\mathfrak{v}}_2(t;\theta)\, \Delta\right)\,.
\end{align}
We next show that
\begin{align*}
    \lim_{\theta\downarrow 0} \tfrac{1}{\theta}
    \left(
    \mathfrak{v}_i(t;\theta) - \hat{\mathfrak{v}}_i(t;\theta)
    \right)
    = 0\,,
\end{align*}
uniformly in $t$ for each $i = 0,1,2$. Thus,
\begin{align*}
    \lim_{\theta\downarrow 0} \tfrac{1}{\theta}
    \left(
        \mathfrak{v}^*(t,q,\Delta;\theta\,c,\theta\,\gamma) - \hat{\mathfrak{v}}(t,q,\Delta;\theta\,c,\theta\,\gamma)
    \right)
    = 0\,,
\end{align*}
locally uniformly in $(t,q,\Delta)$.
		
To prove this, we study the $\theta$ dependence of the ODEs satisfied by $\mathfrak{v}_i$ and $\hat{\mathfrak{v}}_i$. The convergence results follow from continuity and differentiability with respect to a parameter of solutions of said ODEs (see for example \cite{chicone2006ordinary} Theorem 1.3).
		
		Inspection of \eqref{eqn:optimal_control}, \eqref{eqn:approx_optimal_control}, and \eqref{eqn:nu012-approx} shows that $\mathfrak{v}_1(t;\theta) = 2\,h_2(t;\theta)+b$ and $\hat{\mathfrak{v}}_1(t;\theta) = 2\,f_2(t) + b + 2\,\theta\,\gamma\Lambda_2(t)$. The functions $h_2$ and $f_2$ both satisfy ODEs of the form
\begin{align*}
    x'= F(x;\theta) \qquad\text{and}\qquad 	x(T) = -\alpha\,,
\end{align*}
where $F(x;\theta) = \tfrac{1}{2}\,\sigma^2\,\theta\,\gamma - \tfrac{1}{4\,k}\,(2\,x+b)^2$ and the ODE for $f_2$ corresponds to $\theta = 0$. When $\theta\downarrow 0$, $F(x;\theta)\rightarrow F(x;0)$ uniformly in $x$, therefore $h_2(t;\theta)\rightarrow f_2(t)$ uniformly in $t\in[0,T]$. This also implies $\mathfrak{v}_1(t;\theta)\rightarrow 2\,f_2(t)+b$ uniformly in $t\in[0,T]$. By L'Hopital's rule we have
\begin{align}
    \lim_{\theta\downarrow 0} \tfrac{1}{\theta}\left(\mathfrak{v}_1(t;\theta) - \hat{\mathfrak{v}}_1(t;\theta)\right) =
    \lim_{\theta\downarrow 0} \left(\partial_\theta \mathfrak{v}_1(t;\theta) - \partial_\theta \hat{\mathfrak{v}}_1(t;\theta)\right)
    %\\
    =
    2 \lim_{\theta\downarrow 0} \left(\partial_\theta h_2(t;\theta) - \gamma\,\Lambda_2(t)\right)
    .
\end{align}
		
Next, from \eqref{eqn:h_2}, $h_2(t;\theta)$ has continuous mixed second order derivatives (wrt $t$ and $\theta$) for $\theta>0$. Thus we write
\begin{align*}
    \partial_t (\partial_\theta h_2) = \partial_\theta (\partial_t h_2) = \tfrac{1}{2}\sigma^2\gamma - \tfrac{1}{k}(2\,h_2+b)\,\partial_\theta h_2 \qquad\text{and}\qquad
    \partial_\theta h_2(T;\theta) = 0\,.
\end{align*}
		We also have
		\begin{align*}
			\partial_t \Lambda_2 = \tfrac{1}{2}\sigma^2 - \tfrac{1}{k}(2\,f_2+b)\,\Lambda_2 \qquad\text{and}\qquad
			\Lambda_2(T) = 0\,,
		\end{align*}
		and because $h_2\rightarrow f_2$ uniformly in $t$ as $\theta\downarrow 0$, we have $\partial_\theta h_2\rightarrow \gamma\Lambda_2$ uniformly in $t$. Thus,
\begin{align*}
    \lim_{\theta\downarrow 0}
    \tfrac{1}{\theta}
    \left(
    \mathfrak{v}_1(t;\theta) - \hat{\mathfrak{v}}_1(t;\theta)
    \right) &= 0\,,
\end{align*}
		uniformly in $t$. Inspection of \eqref{eqn:optimal_control}, \eqref{eqn:ODE_h_system}, and \eqref{eqn:pf_omega} shows that
$\mathfrak{v}_0$ and $\mathfrak{v}_2$ satisfy the ODEs
		\begin{align*}
			\partial_t \mathfrak{v}_0 &= -\mu - \tfrac{1}{2k}(2\, h_2 + b)\, \mathfrak{v}_0\,, & \mathfrak{v}_0(T) &= 0\,,\\
			\partial_t \mathfrak{v}_2 &= \theta\,\gamma\,\rho\,\sigma\,\eta - \tfrac{1}{2k}(2\, h_2 + b)\, \mathfrak{v}_2\,, & \mathfrak{v}_2(T) &= \theta\, c\,.
		\end{align*}
We wish to make the depence on $\theta$ explicit. To this end, inspection of \eqref{eqn:lambda_1}, \eqref{eqn:Lambda_1}, and \eqref{eqn:approx_optimal_control} shows that we may write
		\begin{align*}
			\hat{\mathfrak{v}}_0 + \hat{\mathfrak{v}}_2\,\Delta &= f_1 + \theta\,c\,\Delta + \theta\,c\,\lambda_1 + \theta\,\gamma\,\Lambda_1\\
			&= f_1 + \theta\,c\,\Delta + \theta\,c\,\tilde{\lambda}_1\,\Delta + \theta\,\gamma\,\overline{\Lambda}_1 + \theta\,\gamma\,\tilde{\Lambda}_1\Delta\,,
		\end{align*}
where the introduced functions satisfy the ODEs
		\begin{align*}
			\partial_t f_1 &= -\mu - \tfrac{1}{2k}(2\,f_2+b)\, f_1\,, &  f_1(T) &= 0\,,\\
			\partial_t \tilde{\lambda}_1 &= -\tfrac{1}{2k}(2\,f_2 + b)(1+\tilde{\lambda}_1)\,, & \tilde{\lambda}_1(T) &= 0\,,\\
			\partial_t \overline{\Lambda}_1 &= -\tfrac{1}{k}\Lambda_2\,f_1 - \tfrac{1}{2k}(2\,f_2 + b)\,\overline{\Lambda}_1\,, & \overline{\Lambda}_1(T) &= 0\,,\\
			\partial_t \tilde{\Lambda}_1 &= \rho\,\sigma\,\eta - \tfrac{1}{2k}(2\,f_2 + b)\,\tilde{\Lambda}_1\,, & \tilde{\Lambda}_1(T) &= 0\,.
		\end{align*}
		Thus, we have
		\begin{align*}
			\partial_t\hat{\mathfrak{v}}_0 &= -\mu - \frac{1}{2k}(2\,f_2+b)\, f_1 - \theta\,\gamma\left(\tfrac{1}{k}\Lambda_2\,f_1
+ \tfrac{1}{2k}(2\,f_2+b)
\overline{\Lambda}_1\right)\,, &
 \hat{\mathfrak{v}}_0(T) &= 0\,,
 \\
			\partial_t\hat{\mathfrak{v}}_2 &= \theta \left(\gamma\,\rho\,\sigma\,\eta - \tfrac{1}{2k}(2\,f_2+b)(c + c\,\tilde{\lambda}_1 + \gamma\,\tilde{\Lambda}_1)\right)\,, & \hat{\mathfrak{v}}_2(T) &= \theta\,c\,.\\
		\end{align*}

Analogous to how we prove $\hat{\mathfrak{v}}_1 = \mathfrak{v}_1 + o(\theta)$ above, we may prove the same for $\hat{\mathfrak{v}}_0$ and $\hat{\mathfrak{v}}_2$: First, repeat the arguments to show that the rhs of the associated ODEs converge to appropriate limits, hence $\lim_{\theta\downarrow 0}\mathfrak{v}_i(t;\theta)-\hat{\mathfrak{v}}_i(t;\theta)=0$, next repeat the arguments to show that $\lim_{\theta\downarrow 0}\partial_\theta\mathfrak{v}_i(t;\theta)-\partial_\theta\hat{\mathfrak{v}}_i(t;\theta)=0$. All limits can be taken uniformly in $t\in[0,T]$.
	
\underline{Part (ii)} (admissibility): In feedback form, the candidate trading strategy is
		\begin{align}
			\mathfrak{v}^*(t,q,\partial_U g(t,U)) &= \tfrac{1}{2\,k}\,\left(\mathfrak{v}_0(t;\theta) + \mathfrak{v}_1(t;\theta)\, q + \mathfrak{v}_2(t;\theta)\, \partial_U g(t,U)\right)\,.
		\end{align}
		This is of the same form as the feedback strategy in \eqref{eqn:pf_nu_hat} (the time dependent coefficients are different, but for fixed $\theta$ they are bounded). Thus, the argument for admissibility is the same.
		
\underline{Part (iii)} (optimality approximation): This part of the proof proceeds similarly to Theorem \ref{prop:approx_nu}. Given the candidate strategy $\nu'_t = \mathfrak{v}^*(t,Q_t^{\nu'},\partial_U g(t,U_t^{\nu'});\theta\, c,\theta\, \gamma)$, define the stochastic process $(G_t)_{t\in[0,T]}$ by
\[
    G_t = \hat{H}(t, X_t^{\nu'}, Q_t^{\nu'}, S_t^{\nu'}, U_t^{\nu'}; \theta c, \theta \gamma)\,,
    \quad \text{where}\quad
    \hat{H}(t,x,q,S,U;\theta\, c, \theta\, \gamma) = -e^{-\theta \,\gamma \left(x + q\,S + \hat{h}(t,q,U;\theta\, c, \theta\,\gamma)\right)}\,.
\]
and $\hat{h}$ is the approximation of $h_\psi$ in Theorem \ref{prop:asymptotic_approximation}. Apply Ito's Lemma to $G$ and write
\begin{equation}
\begin{split}
G_T - G_0
=& -\theta\,\gamma\int_0^T \hat{H} (t, X_t^{\nu'}, Q_t^{\nu'}, S_t^{\nu'}, U_t^{\nu'}; \theta\, c, \theta\, \gamma) 
\\
&\qquad\qquad
\times\biggl(\sum_{i=3}^5 \theta^i \, M_i(t,Q_t^{\nu'},U_t^{\nu'}) + V(t,Q_t^{\nu'},U_t^{\nu'};\theta) \biggr)\,dt
\\
&
\qquad -\theta\,\gamma \,\sigma\int_0^T \hat{H} (t, X_t^{\nu'}, Q_t^{\nu'}, S_t^{\nu'}, U_t^{\nu'}; \theta \,c, \theta \,\gamma)\,Q_t^{\nu'}\, dW_t
\\
& \qquad - \theta\,\gamma \,\eta\int_0^T \hat{H} (t, X_t^{\nu'}, Q_t^{\nu'}, S_t^{\nu'}, U_t^{\nu'}; \theta\, c, \theta\, \gamma)\,\partial_U\hat{h}(t,Q_t^{\nu'},U_t^{\nu'}; \theta \,c, \theta \,\gamma)\,dZ_t\,,
\end{split}
\end{equation}
where $M_{3,4,5}$ are given by \eqref{eqn:pf_ito}. The quantity $V$ is shown by explicit computation to be
\begin{align*}
	V(t,q,U;\theta) &= r(t,q,U;\theta)\biggl(\partial_q \hat{h}(t,q,U;\theta)  +\theta\, c\, \partial_U \hat{h}(t,q,U;\theta)\\
					&\hspace{35mm} + b\,q - 2\,k\, \hat{\mathfrak{v}}(t,q,U;\theta)\biggr) - k\, r^2(t,q,U;\theta)\,,
\end{align*}
where $r = \mathfrak{v}^*-\hat{\mathfrak{v}}$. {More details on the computation of $V$ are given in Appendix C.} By construction of $\hat{\mathfrak{v}}$ we have
\begin{align*}
	\partial_q \hat{h}  +\theta\, c\, \partial_U \hat{h}+ b\,q - 2\,k\, \hat{\mathfrak{v}} &= \theta^2\,\biggl( c\,\partial_U h_0 + c^2\,(\partial_q h_3 + \partial_U h_1) + c\,\gamma\,(\partial_q h_4 + \partial_Uh_2) + \gamma^2\,\partial_q h_5 \\
	&\hspace{20mm} + \theta\, c \,(c^2\,\partial_Uh_3 + c\,\gamma\,\partial_Uh_4 + \gamma^2\,\partial_U h_5) \biggr)\,.
\end{align*}
In particular, $V(t,q,U;\theta)$ is a polynomial with respect to $q$ of degree 3 with coefficients that are bounded functions of $t$ and $U$. Furthermore, due to arguments in the first part of this proof we have
\begin{align*}
	\lim_{\theta\downarrow 0} \frac{1}{\theta^2} V(t,q,U;\theta) &= 0\,,
\end{align*}
where the convergence is locally uniform with respect to $q$ and uniform with respect to $t$ and $U$.

All of the estimates from the proof of Theorem \ref{prop:approx_nu} hold identically (except for possibly different constants $C_1,\dots,C_9$). We write
		\begin{align*}
		 	\left|\sum_{i=3}^5 \theta^i \, M_i(t,Q_t^{\nu'},U_t^{\nu'}) + V(t,Q_t^{\nu'},U_t^{\nu'};\theta)\right|  &\leq \theta^3\, C_9 + V(\theta) \,,
		\end{align*}
with $C_9$ as in \eqref{eqn:C9} from Theorem \ref{prop:approx_nu} and where $V$ satisfies
\begin{align*}
	\lim_{\theta\downarrow 0} \frac{1}{\theta^2} V(\theta) &= 0\,,
\end{align*}
(recall that $Q_t^{\nu'}$ is bounded by a constant). We then have
		\begin{align*}
					\left|\mathbb{E}[G_T] - G_0\right| &= \left| \theta\,\mathbb{E}\biggl[\gamma\,\int_0^T \hat{H} (t, X_t^{\nu'}, Q_t^{\nu'}, S_t^{\nu'}, U_t^{\nu'}; \theta\, c, \theta\, \gamma)\right. \\
		&
		\qquad\qquad\qquad
		\left.\times \biggl(\sum_{i=3}^5 \theta^i \, M_i(t,Q_t^{\nu'},U_t^{\nu'}) + V(t,Q_t^{\nu'},U_t^{\nu'};\theta)\biggr)\,dt\biggr]\right|\nonumber\\
				& \leq \gamma\,\theta^3 \,\biggl(\theta \,C_9\,T\, \mathbb{E}[e^{\theta_0 \gamma C_6(1 + M_W + M_Z)}] + \frac{V(\theta)}{\theta^2}\,T \,\mathbb{E}[e^{\theta_0 \gamma C_6(1 + M_W + M_Z)}]\biggr)\,.
			\end{align*}
Therefore,
		\begin{align*}
			\lim_{\theta\downarrow 0}\tfrac{1}{\theta^3}
\left|H^{\nu'}(0,x,q,S,U) - \hat{H}(0,x,q,S,U)\right| &= 0\,,
		\end{align*}
		which, when combined with Theorem \ref{prop:asymptotic_approximation}, proves the required result.
	}
	\qed
%\end{proof}

\section*{Appendix C - $P_{3,4,5,6}$, $M_{3,4,5}$, and $V$}

{
	\subsection{Full Expressions of $P_{3,4,5,6}$}
	
	The following expressions give the functions $P_{3,4,5,6}(t,q,U)$, which appear in the proof of Theorem \ref{prop:asymptotic_approximation}. {These expressions are found by explicitly computing the supremum in \eqref{eqn:LHhatSup} and then grouping powers of $\theta$.}
	
	Recall that each $h_{0,1,2,3,4,5}$ is quadratic with respect to $q$. Then,  by inspection we see that $P_3$ and $P_4$ are third degree polynomials with respect to $q$, $P_5$ and $P_6$ are fourth degree polynomials with respect to $q$, and the coefficients of these polynomials are uniformly bounded functions of $t$ and $U$. In addition, the coefficients are continuously differentiable with respect to $U$ with bounded derivatives by Lemma \ref{lem:future_delta}.
\begin{subequations}\label{eqn:P3456}
\begin{align}
\begin{split}
		P_3 &= \frac{(\gamma\, \partial_q h_2 + c\,(\partial_q h_1 + \partial_U h_0))\, (c^2 \partial_q h_3 + c\,\gamma\, \partial_q h_4 + \gamma^2 \,\partial_q h_5)}{2\,k}\\
			&\hspace{10mm} - \frac{\gamma\,\eta^2((c\,\partial_U h_1 + \gamma\, \partial_U h_2)^2 + 2\,\partial_U h_0 (c^2\,\partial_U h_3 + c\,\gamma\, \partial_U h_4 + \gamma^2\, \partial_U h_5))}{2}\\
			& \hspace{20mm} + \frac{c\,(c^2\,\partial_U h_3 + c\,\gamma\, \partial_U h_4 + \gamma^2\, \partial_U h_5)\,(\partial_q h_0 + b\,q)}{2\,k}\\
			& \hspace{30mm} + \frac{c\,(\gamma\, \partial_q h_2 + c\,(\partial_q h_1 + \partial_U h_0))\,(c\,\partial_U h_1 + \gamma\, \partial_U h_2)}{2\,k}\\
			& \hspace{40mm} - \gamma\,\rho\,\sigma\,\eta\, q\, (c^2\,\partial_U h_3 + c\,\gamma \partial_U h_4 + \gamma^2\, \partial_U h_5)\,,
\end{split}
\\
\begin{split}
		P_4 &= \frac{(c\,(c\,\partial_U h_1 + \gamma\, \partial_U h_2) + c^2\,\partial_q h_3 + c\,\gamma\,\partial_q h_4 + \gamma^2\,\partial_q h_5)^2}{4\,k}\\
			& \hspace{10mm} + \frac{c\,(c^2\,\partial_U h_3 + c\,\gamma\, \partial_U h_4 + \gamma^2\,\partial_U h_5)(c\,\partial_U h_0 + c\,\partial_q h_1 + \gamma\, \partial_q h_2)}{2\,k}\\
			& \hspace{20mm} - \gamma\, \eta^2\, (c\,\partial_U h_1 + \gamma\, \partial_U h_2)(c^2\,\partial_U h_3 + c\,\gamma\, \partial_U h_4 + \gamma^2\,\partial_U h_5),
\end{split}
\\
\begin{split}
		P_5 &= \frac{c^2\,(c^2\,\partial_U h_3 + c\,\gamma\, \partial_U h_4 + \gamma^2\,\partial_U h_5)(c\,\partial_Uh_1 + \gamma \,\partial_U h_2)}{2\,k}\\
			& \hspace{10mm} + \frac{c\,(c^2\,\partial_U h_3 + c\gamma \partial_U h_4 + \gamma^2\partial_U h_5)(c^2\,\partial_q h_3 + c\,\gamma\,\partial_q h_4 + \gamma^2\,\partial_q h_5)}{2\,k}\\
			& \hspace{20mm} - \frac{\gamma\,\eta^2\,(c^2\,\partial_U h_3 + c\,\gamma\, \partial_U h_4 + \gamma^2\,\partial_U h_5)^2}{2},
\end{split}
\\
P_6 &= \frac{c^2\,(c^2\,\partial_U h_3 + c\,\gamma\, \partial_U h_4 + \gamma^2\,\partial_U h_5)^2}{4\,k}\,.
	\end{align}%
\end{subequations}%

\subsection{Full Expressions of $M_{3,4,5}$}
	
	The following expressions give the functions $M_{3,4,5}(t,q,U)$, which appear in the proofs of Theorem \ref{prop:approx_nu} and Proposition \ref{prop:closed_form}. {These expressions are found by substituting the feedback control $\hat{\nu}$ from \eqref{eqn:approx_optimal_control} into $(\partial_t + \mathcal{L}^{\hat{\nu}}) \hat{H} (t, x, q, S, U; \theta\, c, \theta\, \gamma)$ and grouping powers of $\theta$.}
	
	Recall that each $h_{0,1,2,3,4,5}$ is quadratic with respect to $q$. Then,  we see by inspection that $M_3$ and $M_4$ are third degree polynomials with respect to $q$, $M_5$ is a fourth degree polynomial with respect to $q$, and the coefficients of these polynomials are uniformly bounded functions of $t$ and $U$. In addition, the coefficients are continuously differentiable with respect to $U$ with bounded derivatives by Lemma \ref{lem:future_delta}.
\begin{subequations}\label{eqn:M345}
	\begin{align}
	\begin{split}
		M_3 &= \frac{(\gamma\, \partial_q h_2 + c\,(\partial_q h_1 + \partial_U h_0)) (c^2\, \partial_q h_3 + c\,\gamma\, \partial_q h_4 + \gamma^2\, \partial_q h_5)}{2\,k}\\
		&\hspace{10mm} - \frac{\gamma\,\eta^2((c\,\partial_U h_1 + \gamma \,\partial_U h_2)^2 + 2\,\partial_U h_0 (c^2\,\partial_U h_3 + c\,\gamma\, \partial_U h_4 + \gamma^2 \,\partial_U h_5))}{2}\\
		& \hspace{20mm} + \frac{c\,(c^2\,\partial_U h_3 + c\,\gamma\, \partial_U h_4 + \gamma^2 \,\partial_U h_5)(\partial_q h_0 + b\,q)}{2k}\\
		& \hspace{30mm} + \frac{c\,(\gamma\, \partial_q h_2 + c\,(\partial_q h_1 + \partial_U h_0))(c\,\partial_U h_1 + \gamma\, \partial_U h_2)}{2\,k}\\
		& \hspace{40mm} - \gamma\,\rho\,\sigma\,\eta\, q\, (c^2\,\partial_U h_3 + c\,\gamma\, \partial_U h_4 + \gamma^2v \partial_U h_5)\,,
	\end{split}
	\\
	\begin{split}
		M_4 &= \biggl(\frac{c\,(\gamma\, \partial_q h_2 + c\,(\partial_U h_0 + \partial_q h_1)) - 2\,k\,\gamma\,\eta^2\,(c\,\partial_U h_1 + \gamma\, \partial_U h_2)}{2\,k}\biggr)\biggl(c^2\,\partial_U h_3 + c\,\gamma\,\partial_U h_4 + \gamma^2\,\partial_U h_5\biggr)\,,
	\end{split}
	\\
	\begin{split}
		M_5 &= \frac{-\gamma\,\eta^2(c^2\,\partial_U h_3 + c\,\gamma\,\partial_U h_4 + \gamma^2\,\partial_U h_5)^2}{2}\,.
	\end{split}
	\end{align}
\end{subequations}
}

{
\subsection{Computation of $V$}

	Here we show in more detail the steps required to compute $V$, which appears in the proof of Proposition \ref{prop:closed_form}. We begin with 
	\begin{align}
		\begin{split}
			(\partial_t + \mathcal{L}^\nu) \hat{H} (t, x, q, S, U; \theta\, c, \theta\, \gamma) &= -\theta\gamma\hat{H}\biggl(\partial_t \hat{h} + \mu \,q - \tfrac{1}{2}\,\theta\gamma\,\sigma^2\,q^2 + (\beta- \theta\gamma\,\rho\,\sigma\,\eta \,q)\,\partial_U\hat{h}\\
			&\qquad + \tfrac{1}{2}\,\eta^2\,\partial_{UU}\hat{h} - \tfrac{1}{2}\,\theta\gamma\,\eta^2\,(\partial_U\hat{h})^2 + \nu\partial_q\hat{h}+ \theta c\,\nu\,\partial_U\hat{h} + b\,q\,\nu - k\,\nu^2\biggr)\,,
		\end{split}\label{eqn:V}
	\end{align}
	and recall that in feedback form the control $\nu'$ is given by $\mathfrak{v}^*(t,q,\partial_Ug(t,U);\theta\, c,\theta\,\gamma)$. We write this feedback control as
	\begin{align}
		\mathfrak{v}^*(t,q,\partial_Ug(t,U);\theta\, c,\theta\,\gamma) &= \hat{\mathfrak{v}}(t,q,\partial_Ug(t,U);\theta\, c,\theta\,\gamma) + r(t,q,U;\theta)\,,\label{eqn:vstar}
	\end{align}
	where
	\begin{align*}
		r(t,q,U;\theta) &= \mathfrak{v}^*(t,q,\partial_Ug(t,U);\theta\, c,\theta\,\gamma) - \hat{\mathfrak{v}}(t,q,\partial_Ug(t,U);\theta\, c,\theta\,\gamma)\,.
	\end{align*}
	We now substitute \eqref{eqn:vstar} in \eqref{eqn:V} then expand and group terms which contain $r(t,q,U;\theta)$ separate from those which do not. The resulting expression is
	\begin{align*}
		& (\partial_t + \mathcal{L}^{\nu'}) \hat{H} (t, x, q, S, U; \theta\, c, \theta\, \gamma)\\
		& \hspace{10mm}	= -\theta\gamma\hat{H}\biggl(\partial_t \hat{h} + \mu \,q - \tfrac{1}{2}\,\theta\gamma\,\sigma^2\,q^2 + (\beta- \theta\gamma\,\rho\,\sigma\,\eta \,q)\,\partial_U\hat{h} + \tfrac{1}{2}\,\eta^2\,\partial_{UU}\hat{h}\\
		& \hspace{40mm}	- \tfrac{1}{2}\,\theta\gamma\,\eta^2\,(\partial_U\hat{h})^2 + \mathfrak{v}^*\partial_q\hat{h}+ \theta c\,\mathfrak{v}^*\,\partial_U\hat{h} + b\,q\,\mathfrak{v}^* - k\,(\mathfrak{v}^*)^2\biggr)\\
		& \hspace{10mm}	= -\theta\gamma\hat{H}\biggl(\partial_t \hat{h} + \mu \,q - \tfrac{1}{2}\,\theta\gamma\,\sigma^2\,q^2 + (\beta- \theta\gamma\,\rho\,\sigma\,\eta \,q)\,\partial_U\hat{h} + \tfrac{1}{2}\,\eta^2\,\partial_{UU}\hat{h}\\
		& \hspace{25mm}	- \tfrac{1}{2}\,\theta\gamma\,\eta^2\,(\partial_U\hat{h})^2 + \hat{\mathfrak{v}}\partial_q\hat{h}+ \theta c\,\hat{\mathfrak{v}}\,\partial_U\hat{h} + b\,q\,\hat{\mathfrak{v}} - k\,\hat{\mathfrak{v}}^2\\
		& \hspace{40mm}	+ r\,(\partial_q\,\hat{h} + \theta\,c\,\partial_U\hat{h} + b\,q - 2\,k\,\hat{\mathfrak{v}}) - k\,r^2\biggr)\\
		& \hspace{10mm}	= (\partial_t + \mathcal{L}^{\hat{\nu}}) \hat{H} (t, x, q, S, U; \theta\, c, \theta\, \gamma) - \theta\,\gamma\,\hat{H}\biggl(r\,(\partial_q\,\hat{h} + \theta\,c\,\partial_U\hat{h} + b\,q - 2\,k\,\hat{\mathfrak{v}}) - k\,r^2\biggr)\\
		& \hspace{10mm}	= -\theta\,\gamma\,\hat{H}\biggl(\sum_{i=3}^5 \theta^i \, M_i(t,q,U) + r\,(\partial_q\,\hat{h} + \theta\,c\,\partial_U\hat{h} + b\,q - 2\,k\,\hat{\mathfrak{v}}) - k\,r^2\biggr)\,.
	\end{align*}
	The summation in the last line comes from the definitions of the $M_i$'s in the proof of Theorem \ref{prop:approx_nu}, also outlined earlier in this appendix. The remaining terms in large parentheses are denoted by $V(t,q,U;\theta)$.
}

\bibliographystyle{chicago}
\bibliography{References}

\end{document}